\documentclass[acmsmall]{acmart}

\usepackage{geometry}

\usepackage{amsmath}
\usepackage{amsthm}
\usepackage{hyperref}
\usepackage{url}
\usepackage{multicol}
\usepackage{verbatim}
\usepackage{listings}
\usepackage{booktabs} 
\usepackage{andymacros}
\usepackage[ruled]{algorithm2e} 
\usepackage{enumitem}
\usepackage{wrapfig}
\usepackage{subcaption}
\usepackage[normalem]{ulem}

\usepackage{aliascnt}

\toggletrue{LONG}

\togglefalse{COMMENTS}

\acmJournal{PACMPL}
\acmVolume{1}
\acmNumber{POPL} 
\acmArticle{}
\acmYear{2019}
\acmMonth{1}
\acmDOI{}
\startPage{1}

\authorsaddresses{} 

\setcopyright{rightsretained}

\usepackage[T1]{fontenc}
\usepackage{microtype}

\newcommand{\numbcolor}{darkgray}
\newcommand{\numbdist}{5pt}

\newcommand{\reduline}{\bgroup\markoverwith
	{\textcolor{midred}{\rule[0.4ex]{6pt}{1.2pt}}}\ULon}
\newcommand{\stkw}[1]{\sffamily\bfseries\color{dkblue}\reduline{#1}}
\newcommand{\stlst}[1]{\sffamily\reduline{#1}}

\newcommand{\decls}[1]{\Gamma_{#1}}

\newcommand{\cmds}[1]{\mathrm{S}_{#1}}

\iftoggle{SHORT}{ \newlabel{sec:stan}{{2}{3}{Core Stan}{section.2}{}}
\newlabel{ssec:stanstatements}{{2.1}{3}{Syntax of Core Stan Expressions and Statements}{subsection.2.1}{}}
\newlabel{ssec:stanop}{{2.2}{4}{Operational Semantics of Stan Statements}{subsection.2.2}{}}
\newlabel{Eval Const}{{(\textsc  {Eval Const})}{4}{Operational Semantics of Stan Statements}{rule.1}{}}
\newlabel{Eval Var}{{(\textsc  {Eval Var})}{4}{Operational Semantics of Stan Statements}{rule.2}{}}
\newlabel{Eval Arr}{{(\textsc  {Eval Arr})}{4}{Operational Semantics of Stan Statements}{rule.3}{}}
\newlabel{Eval ArrEl}{{(\textsc  {Eval ArrEl})}{4}{Operational Semantics of Stan Statements}{rule.4}{}}
\newlabel{Eval PrimCall}{{(\textsc  {Eval PrimCall})}{4}{Operational Semantics of Stan Statements}{rule.5}{}}
\newlabel{Eval Assign}{{(\textsc  {Eval Assign})}{4}{Operational Semantics of Stan Statements}{rule.6}{}}
\newlabel{Eval Seq}{{(\textsc  {Eval Seq})}{4}{Operational Semantics of Stan Statements}{rule.7}{}}
\newlabel{Eval IfTrue}{{(\textsc  {Eval IfTrue})}{4}{Operational Semantics of Stan Statements}{rule.8}{}}
\newlabel{Eval IfFalse}{{(\textsc  {Eval IfFalse})}{4}{Operational Semantics of Stan Statements}{rule.9}{}}
\newlabel{Eval ForTrue}{{(\textsc  {Eval ForTrue})}{5}{Operational Semantics of Stan Statements}{rule.10}{}}
\newlabel{Eval ForFalse}{{(\textsc  {Eval ForFalse})}{5}{Operational Semantics of Stan Statements}{rule.11}{}}
\newlabel{Eval Skip}{{(\textsc  {Eval Skip})}{5}{Operational Semantics of Stan Statements}{rule.12}{}}
\newlabel{ssec:stansyntax}{{2.3}{5}{Syntax of Stan}{subsection.2.3}{}}
\newlabel{Stan State}{{(\textsc  {Stan State})}{5}{Syntax of Stan}{rule.13}{}}
\newlabel{ssec:stansem}{{2.4}{6}{Density-based Semantics of Stan}{subsection.2.4}{}}
\newlabel{sec:slicstan}{{3}{7}{SlicStan}{section.3}{}}
\newlabel{tab:blocks}{{1}{7}{Program blocks in Stan. Adapted from \citet {StanTutorial}.\relax }{table.caption.2}{}}
\newlabel{ssec:slicsyntax}{{3.1}{8}{Syntax}{subsection.3.1}{}}
\newlabel{ssec:slictyping}{{3.2}{9}{Typing of SlicStan}{subsection.3.2}{}}
\newlabel{ESub}{{(\textsc  {ESub})}{10}{Typing of SlicStan}{rule.14}{}}
\newlabel{Var}{{(\textsc  {Var})}{10}{Typing of SlicStan}{rule.15}{}}
\newlabel{Const}{{(\textsc  {Const})}{10}{Typing of SlicStan}{rule.16}{}}
\newlabel{Arr}{{(\textsc  {Arr})}{10}{Typing of SlicStan}{rule.17}{}}
\newlabel{ArrEl}{{(\textsc  {ArrEl})}{10}{Typing of SlicStan}{rule.18}{}}
\newlabel{PrimCall}{{(\textsc  {PrimCall})}{10}{Typing of SlicStan}{rule.19}{}}
\newlabel{FCall}{{(\textsc  {FCall})}{10}{Typing of SlicStan}{rule.20}{}}
\newlabel{SSub}{{(\textsc  {SSub})}{10}{Typing of SlicStan}{rule.21}{}}
\newlabel{Assign}{{(\textsc  {Assign})}{10}{Typing of SlicStan}{rule.22}{}}
\newlabel{Decl}{{(\textsc  {Decl})}{10}{Typing of SlicStan}{rule.23}{}}
\newlabel{If}{{(\textsc  {If})}{10}{Typing of SlicStan}{rule.24}{}}
\newlabel{Seq}{{(\textsc  {Seq})}{10}{Typing of SlicStan}{rule.25}{}}
\newlabel{Skip}{{(\textsc  {Skip})}{10}{Typing of SlicStan}{rule.26}{}}
\newlabel{For}{{(\textsc  {For})}{10}{Typing of SlicStan}{rule.27}{}}
\newlabel{DataDecl}{{(\textsc  {DataDecl})}{10}{Typing of SlicStan}{rule.28}{}}
\newlabel{PrimModel}{{(\textsc  {PrimModel})}{10}{Typing of SlicStan}{rule.29}{}}
\newlabel{Void}{{(\textsc  {Void})}{10}{Typing of SlicStan}{rule.30}{}}
\newlabel{Return}{{(\textsc  {Return})}{10}{Typing of SlicStan}{rule.31}{}}
\newlabel{FModel}{{(\textsc  {FModel})}{10}{Typing of SlicStan}{rule.32}{}}
\newlabel{FDef}{{(\textsc  {FDef})}{11}{Typing of SlicStan}{rule.33}{}}
\newlabel{State}{{(\textsc  {State})}{11}{Typing of SlicStan}{rule.34}{}}
\newlabel{th:eval}{{3.1}{11}{Type Preservation for $\Downarrow $}{theorem.3.1}{}}
\newlabel{th:noninterf}{{3.3}{11}{Noninterference}{theorem.3.3}{}}
\newlabel{ssec:slicelab}{{3.3}{12}{Elaboration of SlicStan}{subsection.3.3}{}}
\newlabel{Elab SlicStan}{{(\textsc  {Elab SlicStan})}{12}{Elaboration of SlicStan}{rule.35}{}}
\newlabel{Elab Var}{{(\textsc  {Elab Var})}{12}{Elaboration of SlicStan}{rule.36}{}}
\newlabel{Elab Const}{{(\textsc  {Elab Const})}{12}{Elaboration of SlicStan}{rule.37}{}}
\newlabel{Elab ArrEl}{{(\textsc  {Elab ArrEl})}{12}{Elaboration of SlicStan}{rule.38}{}}
\newlabel{Elab Arr}{{(\textsc  {Elab Arr})}{12}{Elaboration of SlicStan}{rule.39}{}}
\newlabel{Elab PrimCall}{{(\textsc  {Elab PrimCall})}{12}{Elaboration of SlicStan}{rule.40}{}}
\newlabel{Elab FCall}{{(\textsc  {Elab FCall})}{12}{Elaboration of SlicStan}{rule.41}{}}
\newlabel{Elab FDef}{{(\textsc  {Elab FDef})}{12}{Elaboration of SlicStan}{rule.42}{}}
\newlabel{Elab Decl}{{(\textsc  {Elab Decl})}{13}{Elaboration of SlicStan}{rule.43}{}}
\newlabel{Elab Skip}{{(\textsc  {Elab Skip})}{13}{Elaboration of SlicStan}{rule.44}{}}
\newlabel{Elab Assign}{{(\textsc  {Elab Assign})}{13}{Elaboration of SlicStan}{rule.45}{}}
\newlabel{Elab Seq}{{(\textsc  {Elab Seq})}{13}{Elaboration of SlicStan}{rule.46}{}}
\newlabel{Elab If}{{(\textsc  {Elab If})}{13}{Elaboration of SlicStan}{rule.47}{}}
\newlabel{Elab For}{{(\textsc  {Elab For})}{13}{Elaboration of SlicStan}{rule.48}{}}
\newlabel{th:elab}{{3.4}{13}{Type preservation of $\elab $}{theorem.3.4}{}}
\newlabel{ssec:slicsem}{{3.4}{13}{Semantics of SlicStan}{subsection.3.4}{}}
\newlabel{ssec:sem_statement}{{3.4.1}{14}{Semantics of (elaborated) SlicStan statements}{subsubsection.3.4.1}{}}
\newlabel{ssec:sem_prog_density}{{3.4.2}{14}{Semantics of SlicStan programs}{subsubsection.3.4.2}{}}
\newlabel{eq:slicposterior}{{1}{14}{Semantics of SlicStan programs}{equation.3.1}{}}
\newlabel{ssec:slicexamples}{{3.5}{14}{Examples}{subsection.3.5}{}}
\newlabel{sssec:udex}{{3.5.2}{15}{User-defined Functions Example}{subsubsection.3.5.2}{}}
\newlabel{ssec:difficulty}{{3.6}{15}{Difficulty of Specifying Direct Semantics}{subsection.3.6}{}}
\newlabel{sec:translate-slicstan-to-stan}{{4}{15}{Translation of SlicStan to Stan}{section.4}{}}
\newlabel{ssec:shred}{{4.1}{15}{Shredding}{subsection.4.1}{}}
\newlabel{Shred DataAssign}{{(\textsc  {Shred DataAssign})}{16}{Shredding}{rule.49}{}}
\newlabel{Shred ModelAssign}{{(\textsc  {Shred ModelAssign})}{16}{Shredding}{rule.50}{}}
\newlabel{Shred GenQuantAssign}{{(\textsc  {Shred GenQuantAssign})}{16}{Shredding}{rule.51}{}}
\newlabel{Shred Seq}{{(\textsc  {Shred Seq})}{16}{Shredding}{rule.52}{}}
\newlabel{Shred Skip}{{(\textsc  {Shred Skip})}{16}{Shredding}{rule.53}{}}
\newlabel{Shred If}{{(\textsc  {Shred If})}{16}{Shredding}{rule.54}{}}
\newlabel{Shred For}{{(\textsc  {Shred For})}{16}{Shredding}{rule.55}{}}
\newlabel{def:singlelev}{{4.1}{17}{Single-level Statement $\Gamma \vdash \ell (S)$}{theorem.4.1}{}}
\newlabel{lem:shredisleveled}{{4.2}{17}{Shredding produces single-level statements}{theorem.4.2}{}}
\newlabel{def:equiv}{{4.3}{17}{Statement equivalence}{theorem.4.3}{}}
\newlabel{lem:reorder}{{4.4}{17}{Statement Reordering}{theorem.4.4}{}}
\newlabel{lem:halfway}{{4.5}{17}{}{theorem.4.5}{}}
\newlabel{lem:same_sets_when_singlelevel}{{4.6}{17}{}{theorem.4.6}{}}
\newlabel{def:l}{{4.7}{17}{Well-typed if on-higher-level-read-immutable $\vdash _{l}$}{theorem.4.7}{}}
\newlabel{SeqLevelFix}{{(\textsc  {SeqLevelFix})}{17}{Well-typed if on-higher-level-read-immutable $\vdash _{l}$}{rule.56}{}}
\newlabel{lem:commutativity}{{4.8}{17}{Commutativity of sequencing single-level statements}{theorem.4.8}{}}
\newlabel{th:shred}{{4.9}{18}{Semantic Preservation of $\shred $}{theorem.4.9}{}}
\newlabel{ssec:trans}{{4.2}{18}{Transformation}{subsection.4.2}{}}
\newlabel{Trans Prog}{{(\textsc  {Trans Prog})}{18}{Transformation}{rule.57}{}}
\newlabel{Trans Data}{{(\textsc  {Trans Data})}{18}{Transformation}{rule.58}{}}
\newlabel{Trans TrData}{{(\textsc  {Trans TrData})}{18}{Transformation}{rule.59}{}}
\newlabel{Trans Param}{{(\textsc  {Trans Param})}{18}{Transformation}{rule.60}{}}
\newlabel{Trans TrParam}{{(\textsc  {Trans TrParam})}{18}{Transformation}{rule.61}{}}
\newlabel{Trans GenQuant}{{(\textsc  {Trans GenQuant})}{18}{Transformation}{rule.62}{}}
\newlabel{Trans Empty}{{(\textsc  {Trans Empty})}{18}{Transformation}{rule.63}{}}
\newlabel{Trans Data}{{(\textsc  {Trans Data})}{19}{Transformation}{rule.64}{}}
\newlabel{Trans GenQuant}{{(\textsc  {Trans GenQuant})}{19}{Transformation}{rule.65}{}}
\newlabel{Trans ParamAssign}{{(\textsc  {Trans ParamAssign})}{19}{Transformation}{rule.66}{}}
\newlabel{Trans Model}{{(\textsc  {Trans Model})}{19}{Transformation}{rule.67}{}}
\newlabel{Trans ParamSeq}{{(\textsc  {Trans ParamSeq})}{19}{Transformation}{rule.68}{}}
\newlabel{Trans ParamIf}{{(\textsc  {Trans ParamIf})}{19}{Transformation}{rule.69}{}}
\newlabel{Trans ModelIf}{{(\textsc  {Trans ModelIf})}{19}{Transformation}{rule.70}{}}
\newlabel{Trans ParamFor}{{(\textsc  {Trans ParamFor})}{19}{Transformation}{rule.71}{}}
\newlabel{Trans ParamSkip}{{(\textsc  {Trans ParamSkip})}{19}{Transformation}{rule.72}{}}
\newlabel{Trans ModelFor}{{(\textsc  {Trans ModelFor})}{19}{Transformation}{rule.73}{}}
\newlabel{th:trans}{{4.10}{19}{Semantic Preservation of $\tostan $}{theorem.4.10}{}}
\newlabel{sec:demo}{{5}{20}{Examples and Discussion}{section.5}{}}
\newlabel{ssec:typeinference}{{5.1}{20}{Type Inference}{subsection.5.1}{}}
\newlabel{ssec:expert}{{5.2}{20}{Locality}{subsection.5.2}{}}
\newlabel{shred2}{{5.2}{21}{Locality}{subsection.5.2}{}}
\newlabel{ssec:refactor}{{5.3}{21}{Code Refactoring}{subsection.5.3}{}}
\newlabel{refactor}{{5.3}{22}{Code Refactoring}{subsection.5.3}{}}
\newlabel{ssec:encaps}{{5.4}{22}{Code Reuse}{subsection.5.4}{}}
\newlabel{schools}{{5.4}{23}{Code Reuse}{subsection.5.4}{}}
\newlabel{sec:related}{{6}{24}{Related Work}{section.6}{}}
\newlabel{ssec:relatedsem}{{6.1}{24}{Formalisation of Probabilistic Programming Languages}{subsection.6.1}{}}
\newlabel{ssec:relatedstatic}{{6.2}{24}{Static Analysis for Probabilistic Programming Languages}{subsection.6.2}{}}
\newlabel{ssec:relatedusab}{{6.3}{25}{Usability of Probabilistic Programming Languages}{subsection.6.3}{}}
\newlabel{ssec:usab}{{6.3}{25}{Usability of Probabilistic Programming Languages}{subsection.6.3}{}}
\newlabel{sec:conc}{{7}{25}{Conclusion}{section.7}{}}
\newlabel{ap:proofs}{{A}{29}{Definitions and Proofs}{appendix.A}{}}
\newlabel{def:assset}{{A.1}{29}{Assigns-to set $\assset (S)$}{theorem.A.1}{}}
\newlabel{def:readset}{{A.2}{29}{Reads set $\readset (S)$}{theorem.A.2}{}}
\newlabel{def:lvaluetypes}{{A.3}{29}{Assigns-to set $\Gamma (L)$}{theorem.A.3}{}}
\newlabel{def:vector}{{A.4}{29}{Vectorising functions $v_{\Gamma }, v_E, v_S$}{theorem.A.4}{}}
\newlabel{def:read_level_set}{{A.5}{29}{$\readset _{\Gamma \vdash \ell }(S)$}{theorem.A.5}{}}
\newlabel{def:write_level_set}{{A.6}{30}{$\assset _{\Gamma \vdash \ell }(S)$}{theorem.A.6}{}}
\newlabel{def:stanmerge}{{A.7}{30}{}{theorem.A.7}{}}
\newlabel{ap:proofshred}{{A.2}{30}{Proof of Semantic Preservation of Shredding}{subsection.A.2}{}}
\newlabel{lemma:dom}{{A.8}{30}{}{theorem.A.8}{}}
\newlabel{lemma:assset}{{A.9}{30}{}{theorem.A.9}{}}
\newlabel{lem:same_effect}{{A.10}{30}{}{theorem.A.10}{}}
\newlabel{eq:1}{{2}{30}{Proof of Semantic Preservation of Shredding}{equation.A.2}{}}
\newlabel{eq:2}{{3}{31}{Proof of Semantic Preservation of Shredding}{equation.A.3}{}}
\newlabel{eq:4}{{4}{31}{Proof of Semantic Preservation of Shredding}{equation.A.4}{}}
\newlabel{eq:5}{{5}{31}{Proof of Semantic Preservation of Shredding}{equation.A.5}{}}
\newlabel{eq:6}{{6}{31}{Proof of Semantic Preservation of Shredding}{equation.A.6}{}}
\newlabel{eq:7}{{7}{31}{Proof of Semantic Preservation of Shredding}{equation.A.7}{}}
\newlabel{lem:vars_in_exps}{{A.11}{31}{}{theorem.A.11}{}}
\newlabel{lem:exp_in_singleleveled_s}{{A.12}{31}{}{theorem.A.12}{}}
\newlabel{lem:assoc}{{A.13}{31}{Associativity of $\eveq $}{theorem.A.13}{}}
\newlabel{lem:congr}{{A.14}{31}{Congruence of $\eveq $}{theorem.A.14}{}}
\newlabel{lem:forloops}{{A.15}{31}{}{theorem.A.15}{}}
\newlabel{lem:moreassignments}{{A.16}{32}{}{theorem.A.16}{}}
\newlabel{lem:loopreorder}{{A.17}{32}{}{theorem.A.17}{}}
\newlabel{ap:gq_semantics}{{B.1}{35}{Semantics of Generated Quantities}{subsection.B.1}{}}
\newlabel{ap:sampling_semantics}{{B.2}{35}{Relation of Density-based Semantics to Sampling-based Semantics}{subsection.B.2}{}}
\newlabel{Eval Model}{{(\textsc  {Eval Model})}{36}{Difference in the meaning of $\sim $}{rule.79}{}}
\newlabel{Sampling Model}{{(\textsc  {Sampling Model})}{36}{Difference in the meaning of $\sim $}{rule.80}{}}
\newlabel{ap:shred}{{C}{37}{Elaborating and shredding if or for statements}{appendix.C}{}}
\newlabel{ap:ncp}{{D}{38}{Non-centred Reparameterisation}{appendix.D}{}}
\newlabel{ap:examples}{{E}{39}{Examples}{appendix.E}{}}
\newlabel{fig:nealsineff}{{1a}{40}{$24,000$ samples obtained using Stan (default settings) for the \emph {non-efficient} form of Neal's Funnel.\relax }{figure.caption.5}{}}
\newlabel{sub@fig:nealsineff}{{a}{40}{$24,000$ samples obtained using Stan (default settings) for the \emph {non-efficient} form of Neal's Funnel.\relax }{figure.caption.5}{}}
\newlabel{fig:nealseff}{{1b}{40}{$24,000$ samples obtained using Stan (default settings) for the \emph {efficient} form of Neal's Funnel.\relax }{figure.caption.5}{}}
\newlabel{sub@fig:nealseff}{{b}{40}{$24,000$ samples obtained using Stan (default settings) for the \emph {efficient} form of Neal's Funnel.\relax }{figure.caption.5}{}}
\newlabel{fig:neals}{{1c}{40}{The actual log density of Neal's Funnel. Dark regions are of high density (log density greater than $-8$).\relax }{figure.caption.5}{}}
\newlabel{sub@fig:neals}{{c}{40}{The actual log density of Neal's Funnel. Dark regions are of high density (log density greater than $-8$).\relax }{figure.caption.5}{}}
\newlabel{ssec:roaches}{{E.2}{43}{Cockroaches}{subsection.E.2}{}}
\newlabel{roaches}{{E.2}{44}{Cockroaches}{subsection.E.2}{}}
\newlabel{note2}{{9}{44}{}{Hfootnote.9}{}}
\newlabel{tocindent-1}{0pt}
\newlabel{tocindent0}{0pt}
\newlabel{tocindent1}{7.01pt}
\newlabel{tocindent2}{13.79999pt}
\newlabel{tocindent3}{21.7pt}
\newlabel{seeds}{{E.3}{46}{Seeds}{subsection.E.3}{}}
\newlabel{note1}{{10}{46}{}{Hfootnote.10}{}}
\newlabel{TotPages}{{46}{46}{}{page.46}{}}
 } { }
\iftoggle{SUPPLEMENTARY}{  } { }

\begin{document}

\title[SlicStan]{Probabilistic Programming with Densities in SlicStan: Efficient, Flexible and Deterministic}

\author{Maria I. Gorinova}
\affiliation{
	\institution{University of Edinburgh}}
\email{m.gorinova@ed.ac.uk}
\author{Andrew D. Gordon}
\affiliation{%
	\institution{Microsoft Research Cambridge}}
\affiliation{%
	\institution{University of Edinburgh}}
\email{adg@microsoft.com}
\author{Charles Sutton}
\affiliation{%
	\institution{Google Brain}}
\affiliation{%
	\institution{University of Edinburgh}}
\email{cas@inf.ed.ac.uk}

\begin{abstract}
	Stan is a probabilistic programming language that has been increasingly used for real-world scalable projects. However, to make practical inference possible, the language sacrifices some of its usability by adopting a block syntax, which lacks compositionality and flexible user-defined functions. Moreover, the semantics of the language has been mainly given in terms of intuition about implementation, and has not been formalised.
	
	This paper provides a formal treatment of the Stan language, and introduces the probabilistic programming language SlicStan --- a compositional, self-optimising version of Stan. Our main contributions are 
	(1) the formalisation of a core subset of Stan through an operational density-based semantics; 
	(2) the design and semantics of the Stan-like language SlicStan, which facilities better code reuse and abstraction through its compositional syntax, more flexible functions, and information-flow type system; and 
	(3) a formal, semantic-preserving procedure for translating SlicStan to Stan. 

\end{abstract}


\maketitle

\section{Introduction}

\subsection{Background: Probabilistic Programming Languages and Stan}
Probabilistic programming languages \cite{GordonPP} are a concise notation for specifying probabilistic models, while abstracting the underlying inference algorithm.
There are many such languages, including BUGS \cite{BUGS}, JAGS \cite{JAGS}, Anglican \cite{Anglican}, Church \cite{Church}, Infer.NET \cite{InferNET}, Venture \cite{Venture}, Edward \cite{Edward} and many others.

Stan \cite{StanJSS}, with nearly 300,000 downloads of its R interface \cite{RStan}, is perhaps the most widely used probabilistic programming language.
Stan's syntax is designed to enable automatic compilation to an efficient \emph{Hamiltonian Monte Carlo} (HMC) inference algorithm \cite{HMCfirst}, which allows programs to scale to real-word projects in statistics and data science.
(For example, the forecasting tool Prophet \cite{Prophet} uses Stan.)
This efficiency comes at a price: Stan's syntax lacks the compositionality of other similar languages and systems, such as Edward \cite{Edward} and PyMC3 \cite{PyMC3}.
The design of Stan assumes that the programmer needs to organise their model into separate blocks, which correspond to different stages of the inference algorithm (preprocessing, sampling, postprocessing). 
This compartmentalised syntax affects the usability of Stan: related statements may be separated in the source code,
and functions are restricted to only acting within a single compartment.
It is difficult to write complex Stan programs and encapsulate distributions and sub-model structures into re-usable libraries and routines.

\subsection{Goals and Key Insight}
Our goals are (1) to examine the principles and assumptions behind the probabilistic programming language Stan that help it bridge the gap between probabilistic modelling and black-box inference; and (2) to design a suitable abstraction that captures the statistical meaning of Stan's compartments, but allows for compositional and more flexible probabilistic programming language syntax. 

Our key insight is that \emph{the essence of a probabilistic program in Stan is in fact a deterministic imperative program}, that can be automatically sliced into the different compartments used in the current syntax for Stan.
It may come as a surprise that a probabilistic program is deterministic, but when performing Bayesian inference by sampling parameters, the probabilistic program serves to compute a deterministic score at a specific point of the parameter space.
An implication of this insight is that standard forms of procedural abstraction are easily adapted to Stan.

\subsection{The Insight by Example}

As demonstration, and as a candidate for a future re-design of Stan, we present SlicStan\footnotemark --- a compositional, Stan-like language, which supports first-order functions. 
\footnotetext{SlicStan (pronounced \emph{slick-Stan}) stands for ``Slightly Less Intensely Constrained Stan''.}

Below, we show an example of a Stan program (right), and the same program written in SlicStan (left). In both cases, the goal is to obtain samples from the joint distribution of the variables $y \sim \normal(0, 3)$ and $x \sim \normal(0, \exp(y/2))$, using auxiliary standard normal variables for performance. Working with such auxiliary variables, instead of defining the model in terms of $x$ and $y$ directly, can facilitate inference and is a standard technique. We give more details in \autoref{ssec:encaps} and Appendix~\ref{ap:ncp}.

\vspace{-8pt}
\begin{multicols}{2}
	\centering
	\textbf{SlicStan}
	\begin{lstlisting}
	real my_normal(real m, real s) {
		real std ~ normal(0, 1); 
		return s * std + m;
	}
	real y = my_normal(0, 3); 
	real x = my_normal(0, exp(y/2));
	\end{lstlisting} 
	\vspace{2.5cm}
	\textcolor{white}{.}\\
	\textbf{Stan}
	\begin{lstlisting}
	parameters {
		real y_std;
		real x_std;
	}
	transformed parameters {
		real y = 3 * y_std;
		real x = exp(y/2) * x_std;
	}
	model {
		y_std ~ normal(0, 1); 
		x_std ~ normal(0, 1); 
	}
	\end{lstlisting}
\end{multicols}
\vspace{-15pt}

In both programs, the aim is to obtain samples for the random variables $x$ and $y$, which are defined by scaling and shifting standard normal variables. 

In SlicStan we do so by calling the function \lstinline|my_normal| twice, which defines a local parameter \lstinline|std| and encapsulates the transformation of each variable.
Stan, on the other hand, does not support functions that declare new parameters, because all parameters must be declared inside the \lstinline|parameters| block. We need to write out each transformation explicitly, also explicitly declaring each auxiliary parameter (\lstinline|x_std| and \lstinline|y_std|).

The SlicStan code is a conventional deterministic imperative program, where
the model statement \lstinline|std ~ normal(0,1)| is a derived form of an assignment to a reserved variable that holds the score at a particular point of the parameter space.
Due to the absence of blocks, SlicStan's syntax is compositional and more compact. Statements that would belong to different blocks of a Stan program can be interleaved, and no understanding about the performance implications of different compartments is required.
Via an information-flow analysis, we automatically translate the program on the left to the one on the right. 

\subsection{Core Contributions and Outline}
In short, this paper makes the following contributions:
\begin{itemize}
	\item We formalise the syntax and semantics of a core subset of Stan (\autoref{sec:stan}). To the best of our knowledge, this is the first formal treatment of Stan, despite the popularity of the language.
	
	\item We design SlicStan --- a compositional Stan-like language with first-order functions. We formalise an information-flow type system that captures the essence of Stan's compartmentalised syntax, give the formal semantics of SlicStan, and prove standard results for the calculus, including noninterference and type-preservation properties (\autoref{sec:slicstan}).
	
	\item We give a formal procedure for translating SlicStan to Stan, and prove that it is semantics preserving (\autoref{sec:translate-slicstan-to-stan}).
	
	\item We examine the usability of SlicStan compared to that of Stan, using examples (\autoref{sec:demo}).
\end{itemize}

\iftoggle{LONG}{
This paper also includes an appendix, which provides additional details, discussion and examples.
} { }

\iftoggle{SHORT}{
The extended version of this paper \cite{SlicStanArxiv} includes an appendix (referred to as Appendix here) with additional details, discussion and examples.
} { }

\section{Core Stan}\label{sec:stan}
Stan \cite{StanJSS} is a probabilistic programming language, whose syntax is similar to that of BUGS \cite{BUGS,BUGSBook}, and aims to be close to the model specification conventions used in the statistics community. This section gives the syntax (\autoref{ssec:stanstatements} and \autoref{ssec:stansyntax}) and semantics (\autoref{ssec:stanop} and \autoref{ssec:stansem}) of Core Stan, a core subset of the Stan language. The subset omits, for example, constraint data types, \lstinline|while| loops, random number generators, recursive user defined functions, and local variables. 

To the best of our knowledge, our work is the first to give a formal semantics to the core of Stan.
A full descriptive language specification can be found in the official reference manual \cite{StanManual}. 

\subsection{Syntax of Core Stan Expressions and Statements} \label{ssec:stanstatements}
The building blocks of a Stan statement are expressions.
In Core Stan, expressions cover most of what the Stan manual specifies, including variables, constants, arrays and array elements, and function calls (of builtin functions).
We let $x$ range over variables. 
L-values are expressions limited to array elements $x[E_1]\dots[E_n]$, where the case $n=0$ corresponds to a variable $x$. 
Statements cover the core functionality of Stan, with the exception of \lstinline|while| statements, which we omit to to make \emph{shredding} of SlicStan possible (see \autoref{ssec:slicsyntax} and \autoref{ssec:shred}). \vspace{-8pt}
\begin{multicols}{2}
\begin{display}[.35]{Core Stan Syntax of Expressions:}
	\Category{E}{expression}\\
	\entry{x}{variable}\\
	\entry{c}{constant}\\
	\entry{[E_1,...,E_n]}{array}\\
	\entry{E_1[E_2]}{array element}\\ 
	\entry{f(E_1,\dots,E_n)}{function call\footnotemark}\\
	\clause{L ::= x[E_1]\dots[E_n] \quad n \geq 0}{L-value}
\end{display}
\begin{display}[.35]{Core Stan Syntax of Statements:}
	\Category{S}{statement}\\
	\entry{L = E}{assignment}\\
	\entry{S_1; S_2}{sequence}\\
	\entry{\kw{for}(x\;\kw{in}\;E_1:E_2)\;S}{for loop} \\
	\entry{\kw{if}(E)\;S_1\,\kw{else}\;S_2}{if statement} \\
	\entry{\kw{skip}}{skip}\\
\end{display}
\end{multicols} \vspace{-8pt}
\footnotetext{If $f$ is a binary operator, e.g.\,``$+$'', we write it in infix.} 

We assume a set of builtin functions, ranged over by $f$.
We also assume a set of standard builtin continuous or discrete distributions, ranged over by \lstinline|d|.
Each continuous distribution \lstinline|d| has a corresponding builtin function \lstinline|d_lpdf|, which defines its log probability density function. In this paper, we omit discrete random variables for simplicity.

Defined like this, the syntax of Stan statements is one of a standard imperative language. What makes the language \emph{probabilistic} is the reserved variable \lstinline|target|, which holds the logarithm\footnotemark of the probability density function defined by the program (up to an additive constant), 
evaluated at the point specified by the current values of the program variables.

For example, to define a Stan model with random variables, $\mu$ and $x$, where we assume the variables are normally distributed and $\mu \sim \normal(0, 1)$ and $x \sim \normal(\mu, 1)$, we write:
\footnotetext{Stan evaluates the unnormalised density in the $\log$ domain to ensure numerical stability and to simplify internal computations. We follow this style throughout the paper, and define the semantics in terms of $\log p^*$, instead of $p^*$.}
\vspace{-2pt}\begin{lstlisting}
target = normal_lpdf(mu, 0, 1) + normal_lpdf(x, mu, 1);$\footnotemark$
\end{lstlisting}\vspace{-4pt}
\footnotetext{We treat \lstinline|target| as a mutable program variable for simplicity. This slightly differs from the actual implementation of Stan, where \lstinline|target| does not allow for general lookup and update, but it is a special bit of state that can only be incremented.}
Here, \lstinline|normal_lpdf| is the log density of the normal distribution: $\log \normal(x|\,\mu, \sigma) = -\frac{(x-\mu)^2}{2\sigma^2} -\frac{1}{2}\log 2\pi\sigma^2$.
The value of \kw{target} is equal to the logarithm of the joint density over $\mu$ and $x$, $\log p(\mu, x)$, evaluated at the current values of the program variables \kw{mu} and \kw{x}. 
Suppose $x$ is some known \textit{data}, and $\mu$ is an unknown \textit{parameter} of the model. We are interested in computing the \textit{posterior distribution} of $\mu$ given $x$, $p(\mu \mid x) \propto p(\mu, x) = \normal(x \mid \mu, 1)\normal(\mu \mid 0, 1)$. Stan directly encodes a function that calculates the value of the log posterior density (up to an additive constant), and stores it in \kw{target}. 
Thus, in addition to Stan's core statement syntax, we have a derived form for modelling statements:

\begin{display}[.58]{Derived Form for Model Statements:}
	\clause{E \sim  \kw{d}(E_1, \dots E_n) \deq \kw{target} = \kw{target} + \kw{d_lpdf}(E, E_1, \dots E_n)}{model statement} 	
\end{display}

In Stan, ``$\sim$'' is \emph{not} considered to mean ``draw a sample from'', but rather ``modify the joint distribution over parameters and data.''
This is also reflected by the semantics given in \autoref{ssec:stansem}.

\subsection{Operational Semantics of Stan Statements} \label{ssec:stanop}
Next, we define a standard big-step operational semantics for Stan expressions and statements:
\begin{display}[.2]{Big-step Relation}
	\clause{ (s, E) \Downarrow V }{expression evaluation} \\
	\clause{ (s, S) \Downarrow s'}{statement evaluation} 
\end{display}

Here, $s$ and $s'$ are states, and values $V$ are the expressions conforming to the following grammar:
\begin{display}[0.4]{Values and States:}
	\Category{V}{value}\\
	\entry{c}{constant}\\
	\entry{[V_1,\dots,V_n]}{array}\\
	\clause{s ::= x_1 \mapsto V_1, \dots, x_n \mapsto V_n \quad x_i\textrm{ distinct}}{state (finite map from variables to values)}
\end{display}

The relation $\Downarrow$ is deterministic but partial, as we do not explicitly handle error states.
The purpose of the operational semantics is to define a density function in \autoref{ssec:stansem}, and any errors lead to the density being undefined.

In the rest of the paper, we use the notation for states $s = x_1 \mapsto V_1, \dots, x_n \mapsto V_n$:
\begin{itemize}
	\item $s[x \mapsto V]$ is the state $s$, but where the value of $x$ is updated to $V$ if $x \in \dom(s)$, or the element $x \mapsto V$ is added to $s$ if $x \notin \dom(s)$.
	\item $s[-x]$ is the state s, but where $x$ is removed from the domain of $s$ (if it were present).
\end{itemize}

We also define lookup and update operations on values:
\begin{itemize}
	\item If $U$ is an $n$-dimensional array value for $n \geq 0$
	and $c_1$, \dots, $c_n$ are suitable indexes into $U$,
	then the \emph{lookup} $U[c_1]\dots[c_n]$ is the value in $U$ indexed by $c_1$, \dots, $c_n$.
	\item If $U$ is an $n$-dimensional array value for $n \geq 0$
	and $c_1$, \dots, $c_n$ are suitable indexes into $U$,
	then the \emph{update} $U[c_1]\dots[c_n] := V$ is the array that is the same as $U$ except that the
	value indexed by $c_1$, \dots, $c_n$ is $V$.
\end{itemize}

\vspace{3pt}
\begin{display}{Operational Semantics of Expressions:}
	\squad	
	\staterule{Eval Const}
	{ }
	{ (s, c) \Downarrow c }\qquad
	
	\staterule{Eval Var}
	{ V = s(x)  \quad x \in \dom(s)}
	{ (s, x) \Downarrow V  }\qquad
	
	\staterule{Eval Arr}
	{ (s, E_i) \Downarrow V_i \quad \forall i \in 1.. n }
	{ (s, [E_1, \dots, E_n]) \Downarrow [V_1, \dots, V_n]}\qquad	
	
	\\[1.3ex]\squad
	\staterule{Eval ArrEl}
	{ (s, E_1 \Downarrow V) \quad (s, E_2 \Downarrow c)}
	{ (s, E_1[E_2]) \Downarrow V[c]}\qquad
	
	\staterule{Eval PrimCall}
	{ (s, E_i) \Downarrow V_i \quad \forall i \in 1 \dots n \quad V = f(V_1, \dots, V_n)\footnotemark}
	{ (s, f(E_1, \dots, E_n)) \Downarrow V}
\end{display}

\footnotetext{$f(V_1, \dots, V_n)$ means applying the builtin function $f$ on the values $V_1, \dots, V_n$.}
\begin{display}{Operational Semantics of Statements:}
	\staterule[~(where $L=x[E_1{]}\dots[E_n{]}$)]{Eval Assign}
	{ (s,E_i) \Downarrow V_i \quad \forall i \in 1..n \quad (s,E) \Downarrow V \quad
		U = s(x) \quad
		U' = (U[V_1]\dots[V_n] := V) }
	{ (s, L=E) \Downarrow (s[x \mapsto U'])}
	
	\\[1.3ex]
	
	\staterule{Eval Seq}
	{ (s, S_1) \Downarrow s' \quad (s', S_2) \Downarrow s''}
	{ (s, S_1;S_2) \Downarrow s''}\qquad
	
	\staterule{Eval IfTrue}
	{ (s, E) \Downarrow \kw{true} \quad (s, S_1) \Downarrow s'}
	{ (s, \kw{if}(E)\; S_1\; \kw{else}\; S_2) \Downarrow s'}\qquad
	
	\staterule{Eval IfFalse}
	{ (s, E) \Downarrow \kw{false} \quad (s, S_2) \Downarrow s'}
	{ (s, \kw{if}(E)\; S_1\; \kw{else}\; S_2) \Downarrow s'}\qquad
	
	\\[1.3ex]\squad
	\staterule[\footnotemark]{Eval ForTrue}
	{ (s, E_1) \Downarrow c_1 \quad (s, E_2) \Downarrow c_2 \quad c_1 \leq c_2 \quad (s[x \mapsto c_1], S) \Downarrow s' \quad (s'[-x], \kw{for}(x\;\kw{in}\;(c_1+1):c_2)\;S) \Downarrow s''}
	{ (s, \kw{for}(x\;\kw{in}\;E_1:E_2)\;S) \Downarrow s'' } \qquad
	
	\\[1.3ex]\squad
	\staterule{Eval ForFalse}
	{ (s, E_1) \Downarrow c_1 \quad (s, E_2) \Downarrow c_2 \quad c_1 > c_2}
	{ (s, \kw{for}(x\;\kw{in}\;E_1:E_2)\;S) \Downarrow s } \qquad
	
	\staterule{Eval Skip}
	{ }
	{ (s, \kw{skip}) \Downarrow s }\qquad
\end{display}

\footnotetext{To make shredding to Stan possible, Core Stan only supports \lstinline|for|-loops where the loop bounds do not change during execution: $E_2$ does not contain any variables that $S$ writes to. This differs from the more flexible loops implemented in Stan.}

\subsection{Syntax of Stan}  \label{ssec:stansyntax}
A full Stan program consists of six program blocks, 
each of which is optional. Blocks appear in order.
Each block has a different purpose and can reference variables declared in itself or previous blocks. Formally, we define a Stan program as a sequence of six blocks, each containing variable declarations or Stan statements, as shown next. We also present an example Stan program that contains all six blocks in \autoref{ssec:expert}.
\vspace{16pt}
\begin{display}[0.5]{Stan program:}
	\Category{P}{Stan Program}\\
	\begin{lstlisting}
	data { $\decls{d}$ }
	transformed data { $\decls{td}, \cmds{td}$ }
	parameters { $\decls{p}$ }
	transformed parameters { $\decls{tp}, \cmds{tp}$ }
	model { $\cmds{m}$ }
	generated quantities { $\decls{gq}, \cmds{gq}$ }
	\end{lstlisting}
\end{display}
\vspace{-3pt}
Arrays in Stan are sized, but we do not include any static checks on array sizes in this paper. 
\begin{display}[0.5]{Stan Types and Type Environment:}
	\clause{\Gamma ::= x_1:\tau_1, \dots, x_n:\tau_n \quad \forall i \in 1\dots n \hquad x_i\textrm{ distinct}}{declarations} \\	
	\clause{\tau ::= \kw{real} \; | \; \kw{int} \; | \; \kw{bool} \; | \; \tau [n] }{type} \\
	\clause{n}{size}
\end{display}

The size of an array, $n$, can be a number or a variable.
For simplicity, we treat $n$ as decorative and do not include checks on the sizes in the type system of Stan. However, the system can be extended to a lightweight dependent type system, similarly to Tabular as extended by \citet{MarcinDiss}. 

Each program block in Stan has a different purpose as follows:
\begin{itemize}
	\item \lstinline|data|: declarations of the input data. 
	\item \lstinline|transformed data|: definition of known constants and \textit{preprocessing} of the data.
	\item \lstinline|parameters|: declarations of the parameters of the model. 
	\item \lstinline|transformed parameters|: declarations and statements defining transformations of the data and parameters.
	\item \lstinline|model|: statements defining the distributions of random variables in the model.
	\item \lstinline|generated quantities|: declarations and statements that do not affect inference, used for \textit{postprocessing}, or predictions for unseen data.
\end{itemize}

We define a conformance relation on states $s$ and typing environments $\Gamma$. A state $s$ conforms to an environment $\Gamma$, whenever $s$ provides values of the correct types for the variables given in $\Gamma$:
\vspace{-8pt}\begin{multicols}{2}
\begin{display}[-0.03]{Conformance Relation:}
	\\
	\clause{ s \models \Gamma }{state $s$ conforms to environment $\Gamma$}\\[2.5pt]	
\end{display}
\begin{display}[.2]{Rule for the Conformance Relation:}
	\hquad
	\staterule{Stan State}
	{ V_i \models \tau_i \quad \forall i \in I}
	{(x_i \mapsto V_i)^{i \in I} \models (x_i : \tau_i)^{i \in I}}
\end{display}
\end{multicols}  \vspace{-8pt}

Here, $V \models \tau$ denotes that the value $V$ is of type $\tau$, and has the following definition:
\begin{itemize}
	\item $c \models \kw{int}$, if $c \in \mathbb{Z}$, $c \models \kw{real}$, if $c \in \mathbb{R}$, and $c \models \kw{bool}$ if $c \in \{\kw{true}, \kw{false}\}$. 
	\item $[V_1,\dots,V_n] \models \tau[m]$, if $\forall i \in 1\dots n. V_i \models \tau$. 
\end{itemize}

We do not include any checks on array sizes in this paper, thus we do not assume $n$ and $m$ are the same in this definition. 
The evaluation relation is not defined on initial states that lead to array out-of-bounds errors.

\subsection{Density-based Semantics of Stan} \label{ssec:stansem}
Finally, we give the semantics of Stan in terms of the big-step relation from \autoref{ssec:stanop}. As the \citet{StanManual} explain:
\begin{quote}
 A Stan program defines a statistical model through a conditional probability function $p(\theta \mid y,x)$, where $\theta$ is a sequence of modeled unknown values (e.g., model parameters, latent variables, $\dots$), $y$ is a sequence of modeled known values, and $x$ is a sequence of unmodeled predictors and constants (e.g., sizes, hyperparameters). (p.~22)
\end{quote}

More specifically, a Stan program is executed to evaluate a function on the data and parameters $\log p^*(\params \mid \data)$, for some given (and fixed) values of $\data$ and $\params$. This function encodes the log joint density of the data and parameters $\log p(\params, \data)$ up to an additive constant, and also equals
the log density of the posterior $\log p(\params \mid \data)$ up to an additive constant:
$$ \log p(\params \mid \data) =  \log p(\params, \data) - \log p(\data) \propto \log p(\params, \data) \propto \log p^*(\params \mid \data)$$

The return value of $\log p^*(\params \mid \data)$ is stored in the reserved variable $\kw{target}$. We give the semantics of Core Stan through this \textit{unnormalised log posterior density function}.

Consider a Core Stan program $P$ defined as previously, and the statement $S = S_{td}; S_{tp}; S_{m}; S_{gq}$.
The semantics of $P$ is the unnormalised log posterior density function $\log p^*$ on parameters $\params$ given data $\data$ (where $\params \models \Gamma_p$ and $\data \models \Gamma_d$):
$$\log p^*\left( \params \mid \data \right) \deq s'[\kw{target}] \text{ if there is }s'\text{ such that }((\data, \params, \kw{target} \mapsto 0), S) \Downarrow s'$$

If there is no such $s'$ then the log density is undefined. Observe also that if such an $s'$ exists, it is unique, because the operational semantics is deterministic. 

For example, suppose that $P$ specifies a simple model for the data array \lstinline|y|:
\begin{lstlisting}
	data { int N; real[N] y; }
	parameters { real mu; real sigma; }
	model {
		mu ~ normal(0, 1); 
		sigma ~ normal(0, 1);
		for(i in 1:N){ y[i] ~ normal(mu, sigma); }
	}
\end{lstlisting}

Suppose also that $\params = (\kw{mu} \mapsto \mu, \kw{sigma} \mapsto \sigma)$ and $\data = (\kw{N} \mapsto n, \kw{y} \mapsto \mathbf{y})$, for some $\mu$, $\sigma$, $n$, and a vector $\mathbf{y}$ of length $n$. The statement $S = S_{td}; S_{tp}; S_{m}; S_{gq}$ is then the body of the \lstinline|model| block as specified above.
Then $((\data, \params, \kw{target} \mapsto 0), S) \Downarrow s'$, with
$s'[\kw{target}] = \log\normal(\mu \mid 0, 1) + \log\normal(\sigma \mid 0, 1) + \sum_{i=1}^n\log\normal(y_i \mid \mu, \sigma)$.
This is precisely the log joint density on \lstinline|mu|, \lstinline|sigma| and \lstinline|y|, which is proportionate to the posterior of \lstinline|mu| and \lstinline|sigma| given \lstinline|y|.

The function $\log p^*\left( \params \mid \data \right)$ is \emph{not} a (log) density, but rather it encodes the logarithm of the density of the posterior up to an additive constant. Such unnormalised log density uniquely defines the log density $\log p\left( \params \mid \data \right)$:
 $$\log p\left( \params \mid \data \right) = \log p^*\left( \params \mid \data \right) - \log Z(\data)\, \text{ where }\, Z(\data) = \int p^*(\params \mid \data)\,d\params$$

The value $Z(\data)$ is called the \emph{normalising constant} (it is a constant with respect to the variables $\params$ that the density is defined on). Computing $Z(\data)$ is in most cases intractable. Thus, many inference algorithms (including those of Stan) are designed to successfully approximate the posterior, relying only on being able to evaluate a function proportional to it: an \emph{unnormalised density function}, such as $\log p^*\left( \params \mid \data \right)$ above. 

The goal of this paper is to formalise the statistical meaning of a Stan program, as given by the quotation from the reference manual above. This semantics concentrates on defining the unnormalised log posterior of parameters given data, but omits the fact that the values of transformed parameters and generated quantities blocks are also part of the observable state.
Transformed parameters and generated quantities can be seen as variables that are generated using the function $g(\params, \data) \deq s'[\dom(\Gamma_{tp} \cup \Gamma_{gq})]$ for $s'$ defined as previously.
Appendix~\ref{ap:gq_semantics} discusses generated quantities in more detail, and we leave their full treatment for future work. 
Moreover, Appendix~\ref{ap:sampling_semantics} discusses how this density-based semantics relates to other imperative probabilistic languages semantics, such as the sampling-based semantics of \textsc{Prob} \cite{HurNRS15}.

\subsection{Inference}
Executing a Stan program consists of generating samples from the  \textit{posterior distribution} $p(\params \mid \data)$, as a way of performing \textit{Bayesian inference}.
The primary algorithm that Stan uses is the asymptotically exact algorithm  Hamiltonian Monte Carlo (HMC) \cite{HMCfirst}, and more specifically, an enhanced version of the No-U-Turn Sampler (NUTS) \cite{NUTS, HMCConceptual}, which is an adaptive path lengths extension to HMC.                    

HMC is a Markov Chain Monte Carlo (MCMC) method (see \citet{MCMCIain} for a review on MCMC). Similarly to other MCMC methods, it obtains samples $\{\params_i\}_{i=1}^{\infty}$ from the target distribution, by using the latest sample $\params_n$ and a carefully designed transition function $\delta$ to generate a new sample $\params_{n+1} = \delta(\params_n)$. 
When sampling from the posterior $p(\params \mid \data)$, HMC evaluates the unnormalised density $p^*(\params \mid \data)$ at several points in the parameter space at each step $n$.
To improve performance, HMC also uses the gradient of $\log p^*(\params \mid \data)$ with respect to $\boldsymbol{\theta}$.

\section{SlicStan} \label{sec:slicstan}
This section outlines the second key contribution of this work --- the design and semantics of SlicStan.
SlicStan is a probabilistic programming language, which aims to provide a more compositional alternative to Stan, while retaining Stan's efficiency and statement syntax natural to the statistics community. Thus, we design the language so that:
\begin{enumerate}
	\item SlicStan statements are a superset of the Core Stan statements given in \autoref{sec:stan}, 
	\item SlicStan programs contain no program blocks, and allow the interleaving of statements that would belong to different program blocks if the program was written in Stan, and 
	\item SlicStan supports first-order non-recursive functions.
\end{enumerate}

This results in a flexible syntax, that allows for better encapsulation and code reuse, similarly to Edward \cite{Edward} and PyMC3 \cite{PyMC3}.

The key idea behind SlicStan is to use \textit{information flow analysis} to optimise and transform the program to Stan code. 
Secure information flow analysis has a long history, summarised by \citet{InfoFlowSurvey}, and \citet{InfoFlowSmithPrinciples}. It concerns systems where variables have one of several \textit{security levels}, and the aim is to disallow the flow of high-level information to a low-level variable, but allow other flows of information. For example, consider two security levels, \lev{public} and \lev{secret}. We want to forbid \lev{public} information to depend on \lev{secret} information.  
Formally, the levels form a \textit{lattice} $\left(\{\lev{public}, \lev{secret}\}, <\right)$ with $\lev{public} < \lev{secret}$.
Secure information flow analysis is used to ensure that information flows only upwards in the lattice. This is formalized as the \textit{noninterference property} \cite{goguen1982security}: confidential data may not interfere with public data. 

Looking back to the description of Stan's program blocks in \autoref{ssec:stansyntax}, as well as the Stan Manual, we identify three information levels in Stan: \lev{data}, \lev{model}, and \lev{genquant}. We assign one of these levels to each program block, as summarised by \autoref{tab:blocks}. `Chain', `sample' and `leapfrog' refer to stages of the Hamiltonian Monte Carlo sampling algorithm. Usually, Stan runs several chains to perform inference, where there are many samples per chain, and many leapfrogs per sample.

\begin{table}[h]
	\centering
	\begin{tabular}{@{}lll@{}}\toprule[1.2pt]
		Block & Execution & Level \\\midrule
		\lstinline|data| & --- & \lev{data} \\
		\lstinline|transformed data| & per chain & \lev{data}\\
		\lstinline|parameters| & --- & \lev{model} \\
		\lstinline|transformed parameters| & per leapfrog & \lev{model} \\
		\lstinline|model| & per leapfrog & \lev{model} \\
		\lstinline|generated quantities| & per sample & \lev{genquant}\\		 
		\bottomrule[1.2pt]
	\end{tabular}
	\caption{Program blocks in Stan. Adapted from \citet{StanTutorial}.}
	\label{tab:blocks}
	\vspace{-21pt}
\end{table}

Even though our insight about the three information levels comes from Stan, they are not tied to Stan's peculiarities. Variables at level \lev{data} are the known quantities in the statistical inference problem, that is, the data. Computations at this level can be seen as a form of \textit{preprocessing}. 
Variables at level \lev{model} are unknown --- they are the quantities we wish to infer. Changing the \lev{model} variables or the dependencies between them changes the statistical model we are working with, which can have a huge effect on the quality of inference.   
Finally, generated quantities are variables that \lev{data} and \lev{model} variables do not depend on, and computing them can be seen as a form of \textit{postprocessing}. All three are fundamental concepts of statistical inference and are not specific to Stan.

The rest of this section defines the SlicStan language. 
The syntax of SlicStan statements (\autoref{ssec:slicsyntax}) extends that of the Core Stan statements from \autoref{sec:stan}, and its type system (\autoref{ssec:slictyping}) assumes level types \lev{data}, \lev{model} and \lev{genquant} on variables. The typing rules are then implemented so that in well-typed SlicStan programs, information flows from level \lev{data}, through \lev{model} to \lev{genquant}. Every Core Stan program can be turned into an equivalent SlicStan program by concatenating the statements and declarations in its compartments.

Next, we give the semantics of a SlicStan program, much as we did for Core Stan, as an unnormalised log density function on parameters and data (\autoref{ssec:slicsem}), and show some examples (\autoref{ssec:slicexamples}).
To do so, we \textit{elaborate} SlicStan's statements to Core Stan statements by statically unrolling user-defined function calls and bringing all variable declarations to the top level (\autoref{ssec:slicelab}).
The main purpose of elaboration is to identify all parameters statically so as to give the semantics as a function on the parameters.
Elaboration also serves as a first step in translating SlicStan to Stan (\autoref{sec:translate-slicstan-to-stan}).

\subsection{Syntax} \label{ssec:slicsyntax}

A SlicStan program is a sequence of function definitions $F_i$, followed by top-level statement $S$. 
\vspace{-2pt} 
\begin{display}{Syntax of a SlicStan Program}
	\clause{F_1, \dots, F_n, S \quad n \geq 0}{SlicStan program}
\end{display}

SlicStan's user-defined functions are not recursive (a call to $F_i$ can only occur in the body of $F_j$ if $i < j$).
Functions are specified by a return type $T$, arguments with their types $a_i:T_i$, and a body $S$.
There is a reserved variable \kw{ret_g} associated with each function, to hold the return value.
\vspace{-2pt}  
\begin{display}{Syntax of Function Definitions}
	\clause{F ::= T\;g(T_1\,a_1, \dots, T_n\,a_n)\;S}{function definition (signature $g: T_1,\dots,T_n \to T$)}
\end{display}

SlicStan's expressions and statements extend those of Stan, by user-defined function calls $g(E_1,\dots,E_n)$ and variable declarations $T\,x;\;S $ (in italic). 

In both declarations $T\,x;\;S$ and loops $\kw{for}(x\;\kw{in}\;E_1:E_2)\;S$, the variable $x$ is locally bound with scope $S$. We identify statements up to consistent renaming of bound variables. Note that occurrences of variables in L-values are free.  

\begin{multicols}{2}
\begin{display}[.2]{SlicStan Syntax of Expressions:}
	\Category{E}{expression}\\
	\entry{x}{variable}\\
	\entry{c}{constant}\\
	\entry{[E_1,...,E_n]}{array}\\
	\entry{E_1[E_2]}{array element}\\ 
	\entry{f(E_1,\dots,E_n)}{builtin function call}\\
	\entry{\mathit{g(E_1,\dots,E_n)}}{\emph{user-defined fun. call}}\\
	\Category{L}{L-value}\\
	\entry{x}{variable}\\
	\entry{x[E_1]\dots[E_n]}{array element} 
\end{display}
\begin{display}[.25]{SlicStan Syntax of Statements:}
	\\
	\Category{S}{statement}\\
	\entry{L = E}{assignment}\\
	\entry{S_1; S_2}{sequence}\\
	\entry{\kw{for}(x\;\kw{in}\;E_1:E_2)\;S}{for loop} \\
	\entry{\kw{if}(E)\;S_1\,\kw{else}\;S_2}{if statement} \\
	\entry{\kw{skip}}{skip} \\
	\entry{\mathit{T\,x;\;S}}{\emph{declaration}}\\
	\\
\end{display}
\end{multicols} \vspace{-12pt} 

We constrain the language to only support \lstinline|for| loops, disallowing the value of the loop guard to depend on the body of the loop. As described in later subsections, in order to give the semantics of a SlicStan program, as well as to translate it to Stan, we need to \textit{elaborate} the statements to Core Stan statements (\autoref{ssec:slicelab}), statically unrolling user-defined functions and extracting variable declarations to the top-level. Extending the language to support \lstinline|while| loops (or recursive functions) means risk of non-terminating elaboration step, and a potentially inefficient resulting Stan program.
This design choice is a small restriction on the usability and range of expressible models compared to Stan: models in Stan can only have a fixed number of parameters. As a result, an overwhelming number of examples in the Stan official repository use \lstinline|for| loops only. 

We define derived forms for data declarations, modelling statements, and return statements. 
Any user-defined function \lstinline|D_lpdf| can be used as the log density function of a user-defined distribution \lstinline|D| on the first argument of \lstinline|D_lpdf|.
For the sake of simplicity, we assume the body of a user-defined function \lstinline|g| contains \textit{at most one} return statement, at the end,
and we treat it as an assignment to the return variable \lstinline|ret_g|. 

\vspace{-1pt} 
\begin{display}[.52]{Derived Forms}
	\clause{\kw{data }\textrm{ }\tau\,x; S \deq (\tau, \lev{data})\,x; S}{data declaration} \\	
	\clause{E \sim \kw{d}(E_1, \dots E_n) \deq \kw{target} = \kw{target} + \kw{d_lpdf}(E, E_1, \dots E_n)}{model, builtin distribution} \\
	\clause{E \sim \kw{D}(E_1, \dots E_n) \deq \kw{target} = \kw{target} + \kw{D_lpdf}(E, E_1, \dots E_n)}{model, user-defined distribution} \\	
	\clause{\kw{return}\; E  \deq \kw{ret_g} = E}{return}
\end{display}
\vspace{-10pt} 
\subsection{Typing of SlicStan} \label{ssec:slictyping} 

Next, we present SlicStan's type system. 
We define a lattice $\left(\{\lev{data}, \lev{model}, \lev{genquant}\}, <\right)$ of level types, where $\lev{data} < \lev{model} < \lev{genquant}$. 
Types $T$ in SlicStan range over pairs $(\tau, \ell)$ of a base type $\tau$, and a level type $\ell$ --- one of \lev{data}, \lev{model}, or \lev{genquant}.
Arrays are sized, with $n \geq 0$.
Each builtin function $f$ has a family of signatures $f:(\tau_1,\ell),\dots,(\tau_n,\ell) \to (\tau,\ell)$, one for each level $\ell$.

\vspace{-1pt} 
\begin{display}{Types, and Typing Environment:}
	\Category{\ell}{level type}\\
	\entry{\lev{data}}{data, transformed data}\\
	\entry{\lev{model}}{parameters, transformed parameters}\\
	\entry{\lev{genquant}}{generated quantities}\\
	\clause{n}{size} \\
	\clause{\tau ::= \kw{real} \; | \; \kw{int} \; | \; \kw{bool} \; | \; \tau [n] }{base type} \\
	\clause{T ::= (\tau,\ell)}{type: base type and level} \\
	\clause{\Gamma ::= x_1:T_1, \dots, x_n:T_n \quad x_i\textrm{ distinct}}{typing environment} 
\end{display}
(While builtin functions of our formal system are level polymorphic, user-defined functions are monomorphic.
This design choice was made to keep the system simple, and we see no challenges to polymorphism that are unique to SlicStan.)
 
We assume the type of the reserved \kw{target} variable to be $(\kw{real}, \lev{model})$: this variable can only be accessed within the \lstinline|model| block in Stan, thus its level is \lev{model}.
Each function $g$ is associated with a single return variable \lstinline|ret_g| matching the return type of the function.

\begin{display}[.2]{Reserved variables}
	\clause{\kw{target} : (\kw{real}, \lev{model})}{log joint probability density function}\\
	\clause{\kw{ret_g} : T}{return value of a function $T\;g(T_1\,a_1, \dots, T_n\,a_n)\;S$}
\end{display}
 
We present the full set of declarative typing rules, inspired by those of the secure information flow calculus defined by \citet{InfoFlowVolpano}, and more precisely, its summary by \citet{InfoFlowAbadi}. The information flow constraints are enforced by the subsumption rules \ref{ESub} and \ref{SSub}, which together allow information to only flow upwards the $\lev{data} < \lev{model} < \lev{genquant}$ lattice. 

Intuitively, we need to associate each expression $E$ with a level type to prevent lower-level variables to \textit{directly} depend on higher-level variables, such as in the case \lstinline|d = m + 1|, for \lstinline|d| of level \lev{data} and \lstinline|m| of level \lev{model}. We also need to associate each statement $S$ with a level type to prevent lower-level variables to \textit{indirectly} depend on higher-level variables, such as in the case \lstinline|if(m > 0)$\,$d = 2|.

\vspace{3pt} 
\begin{display}[.2]{Judgments of the Type System:}
	\clause{ \Gamma \vdash E : (\tau, \ell) }{expression $E$ has type $(\tau, \ell)$ and reads only level $\ell$ and below}\\
	\clause{ \Gamma \vdash S : \ell }{statement $S$ assigns only to level $\ell$ and above}\\
    \clause{ \vdash F }{function definition $F$ is well-typed}
\end{display}

\vspace{3pt} 
The function \lstinline|ty(c)| maps constants to their types (for example \lstinline|ty(5.5) = real|). 

\vspace{3pt} 
\begin{display}{Typing Rules for Expressions:}
	\squad
	\staterule{ESub}
	{ \Gamma \vdash E : (\tau,\ell) \quad \ell \leq \ell'}
	{ \Gamma \vdash E : (\tau,\ell') }
	\quad\,
	\staterule{Var}
	{ }
	{ \Gamma, x:T \vdash x:T}  \quad\,
	\staterule{Const}
	{ \kw{ty}(c) = \tau }
	{ \Gamma \vdash c : (\tau,\lev{data}) }\quad\,
	\staterule{Arr}
	{\Gamma \vdash E_i : (\tau,\ell) \quad \forall i \in 1..n}
	{\Gamma \vdash [E_1,...,E_n] : (\tau [n],\ell)}
	
	\\[\GAP]\squad
	\staterule{ArrEl}
	{\Gamma \vdash E_1 : (\tau[n], \ell) \quad \Gamma \vdash E_2 : (\kw{int}, \ell)}
	{\Gamma \vdash E_1[E_2] : (\tau,\ell)}\quad

	\staterule[($f: T_1,\dots,T_n \to T$)]
	{PrimCall}
	{ \Gamma \vdash E_i : T_i \quad \forall i \in 1..n}
	{ \Gamma \vdash f(E_1,\dots,E_n) : T } \squad
	
	\staterule[($g: T_1,\dots,T_n \to T$)]
	{FCall}
	{ \Gamma \vdash E_i : T_i \quad \forall i \in 1..n}
	{ \Gamma \vdash g(E_1,\dots,E_n) : T }
\end{display}

\vspace{3pt} 
Here and throughout, we make use of several functions on the language building blocks:
\begin{itemize}
	\item $\assset(S)$ (Definition~\ref{def:assset}) is the set of variables that are assigned to in $S$: $\assset(x=2*y) = \{x\}$. 
	\item $\readset(S)$ (Definition~\ref{def:readset}) is the set of variables read by $S$: $\readset(x=2*y) = \{y\}$.
	\item $\Gamma(L)$ (Definition~\ref{def:lvaluetypes}) is the type of the L-value $L$ in the context $\Gamma$: \\ $\Gamma(x[0]) = (\kw{real}, \lev{data})$ for $x:(\kw{real}[], \lev{data}) \in \Gamma$.
\end{itemize}

The rule \ref{Decl} for a variable declaration $(\tau, \ell)\,x; S$ has a side-condition ($x \notin \dom(\Gamma)$), where $\Gamma$ is the local typing environment,
that enforces that the variable $x$ is globally unique, that is, there is no other declaration of $x$ in the program.
The condition $x \notin \assset(S)$ in \ref{For} enforces that the loop index $x$ is immutable inside the body of the loop. In \ref{Seq}, we make sure that the sequence $S_1; S_2$ is \textit{shreddable}, through the predicate $\shreddable(S_1, S_2)$ (Definition~\ref{def:shreddable}). This imposes a restriction on the range of well-typed programs, which is needed both to allow translation to Stan (see \autoref{ssec:shred}), and to allow interpreting of the program in terms of preprocessing, inference and postprocessing. 

Using the rules for expressions and statements, we can also obtain rules for the derived statements.

\vspace{1pt} 
\begin{display}{Typing Rules for Statements:}
	\squad
	\staterule{SSub}
	{ \Gamma \vdash S : \ell' \quad \ell \leq \ell'}
	{ \Gamma \vdash S : \ell }\quad
	\staterule{Assign}
	{ \Gamma(L) = (\tau, \ell) \quad \Gamma \vdash E : (\tau,\ell)}
	{ \Gamma \vdash (L = E) : \ell }\qquad 

	\staterule{Decl}
	{\Gamma, x:(\tau, \ell) \vdash S : \ell' \quad x \notin \dom(\Gamma)
    }
	{\Gamma \vdash (\tau, \ell)\,x; S : \ell'}\qquad

	\\[\GAP]\squad	
	\staterule{If}
	{ \Gamma \vdash E : (\kw{bool},\ell) \quad \Gamma \vdash S_1 : \ell \quad \Gamma \vdash S_2 : \ell}
	{ \Gamma \vdash \kw{if}(E)\;S_1 \;\kw{else}\; S_2: \ell }\qquad
	
	\staterule{Seq}
	{ \Gamma \vdash S_1 : \ell \quad \Gamma \vdash S_2 : \ell \quad \shreddable(S_1, S_2)}
	{ \Gamma \vdash (S_1; S_2) : \ell }\qquad	
	
	\staterule{Skip}
	{ }
	{ \Gamma \vdash \kw{skip} : \ell } \qquad
	
	\\[\GAP]\squad
	\staterule{For}
	{ \Gamma \vdash E_1 : (\kw{int},\ell) \quad \Gamma \vdash E_2 : (\kw{int},\ell) \quad \Gamma, x:(\kw{int}, \ell) \vdash S : \ell \quad x \notin \dom(\Gamma) \quad x \notin \assset(S)}
	{ \Gamma \vdash \kw{for}(x\;\kw{in}\;E_1:E_2)\;S : \ell } \qquad
\end{display}

\vspace{1pt} 
\begin{display}{Derived Typing Rules}
	\squad
	\staterule{DataDecl}
	{ \Gamma, x: (\tau, \lev{data}) \vdash S:\ell \quad x \notin \dom(\Gamma)}
	{\Gamma \vdash \kw{data }\textrm{ }\tau\,x; S : \ell}\qquad

	\staterule[($\kw{d_lpdf}: T, T_1,\dots,T_n \to (\kw{real}, \lev{model})$)]{PrimModel}
	{ \Gamma \vdash E : T \quad \Gamma \vdash E_i : T_i \quad \forall i \in 1..n} 
	{ \Gamma \vdash E \sim \kw{d}(E_1, \dots E_n) : \lev{model} }\qquad

	\\[\GAP]\squad		
	
	\staterule{Return}
	{ \Gamma \vdash \kw{ret_g}:(\tau,\ell) \quad \Gamma \vdash E : (\tau, \ell) }
	{ \Gamma \vdash \kw{return}\;E : \ell }\qquad

	\staterule[($\kw{D}: T, T_1,\dots,T_n \to (\kw{real}, \lev{model})$)]{FModel}
	{ \Gamma \vdash E : T \quad \Gamma \vdash E_i : T_i \quad \forall i \in 1..n}
	{ \Gamma \vdash E \sim \kw{D_dist}(E_1, \dots E_n) : \lev{model} }\qquad	
\end{display}

Finally, we complete the three judgments of the type system with the rule \ref{FDef} for checking the well-formedness of a function definition.
The condition $\ell_i \leq \ell$ ensures that the level of the result of a function call is no smaller than the level of its arguments.

\vspace{1pt} 
\begin{display}{Typing Rule for Function Definitions:}
	\squad
	\staterulealt{FDef}
	{ a_1:T_1,...,a_n:T_n, \kw{ret_g}:(\tau,\ell) \vdash S : \ell  \quad T_i = (\tau_i,\ell_i) \quad \ell_i \leq \ell \quad \forall i \in 1..n}
	{ \vdash (\tau,\ell)\;g(T_1\,a_1, \dots, T_n\,a_n)\;S }  
\end{display}

In our formal development, we implicitly assume a fixed program with well-typed functions $\vdash F_1$, \dots, $\vdash F_n$.
More precisely, we assume a given well-formed program defined as follows.

\vspace{1pt} 
\begin{display}{Well-Formed SlicStan Program:}
\squad
A program $F_1,\dots F_N, S$ is \emph{well-formed} iff $\vdash F_1$, \dots, $\vdash F_n$, and $\emptyset \vdash S : \lev{data}$.
\end{display}

SlicStan statements are, by design, a superset of Core Stan statements. Thus, we can treat any Core Stan statement as a SlicStan statement with big-step operational semantics defined as in \autoref{sec:stan}. By extending the conformance relation $s \models \Gamma$ to correspond to a SlicStan typing environment, we can prove type preservation of the operational semantics, with respect to SlicStan's type system. 

\begin{display}[.2]{Rule of the Conformance Relation:}
	\squad
	\staterulealt{State}
	{V_i \models \tau_i \quad \forall i \in I}
	{(x_i \mapsto V_i)^{i \in I} \models (x_i : (\tau_i, \ell_i))^{i \in I}}
\end{display}

\begin{theorem}[Type Preservation for $\Downarrow$]\label{th:eval}~
	
	For a Core Stan statement $S$ and a Core Stan expression $E$:
	\begin{enumerate}
		\item If $s \models \Gamma$ and $\Gamma \vdash E : (\tau, \ell)$ and $(s,E) \Downarrow V$ then $ V \models \tau$.
		\item If $s \models \Gamma$ and $\Gamma \vdash S : \ell$ and $(s,S) \Downarrow s'$ then $s' \models \Gamma$.
	\end{enumerate}
\end{theorem}
\begin{proof}
	By inductions on the size of the derivations of the judgments $(s,E) \Downarrow V$ and $(s,S) \Downarrow s'$.	
\end{proof}

Finally, we state a termination-insensitive noninterference result.
Intuitively, the result means that (observed) data cannot depend on the model parameters, and that generated quantities do not affect the log density distribution defined by the model.
\begin{definition}[$\ell$-equal states] Given a typing environment $\Gamma$, states $s_1 \models \Gamma$ and $s_2 \models \Gamma$ are $\ell$-equal for some level $\ell$ (written $s_1 \approx_{\ell} s_2$), if they differ only for variables of a level strictly higher than $\ell$:
	$$s_1 \approx_{\ell} s_2 \deq \forall x:(\tau, \ell') \in \Gamma. \left( \ell' \leq \ell \implies s_1(x) = s_2(x) \right)$$
\end{definition}

\begin{theorem}[Noninterference] \label{th:noninterf} Suppose $s_1 \models \Gamma$, $s_2 \models \Gamma$, and $s_1 \approx_{\ell} s_2$ for some $\ell$. Then for Core Stan statements $S$ and Core Stan expressions $E$:
	\begin{enumerate}
		\item If $~\Gamma \vdash E:(\tau,\ell)$ and $(s_1, E) \Downarrow V_1$ and $(s_2, E) \Downarrow V_2$ then $V_1 = V_2$. 
		\item If $~\Gamma \vdash S:\ell$ and $(s_1, S) \Downarrow s_1'$ and $(s_2, S) \Downarrow s_2'$ then $s_1' \approx_{\ell} s_2'$.
	\end{enumerate}
\end{theorem}

\begin{proof} (1)~follows by rule induction on the derivation $\Gamma \vdash E:(\tau, \ell)$, and using that if $\Gamma \vdash E:(\tau, \ell)$, $x \in \readset(E)$ and $\Gamma(x) = (\tau', \ell')$, then $\ell' \leq \ell$. (2)~follows by rule induction on the derivation $\Gamma \vdash S:\ell$ and using (1).
\end{proof}

\subsection{Elaboration of SlicStan} \label{ssec:slicelab}
Similarly to Stan, a SlicStan program defines a probabilistic model, through an unnormalised log density function on the model parameters and data. That is, the semantics of SlicStan is in terms of a fixed (or data-dependent) number of variables. Therefore, in order to be able to formally give the semantics, we need to statically unroll calls to user-defined functions, and pull all variable declarations to the top level.
(We discuss the difficulties of directly specifying the semantics of SlicStan without elaboration in \autoref{ssec:difficulty}.)

We call this static unrolling step \textit{elaboration}, and we formalise it through the elaboration relation $\elab$. Intuitively, to elaborate a program $F_1,\dots F_N, S$, we elaborate its main body $S$ by unrolling any calls to $F_1, \dots, F_N$ (as specified by \ref{Elab FCall}), and move all variable declarations to the top level (as specified by \ref{Elab Decl}). The result is an elaborated SlicStan statement $S'$ and a list of variable declarations $\Gamma$.
As mentioned previously, to avoid notational clutter, we assume a top-level SlicStan program $F_1,\dots, F_N, S$.
Since the syntax of a SlicStan statement differs from that of a Core Stan statement only by the presence of user-defined function calls and variable declarations, an \textit{elaborated SlicStan} statement is also a well-formed \textit{Core Stan} statement.

\vspace{-8pt}
\begin{multicols}{2} 
\begin{display}[.20]{Elaboration Relation}
	\clause{P \elab[\emptyset] \elpair{\ctx}{S'}}{program elaboration}\\	
	\clause{S \elab \elpair{\ctx}{S'}}{statement elaboration} \\
	\clause{E \elab \elpair{\ctx}{S.E'}}{expression elaboration} \\
	\clause{F \elab \elpair{r\!:\!T,A,\ctx}{S}}{fun. definition elaboration}
\end{display}
\begin{display}{Elaboration Rule for a SlicStan Program}
	\quad
	\staterule{Elab SlicStan}
	{S \elab[\emptyset] \elpair{\ctx}{S'}}
	{F_1,\dots F_N, S \elab[\emptyset] \elpair{\ctx}{S'}}  \\[-2.7pt]
\end{display}
\end{multicols}\vspace{-8pt}

The \textit{unrolling} rule \ref{Elab FCall} assumes a call to a user-defined function $g$ with definition $F=T\;g(T_1\,a_1, \dots, T_n\,a_n)\;S$, which elaborates as described by \ref{Elab FDef}. 

\begin{display}{Elaboration Rules for Expressions:}
	\staterule{Elab Var}
	{ }
	{ x \elab  \elpair{\emptyset}{\kw{skip}.x}}\qquad
	
	\staterule{Elab Const}
	{ }
	{ c \elab  \elpair{\emptyset}{\kw{skip}.c}}\qquad
	
	\staterule{Elab ArrEl}
	{ E_1 \elab \elpair{\ctx_1}{S_1.E_1'} \quad 
		E_2 \elab \elpair{\ctx_2}{S_2.E_2'} \quad 
		\ctx_1 \cap \ctx_2 = \emptyset}
	{ E_1[E_2] \elab \elpair{\ctx_1 \cup \ctx_2}{S_1; S_2.(E_1'[E_2'])}}\qquad
	
	\\[\GAP]
	\staterule{Elab Arr}
	{ E_i \elab \elpair{\ctx_i}{S_i.E_i'} \quad \forall i \in 1..n \quad \bigcap_{i=1}^n \ctx_i = \emptyset}
	{[E_1, ..., E_n] \elab \elpair{\bigcup_{i=1}^n \ctx_i}{S_1;...;S_n.([E_1', ..., E_n'])}}\quad
	
	\staterule{Elab PrimCall}
	{ E_i \elab \elpair{\ctx_i}{S_i.E_i'} \quad \forall i \in 1..n \quad \bigcap_{i=1}^n \ctx_i = \emptyset}
	{ f(E_1, ..., E_n) \elab \elpair{\bigcup_{i=1}^n \ctx_i}{S_1;...;S_n. f(E_1',...,E_n')}}

	\\[\GAP]
	\staterule[~(where $F$ is the definition for function $g$)]{Elab FCall}
	{ E_i \elab \elpair{\ctx_{E_i}}{S_{E_i}.E_i'} \quad \forall i \in 1..n \quad
		F \elab \elpair{r_F:T_F, A_F,\ctx_F}{S_F} \\
		A_F = \set{a_i:T_i \given i \in 1..n} \quad
		\{r_F:T_F\}\cap A_F\cap(\bigcap_{i=1}^n \ctx_i)\cap\ctx_F = \emptyset }
	{ g(E_1, ..., E_n) \elab \elpair{\{r_F:T_F\} \cup A_F\cup(\bigcup_{i=1}^n \ctx_i)\cup\ctx_F}{(S_{E_1}; a_1=E_1';...;S_{E_n};a_n=E_n'; S_F.r_F)}} 
\end{display}

\begin{display}{Elaboration Rule for Function Definitions:}
	\squad
	\staterule{Elab FDef}
	{ S \elab[\set{r:T} \cup \Gamma_A \cup \Gamma] \elpair{\ctx'}{S'} \quad 
		\Gamma_A = \set{a_1:T_1,...,a_n:T_n} \quad
		\set{r} \cap \dom(\Gamma_A) \cap \dom(\ctx') = \emptyset }
	{ T\;g(T_1\,a_1,...,T_n\,a_n)\;S \elab \elpair{r:T,\Gamma_A,\ctx' }{S'}}  
\end{display}

As we identify statements up to $\alpha$-conversion, $T\,x;\;x = 1$ elaborates to $\elpair{\{x_1:T\}}{x_1=1}$, but also to $\elpair{\{x_2:T\}}{x_2=1}$, and so on. The \ref{Elab Decl} rule simply extracts a variable declaration to the top level. Other than recursively applying $\elab$ to sub-parts of the statement, the \ref{Elab If} and \ref{Elab For} rules transform the guards of the respective compound statement to be the fresh variables $g$ or $g_1$, $g_2$ respectively (as opposed to unrestricted expressions). This is a necessary preparation step needed for the program to be correctly translated to Stan later (see \autoref{ssec:shred} and Appendix~\ref{ap:shred}). 

\begin{display}{Elaboration Rules for Statements:}		
	\staterule{Elab Decl}
	{ S \elab[\{x:T\} \cup \Gamma] \elpair{\ctx'}{S'} \quad x \notin \dom(\ctx')}
	{ T\,x; S \elab \elpair{\set{x:T}\cup\ctx'}{S'}} \qquad

	\staterule{Elab Skip}
	{}
	{\kw{skip} \elab \elpair{\emptyset}{\kw{skip}}}\qquad
	
	\\[\GAP]\squad
	\staterule{Elab Assign}
	{ L \elab \elpair{\ctx_L}{S_L.L'} \quad E \elab \elpair{\ctx_E}{S_E.E'} \quad \ctx_L \cap \ctx_E = \emptyset}
	{ L = E \elab \elpair{\ctx_L \cup \ctx_E}{S_L; S_E; L' = E' }}\quad	
	
	\staterule{Elab Seq}
	{ S_1 \elab \elpair{\ctx_1}{S_1'} \quad 
		S_2 \elab \elpair{\ctx_2}{S_2'} \quad 
		\ctx_1 \cap \ctx_2 = \emptyset}
	{ S_1; S_2 \elab \elpair{\ctx_1 \cup \ctx_2}{S_1';S_2'}  } \qquad		

	\\[\GAP]\squad
	\staterule[~(where $\Gamma \vdash E : T$)]{Elab If}
	{   E \elab \elpair{\ctx_E}{S_E.E'} \quad
		S_1 \elab \elpair{\ctx_1}{S_1'} \quad 
		S_2 \elab \elpair{\ctx_2}{S_2'} \quad 
		\{g:T\} \cap \ctx_E \cap \ctx_{1} \cap \ctx_{2} = \emptyset }
	{ \kw{if}(E)\; S_1\; \kw{else}\; S_2 \elab \elpair{\{g:T\} \cup \ctx_E \cup \ctx_{1} \cup \ctx_{2}}{(S_E; g=E'; \kw{if}(g)\; S_1'\; \kw{else}\; S_2')}} \qquad	
	
	\\[\GAP]\squad
	\staterule[~(where $\Gamma \vdash E_1 : T_1$ and $\Gamma \vdash E_2 : T_2$)]{Elab For}
	{   E_1 \elab \elpair{\ctx_1}{S_1.E_1'} \quad 
		E_2 \elab \elpair{\ctx_2}{S_2.E_2'} \quad 
		S \elab[\ctx \cup \{x:(\kw{int}, \lev{data})\}] \elpair{\ctx_S}{S'} \\ 
		\ctx_V = v_{\Gamma}(\ctx_S, n) \quad 
		\{g_1:T_1, g_2:T_2, n:(\kw{int}, \lev{data})\} \cap \ctx_1 \cap \ctx_2 \cap \ctx_V = \emptyset}
	{ \kw{for}(x\;\kw{in}\;E_1:E_2)\;S \elab \elpair{\{g_1:T_1, g_2:T_2, n:(\kw{int}, \lev{data})\} \cup \ctx_1 \cup \ctx_2 \cup \ctx_V}{\\S_1; S_2; g_1=E_1'; g_2=E_2'; n = g_2 - g_1 + 1; \kw{for}(x\;\kw{in}\;g_1:g_2)\;v_S(x,\ctx_V,S')} } \qquad
\end{display}

In some cases when elaborating a \lstinline|for| loop, $\Gamma_S$ will not be empty (in other words, the body of the loop will declare new variables). Thus, as \ref{Elab For} shows, variables in $\Gamma_S$ are upgraded to an array, and then accessed by the index of the loop. We use the function $v_S$ (Definition~\ref{def:vector}) which takes a variable $x$, a typing environment $\Gamma$, and a statement $S$, and returns a statement $S'$, where any mention of a variable $x' \in \dom(\Gamma)$ is substituted with $x'[x]$. For example, consider the statement
\lstinline|for(i in 1:N){real +++model +++ x ~ normal(0,1); y[i] ~ normal(x,1);}| and an environment $\Gamma$, such that $\Gamma \vdash N:(\kw{int}, \lev{data})$. The body of the loop declares a new variable $x$, thus it elaborates to $\elpair{\Gamma_S}{S'}$, where
$\Gamma_S = \{x:(\kw{real}, \lev{model})\}$, and \lstinline|$S' =$ {x ~ normal(0,1); y[i] ~ normal(x,1);}|.

By \ref{Elab For}, 
$S \elab \langle \{\kw{g1}:(\kw{int}, \lev{data}), \kw{g2}:(\kw{int}, \lev{data})\}\cup\Gamma_V,~$\lstinline|for(i in g1:g2){$S''$}|$\rangle$
where:
\begin{align*}
\Gamma_V &= v_{\Gamma}(\Gamma_S, N) =  \{x:(\kw{real[N]}, \lev{model})\} \\
S'' &= v_S(i, \Gamma_V, S') = \kw{x[i] ~ normal(0,1); y[i] ~ normal(x[i],1);}
\end{align*}\vspace{-10pt}

Next, we state and prove type preservation of the elaboration relation. 
\begin{theorem}[Type preservation of $\elab$] \label{th:elab} For SlicStan statements $S$, SlicStan expressions $E$, and SlicStan function definitions $F$:
	~
	\begin{enumerate}
		\item If $ \Gamma \vdash E:(\tau,\ell)$ and $E \elab \elpair{\Gamma'}{S'.E'}$ then $\Gamma, \Gamma' \vdash S':\lev{data}$ and $\Gamma, \Gamma' \vdash E':(\tau, \ell)$.
		\item If $\Gamma \vdash S:\ell$ and $S \elab \elpair{\Gamma'}{S'} $ then $ \Gamma, \Gamma' \vdash S':\ell $
		\item If $F \elab \elpair{\Gamma'}{S'.\kw{ret}} $ then $ \Gamma, \Gamma' \vdash S':\lev{data}$
	\end{enumerate}
\end{theorem}
\begin{proof}
	By inductions on the size of the derivations of the judgments $E \elab \elpair{\Gamma'}{S'.E'}$, $S \elab \elpair{\Gamma'}{S'}$, and $F \elab \elpair{\Gamma'}{S'.\kw{ret}}$.	
\end{proof}

\subsection{Semantics of SlicStan} \label{ssec:slicsem}

We now show how SlicStan's type system allows us to specify the semantics of the probabilistic program as an unnormalised posterior density function. This shows how the semantics of SlicStan connects to that of Stan, and demonstrates that explicitly encoding the roles of program variables into the block syntax of the language is not needed.  

We specify the semantics --- the unnormalised density $\log p^*_{F_1, \dots, F_n, S}(\params \mid \data)$ --- in two steps.

\subsubsection{Semantics of (elaborated) SlicStan statements} \label{ssec:sem_statement}
Consider an elaborated SlicStan statement $S$ such that $\Gamma \vdash S : \lev{data}$. The semantics of $S$ is the function $\log p^*_{\Gamma \vdash S}$, such that
for any state $s \models \Gamma$:
$$\log p^*_{\Gamma \vdash S} (s) \deq 
s'[\kw{target}] \text{ if there is } s' \text{ such that } ((s,\kw{target} \mapsto 0), S) \Downarrow s'$$

\subsubsection{Semantics of SlicStan programs} \label{ssec:sem_prog_density}

Consider a well-formed SlicStan program $F_1, \dots, F_n, S$ and suppose that $S \elab[\emptyset] \elpair{\Gamma'}{S'}$.
(Observe that $\Gamma'$ and $S'$ are uniquely determined by $F_1, \dots, F_n, S$.)
Suppose also that:
\begin{itemize}
	\item $\Gamma_{\data}$ corresponds to \emph{data variables}, $\Gamma_{\data} = \{ x:\ell \in \Gamma' \mid \ell=\lev{data} \wedge x \notin \assset(S') \}$, and
	\item $\Gamma_{\theta}$ corresponds to \emph{model parameters}, $\Gamma_{\theta} = \{ x:\ell \in \Gamma' \mid \ell=\lev{model} \wedge x \notin \assset(S')\}$.
\end{itemize}

Similarly to Stan (\autoref{ssec:stansem}),
the semantics of a SlicStan program $S$ is the unnormalised log posterior density function $\log p^*_{F_1, \dots, F_n, S}$ on parameters $\params$ given data $\data$
(with $\params \models \Gamma_{\theta}$ and $\data \models \Gamma_{\data}$):
\begin{equation}\label{eq:slicposterior}
\log p^*_{F_1, \dots, F_n, S}\left( \params \mid \data \right) \deq \log p^*_{\Gamma' \vdash S'}(\params, \data)
\end{equation}

\subsection{Examples} \label{ssec:slicexamples}
Next, we give two examples of SlicStan programs, their elaborated versions, and their semantics in the form of an unnormalised log density function. Here, we specify the levels of variables in SlicStan programs explicitly. In \autoref{sec:demo} we describe how type inference can be implemented to infer optimal levels for program variables, thus making explicit declaration of levels unnecessary. 

\subsubsection{Simple Example} 
Consider a SlicStan program $\emptyset,S$
($\emptyset$ denotes no function definitions), where we simply model the distribution of a data array $\mathbf{y}$:
\begin{lstlisting}
	$S =\;$ real +++model+++ mu ~ normal(0, 1);
			 real +++model+++ sigma ~ normal(0, 1);
			 int +++data+++ N;
			 real +++data+++ y[N];
			 for(i in 1:N){ y[i] ~ normal(mu, sigma); }
\end{lstlisting}
We define the semantics of $S$ in three steps:
\begin{enumerate}
	\item Elaboration: $S \elab[\emptyset] \elpair{\Gamma'}{S'}$, where: \vspace{-20pt}
\begin{multicols}{2}
\begin{align*}
\Gamma' =& \kw{mu}:(\kw{real}, \lev{model}), \kw{sigma}:(\kw{real}, \lev{model}), \\
&\kw{y}:(\kw{real[N]}, \lev{model}), \kw{N}:(\kw{int}, \lev{data})
\end{align*}
\vspace{5pt}
\begin{lstlisting}
	$S' =$ mu ~ normal(0, 1);
	     sigma ~ normal(0, 1);
	     for(i in 1:N){ 
		     y[i] ~ normal(mu, sigma); }
\end{lstlisting}
\end{multicols}\vspace{-14pt}

\item Semantics of $S'$: For any state $s \models \Gamma'$, $\log p^*_{\Gamma' \vdash S'}(s) = s'[\kw{target}]$, where $(s, S') \Downarrow s'$. Thus: \vspace{-1pt}
$$\textstyle \log p^*_{\Gamma' \vdash S'} (s) = \log \mathcal{N}(\mu, 0, 1) + \log \mathcal{N}(\sigma, 0, 1) + \sum_{i=1}^{N}\log \mathcal{N}(y_i, \mu, \sigma)$$

\item Semantics of $S$: We derive $\Gamma_{\data} = \{ x:\ell \in \Gamma' \mid \ell=\lev{data} \wedge x \notin \assset(S') \} = \{N, y\}$, and  $\Gamma_{\theta} = \{ x:\ell \in \Gamma' \mid \ell=\lev{model} \wedge x \notin \assset(S')\} = \{\mu, \sigma\}$.
Therefore, the semantics of $S$ is the unnormalised density on the parameters $\mu$ and $\sigma$, given data $N$ and $y$: \vspace{-1pt}
$$\textstyle \log p_S^*(\mu, \sigma \mid y, N) = \log \mathcal{N}(\mu, 0, 1) + \log \mathcal{N}(\sigma, 0, 1) + \sum_{i=1}^{N}\log \mathcal{N}(y_i, \mu, \sigma)$$

\end{enumerate}

\subsubsection{User-defined Functions Example} \label{sssec:udex}

Next, we look at an example that includes a user-defined function. Here, the function \lstinline|my_normal| is a reparameterising function (\autoref{ssec:encaps}), that defines a Gaussian random variable, by scaling and shifting a standard Gaussian variable:
\begin{lstlisting}
	$S\,$ = real +++model+++ my_normal(real +++model+++ m, real +++model+++ s){
				 real +++model+++ x_std ~ normal(0, 1);
				 return m + x_std * s;
			 }	
			 real +++model+++ mu ~ normal(0, 1);
			 real +++model+++ sigma ~ normal(0, 1);	
			 int +++data+++ N;
			 real +++genquant+++ x[N]; 
			 for(i in 1:N) { x[i] = my_normal(mu, sigma); }
\end{lstlisting}
\vspace{-2pt}
\begin{enumerate}
	\item Elaboration: $S \elab[\emptyset] \elpair{\Gamma'}{S'}$, where: \vspace{-10pt}
\begin{multicols}{2}
\noindent
\begin{align*}
\Gamma' =& \kw{mu}:(\kw{real}, \lev{model}), \kw{sigma}:(\kw{real}, \lev{model}), \\
&\kw{m}:(\kw{real}, \lev{model}), \kw{s}:(\kw{real}, \lev{model}), \\ &\kw{x_std}:(\kw{real[N]}, \lev{model}), \\
&\kw{x}:(\kw{real[N]}, \lev{genquant}), \kw{N}:(\kw{int}, \lev{data})
\end{align*}
\vspace{1.5cm}
\begin{lstlisting}
	$S' =$ mu ~ normal(0, 1);
			sigma ~ normal(0, 1);		
			for(i in 1:N){
				m = mu; s = sigma;
				x_std[i] ~ normal(0, 1);
				x[i] = m + x_std[i] * s; }
\end{lstlisting}
\end{multicols}\vspace{-20pt}
\item Semantics of $S'$: Consider any $s \models \Gamma'$.
Then: \vspace{-1pt} 
$$\textstyle \log p^*_{\Gamma' \vdash S'} (s) = \log \mathcal{N}(\mu, 0, 1) + \log \mathcal{N}(\sigma, 0, 1) + \sum_{i=1}^{N}\log \mathcal{N}(x^{\mathrm{std}}_i, 0, 1)$$
\item  Semantics of $S$: We derive $\Gamma_{\data} = \{N\}$, and $\Gamma_{\theta} =  \{\mu, \sigma, \mathbf{x}^{\mathrm{std}}\}$.
The semantics of the program $S$ is the unnormalised density on the parameters $\mu$, $\sigma$, and $\mathbf{x}^{\mathrm{std}}$, given data $N$: \vspace{-1pt}
$$\textstyle \log p_S^*(\mu, \sigma, \mathbf{x}^{\mathrm{std}} \mid N) = \log \mathcal{N}(\mu, 0, 1) + \log \mathcal{N}(\sigma, 0, 1) + \sum_{i=1}^{N}\log \mathcal{N}(x^{\mathrm{std}}_i, 0, 1)$$

\end{enumerate}

\subsection{Difficulty of Specifying Direct Semantics Without Elaboration} \label{ssec:difficulty}
Specifying the direct semantics $\log p^*_{\emptyset \vdash S}(s)$, without an elaboration step, is not simple.
SlicStan's user-defined functions are flexible enough to allow new model parameters to be declared inside of the body of a function. Having some of the parameters declared this way means that it is not obvious what the complete set of parameters is, unless we elaborate the program. 

Consider the program from \autoref{sssec:udex}. Its semantics is ${\log p^*_{S} (\mu, \sigma, \mathbf{x}^{\mathrm{std}} \mid N)} = \log \mathcal{N}(\mu, 0, 1) + \log \mathcal{N}(\sigma, 0, 1) + \sum_{i=1}^{N}\log \mathcal{N}(x^{\mathrm{std}}_i, 0, 1)$. 
This differs from ${\log p^*_{S} (\mu, \sigma, x^{\mathrm{std}} \mid N))} = \log \mathcal{N}(\mu, 0, 1) + \log \mathcal{N}(\sigma, 0, 1) + N \times \log \mathcal{N}(x^{\mathrm{std}}, 0, 1)$, which would be the accumulated log density in case we do not unroll the \lstinline|my_normal| call, and instead implement direct semantics. In one case, the model has $N+2$ parameters: $\mu, \sigma, x^\mathrm{std}_1, \dots, x^\mathrm{std}_N$. In the other, the model has only 3 parameters: $\mu, \sigma, x^{\mathrm{std}}$. 

\section{Translation of SlicStan to Stan} \label{sec:translate-slicstan-to-stan}
Translating SlicStan to Stan happens in two steps: \emph{shredding} (\autoref{ssec:shred}) and \emph{transformation} (\autoref{ssec:trans}). In this section, we formalise these steps and show that the semantics, seen as an unnormalised log posterior density function on parameters given data, is preserved in the translation.

\subsection{Shredding} \label{ssec:shred}
The first step in translating an elaborated SlicStan program to Stan is the idea of \textit{shredding} (or \textit{slicing}) by level. SlicStan allows statements that assign to variables of different levels to be interleaved. Stan, on the other hand, requires all \lev{data} level statements to come first (in the \lstinline|data| and \lstinline|transformed data| blocks), then all \lev{model} level statements (in the \lstinline|parameters|, \lstinline|transformed parameters| and \lstinline|model| blocks), and finally, the \lev{genquant} level statements (in the \lstinline|generated quantities| block).
 
Therefore, we define the shredding relation $\shred$ on an elaborated SlicStan statement $S$ and triples of \textit{single-level statements} $\shredded$ (Definition~\ref{def:singlelev}). That is, $\shred$ shreds a statement into three elaborated SlicStan statements  $S_D$, $S_M$ and $S_Q$, where $S_D$ only assigns to variables of level $\lev{data}$, $S_M$ only assigns to variables of level $\lev{model}$, and $S_Q$ only assigns to variables of level $\lev{genquant}$. We formally state and prove this result in Lemma~\ref{lem:shredisleveled}.

\begin{display}[.3]{Shredding Relation}
	\clause{S \shred \shredded}{statement shredding} 
\end{display}

Currently, Stan can only assign to \lev{data} variables inside the \lstinline|transformed data| block, to \lev{model} variables inside the \lstinline|transformed parameters| block, and to generated quantities inside the \lstinline|generated quantities| block. Therefore, in Stan it is not possible to write an \lstinline|if| statement or a \lstinline|for| loop which assigns to variables of different levels inside its body. The \ref{Shred If} and \ref{Shred For} rules resolve this by copying the entire body of the \lstinline|if| statement or \lstinline|for| loop on each of the three levels. Notice that we restrict the \lstinline|if| and \lstinline|for| guards to be variables (as opposed to any expression), which we have ensured is the case after the elaboration step (\ref{Elab If} and \ref{Elab For}).

For example, consider the SlicStan program $S$, as defined below. It elaborates to $S'$ and $\Gamma'$, and it is then shredded to the single-level statements $\shredded$:
\vspace{-8pt}\begin{multicols}{3}
\begin{lstlisting}
$S=$ real +++data+++ d;
	  $\hspace{1pt}$real +++model+++ m;
	  $\hspace{1pt}$if(d > 0){
	     d = 1; 
	     m = 2;	 
	  }
\end{lstlisting}
	\vspace{9pt}
\begin{align*}\hspace{-14pt}
\Gamma' =\;\{&\kw{d}:(\kw{real}, \lev{data}),\\
&\kw{m}:(\kw{real}, \lev{model}),\\
&\kw{g}:(\kw{bool}, \lev{data})\}
\end{align*}
\vspace{-18pt}
\begin{lstlisting}
 $S' =$ g = (d > 0);
		  if(g){d=1; m=2;}
\end{lstlisting}
\begin{lstlisting}
 $S_D =$ g = (d > 0);
		  $\hspace{3pt}$if(g){d = 1;}
 $S_M =$ if(g){m = 2;}
 $S_Q =$ skip;
\end{lstlisting}
\end{multicols}
\begin{display}{Shredding Rules for Statements:}
	\squad	
	\staterule{Shred DataAssign}
	{\Gamma(L) = (\_,\lev{data})}
	{ L = E \shred (L = E, \kw{skip}, \kw{skip})}\quad\hquad
	
	\staterule{Shred ModelAssign}
	{ \Gamma(L) = (\_,\lev{model}) }
	{ L = E \shred \kw{skip}, L = E, \kw{skip}}\quad\hquad
	
	\staterule{Shred GenQuantAssign}
	{ \Gamma(L) = (\_,\lev{genquant})}
	{ L = E \shred \kw{skip}, \kw{skip}, L = E}\qquad	
	
	\\[\GAP]\squad
	\staterule{Shred Seq}
	{ S_1 \shred S_{D_1}, S_{M_1}, S_{Q_1} \quad 
		S_2 \shred S_{D_2}, S_{M_2}, S_{Q_2}}
	{ S_1; S_2 \shred (S_{D_1};S_{D_2}), (S_{M_1};S_{M_2}), (S_{Q_1};S_{Q_2})  } \qquad		

	\staterule{Shred Skip}
	{}
	{\kw{skip} \shred (\kw{skip}, \kw{skip}, \kw{skip})}\qquad

	\\[\GAP]\squad	
	\staterule{Shred If}
	{   S_1 \shred \shredded[1] \quad 
		S_2 \shred \shredded[2] \quad}
	{ \kw{if}(g)\; S_1\; \kw{else}\; S_2 \shred  
			(\kw{if}(g)\; S_{D_1}\; \kw{else}\; S_{D_2}),  
			(\kw{if}(g)\; S_{M_1}\; \kw{else}\; S_{M_2}), 
			(\kw{if}(g)\; S_{Q_1}\; \kw{else}\; S_{Q_2})} \qquad	
	
	\\[\GAP]\squad
	\staterule{Shred For}
	{   S \shred \shredded  }
	{ \kw{for}(x\;\kw{in}\;g_1:g_2)\;S \shred  
		(\kw{for}(x\;\kw{in}\;g_1:g_2)\;S_D),  
		(\kw{for}(x\;\kw{in}\;g_1:g_2)\;S_M), 
		(\kw{for}(x\;\kw{in}\;g_1:g_2)\;S_Q)} \qquad
\end{display}

In the rest of this section, we show that shredding a SlicStan program preserves its semantics (Theorem~\ref{th:shred}), in the sense that an elaborated program $S$ has the same meaning as the sequence of its shredded parts $S_D; S_M; S_Q$. We do so by:
\begin{enumerate}
	\item Proving that shredding produces \textit{single-level statements} (Definition~\ref{def:singlelev} and Lemma~\ref{lem:shredisleveled}).
	\item Defining a notion of \textit{statement equivalence} (Definition~\ref{def:equiv}) and specifying what conditions need to hold to change the order of two statements (Lemma~\ref{lem:reorder}).
	\item Showing how to extend the type system of SlicStan in order for the language to fulfil the criteria from (2) (Definition~\ref{def:shreddable}, Lemma~\ref{lem:commutativity}).
\end{enumerate} 

Intuitively, a single-level statement of level $\ell$ is one that updates only variables of level $\ell$. 
\begin{definition}[Single-level Statement $\Gamma \vdash \ell(S)$] \label{def:singlelev}
$S$ is a single-level statement of level $\ell$ with respect to $\Gamma$ (written $\Gamma \vdash \ell(S)$) if and only if,
$\Gamma \vdash S : \ell$ and $\forall x \in \assset(S)$ there is some $\tau$, s.t. $x:(\tau, \ell) \in \Gamma$.
\end{definition}

\begin{lemma}[Shredding produces single-level statements] \label{lem:shredisleveled}
	$$S \shred[\Gamma] \shredded \implies \singlelevelS{\lev{data}}{S_D} \wedge \singlelevelS{\lev{model}}{S_M} \wedge \singlelevelS{\lev{genquant}}{S_Q}$$
\end{lemma}

The core of proving Theorem~\ref{th:shred} is that if we take a statement $S$ that is well-typed in $\Gamma$, and reorder its building blocks according to $\shred$, the resulting statement $S'$ will be \textit{equivalent} to $S$.
\begin{definition}[Statement equivalence] \label{def:equiv}
	$S \eveq S' \deq \left( \forall s, s'. (s, S) \Downarrow s'  \iff (s, S') \Downarrow s' \right)$
\end{definition}

In the general case, to swap the order of executing $S_1$ and $S_2$, it is enough for each statement not to assign to a variable that the other statement reads or assigns to:

\begin{lemma}[Statement Reordering]
\label{lem:reorder}
	For statements $S_1$ and $S_2$ that are well-typed in $\Gamma$, if $\readset(S_1)\cap\assset(S_2) = \emptyset$, $\assset(S_1)\cap\readset(S_2) = \emptyset$, and $\assset(S_1)\cap\assset(S_2) = \emptyset$ then $S_1;S_2 \eveq S_2; S_1$.
\end{lemma}

Shredding produces single-level statements, therefore we only encounter reordering single-level statements of distinct levels. Thus, two of the conditions needed for reordering already hold. 
\begin{lemma}[] \label{lem:halfway}	If $~\Gamma \vdash \ell_1(S_1)$, $\Gamma \vdash \ell_2(S_2)$ and $\ell_1 < \ell_2$ then $\readset(S_1) \cap \assset(S_2) = \emptyset$ and $\assset(S_1) \cap \assset(S_2) = \emptyset$.
\end{lemma}

To reorder the sequence $S_2;S_1$ according to Lemma~\ref{lem:reorder}, we need to satisfy one more condition, which is 
$\readset(S_2) \cap \assset(S_1) = \emptyset$. We achieve this through the predicate $\shreddable$ in the \ref{Seq} typing rule.

One way to define $\shreddable(S_2, S_1)$ is so that it directly reflects this condition: $\shreddable(S_2, S_1) = \readset(S_2) \cap \assset(S_1)$. This corresponds to a form of a single-assignment system, where variables become immutable once they are read.

We adopt a more flexible strategy, where we enforce variables of level $\ell$ to become immutable only once they have been  \textit{read at a level higher than} $\ell$. We define:
\begin{itemize}
	\item $\readset_{\Gamma \vdash \ell}(S)$: the set of variables $x$ that are read at level $\ell$ in $S$. For example, if $y$ is of level $\ell$, then $x\in \readset_{\Gamma \vdash \ell}(y=x)$. (Definition~\ref{def:read_level_set}).
	\item $\assset_{\Gamma \vdash \ell}(S)$: the set of variables $x$ of level $\ell$ that have been assigned to in $S$ (Definition~\ref{def:write_level_set}).
\end{itemize}

Importantly, if $\Gamma \vdash \ell(S)$, then the sets $\readset_{\Gamma \vdash \ell}(S)$ and $\assset_{\Gamma \vdash \ell}(S)$ are the same as $\readset(S)$ and $\assset(S)$:
\begin{lemma} \label{lem:same_sets_when_singlelevel}
	If $~\Gamma \vdash \ell(S)$, then $\readset_{\Gamma \vdash \ell}(S) = \readset(S)$ and $\assset_{\Gamma \vdash \ell}(S) = \assset(S)$.
\end{lemma}
Finally, we give the formal definition of $\shreddable$:
\begin{definition}[Shreddable sequence] \label{def:shreddable} $\shreddable(S_1, S_2) \deq  \forall \ell_1,\ell_2. (\ell_2 < \ell_1) \implies \readset_{\Gamma \vdash \ell_1}(S_1) \cap \assset_{\Gamma \vdash \ell_2}(S_2) = \emptyset$
\end{definition}

\begin{lemma}[Commutativity of sequencing single-level statements]  \label{lem:commutativity} ~ 
		
If $~\singlelevelS{\ell_1}{S_1}$, $\singlelevelS{\ell_2}{S_2}$, $\Gamma \vdash S_2;S_1 : \lev{data}$ and $\ell_1 < \ell_2$ then $S_2; S_1; \eveq S_1; S_2;$
\end{lemma}

\begin{theorem}[Semantic Preservation of $\shred$] \label{th:shred} ~
	
If $~\Gamma \vdash S:\lev{data} $ and $ S \shred[\Gamma] \shredded $ then $ \log p^*_{\Gamma \vdash S}(s) = \log p^*_{\Gamma \vdash (S_D; S_M; S_Q)}(s)$, for all $s \models \Gamma$.
\end{theorem}

\begin{proof}
	Note that if $S \eveq S'$ then $\log p^*_{\Gamma \vdash S}(s) = \log p^*_{\Gamma \vdash S'}(s)$ for all states $s \models \Gamma$. 
	Semantic preservation then follows from proving the stronger result $\Gamma \vdash S:\lev{data} \wedge S \shred[\Gamma] \shredded \implies S \eveq (S_D; S_M; S_Q)$ by structural induction on the structure of $S$. 
	
	We give the full proof, together with proofs for Lemma~\ref{lem:shredisleveled}, \ref{lem:reorder}, \ref{lem:halfway} and \ref{lem:same_sets_when_singlelevel}, in \ref{ap:proofshred}.
\end{proof}

\subsection{Transformation} \label{ssec:trans}

The last step of translating SlicStan to Stan is \textit{transformation}. We formalise how a shredded SlicStan program $\elpair{\Gamma}{\shredded}$ transforms to a Stan program $P$, through the transformation relations:
\vspace{4pt} 
\begin{display}[.2]{Transformation Relations}
	\clause{\Gamma \ctxtostan \stan}{variable declarations transformation} \\
	\clause{S \dtostan \stan}{\lev{data} statement transformation} \\
	\clause{S \mtostan \stan}{\lev{model} statement transformation} \\
	\clause{S \qtostan \stan}{\lev{genquant} statement transformation} \\
	\clause{\elpair{\Gamma}{S} \tostan \stan}{top-level transformation}
\end{display}

Intuitively, a shredded program $\elpair{\Gamma}{\shredded}$ transforms to Stan in four steps:
\begin{enumerate}
	\item The declarations $\Gamma$ are split into blocks, depending on the level of variables and whether or not they have been assigned to inside of $S_D$, $S_M$ or $S_Q$.
	\item The \lev{data}-levelled statement $S_D$ becomes the body of the \lstinline|transformed data| block.
	\item The \lev{model}-levelled statement $S_M$ is split into the \lstinline|transformed parameters| and \lstinline|model| block, depending on whether or not substatements assign to the \lstinline|target| variable or not.
	\item The \lev{genquant}-levelled statement $S_Q$ becomes the body of the \lstinline|generated quantities| block.
\end{enumerate} 

This is formalised by the \ref{Trans Prog} rule below. The Stan program $P_1;P_2$ is the Stan programs $P_1$ and $P_2$ merged by composing together the statements in each program block (Definition~\ref{def:stanmerge}).

\vspace{12pt} 
\begin{display}{Top-level Transformation Rule}
	\squad\staterule{Trans Prog}
	{ S \shred[\Gamma] \shredded \quad \Gamma \ctxtostan[(S_D;S_M;S_Q)] \stan \quad S_D \dtostan \stan_D \quad S_M \mtostan \stan_M \quad S_Q \qtostan \stan_Q}
	{ \elpair{\Gamma}{S} \tostan \stan ; \stan_D; \stan_M; \stan_Q}
\end{display}
\begin{display}{Transformation Rules for Declarations:}
	\squad
	\staterule{Trans Data}
	{ \Gamma \ctxtostan P \quad  \quad x \notin \assset(S)}
	{ \Gamma, x:(\tau, \lev{data}) \ctxtostan \databl{x:\tau} ; P} \qquad
	
	\staterule{Trans TrData}
	{ \Gamma \ctxtostan P \quad  \quad x \in \assset(S)}
	{ \Gamma, x:(\tau, \lev{data}) \ctxtostan \trdatabl{x:\tau} ; P} 
	
	\\[\GAP]\squad
	\staterule{Trans Param}
	{ \Gamma \ctxtostan P \quad  \quad x \notin \assset(S)}
	{ \Gamma, x:(\tau, \lev{model}) \ctxtostan \paramsbl{x:\tau} ; P} \quad
	
	\staterule{Trans TrParam}
	{ \Gamma \ctxtostan P \quad  \quad x \in \assset(S)}
	{ \Gamma, x:(\tau, \lev{model}) \ctxtostan \trparamsbl{x:\tau} ; P} 
	
	\\[\GAP]\squad
	\staterule{Trans GenQuant}
	{ \Gamma \ctxtostan P}
	{ \Gamma, x:(\tau, \lev{genquant}) \ctxtostan \genquantbl{x:\tau} ; P} \qquad
	
	\staterule{Trans Empty}{}{\emptyset \tostan \emptyprog}
\end{display}

\vspace{-8pt}\begin{multicols}{2}
\begin{display}{Transformation Rule for Data Statements:}
	\squad	
	\staterule{Trans Data}
	{ }
	{ S_D \dtostan \trdatabl{S_D} }\qquad	
	\\[-8.5pt]	
\end{display}
\begin{display}{Transformation Rule for GenQuant Statements:}
	\squad	
	\staterule{Trans GenQuant}
	{ }
	{ S_Q \dtostan \genquantbl{S_Q} }\qquad	
\end{display}
\end{multicols}\vspace{-8pt}

The rules \ref{Trans ParamIf}, \ref{Trans ModelIf}, \ref{Trans ParamFor}, and \ref{Trans ModelFor} might produce a Stan program that does not compile in the current version of Stan. This is because Stan restricts the \lstinline|transformed parameters| block to only assign to transformed parameters, and the \lstinline|model| block to only assign to the \lstinline|target| variable. However, a \lstinline|for| loop, for example, can assign to both kinds of variables in its body:
\begin{lstlisting}
	for(i in 1:N){
		sigma[i] = pow(tau[i], -0.5);
		y[i] ~ normal(0, sigma[i]); }
\end{lstlisting}

To the best of our knowledge, this limitation is an implementational particularity of the current version of the Stan compiler, and does not have an effect on the semantics of the language.\footnotemark Therefore, we assume Core Stan to be a slightly more expressive version of Stan, that allows transformed parameters to be assigned in the \lstinline|model| block.

\footnotetext{Moreover, there is an ongoing discussion amongst Stan developers to merge the parameters, transformed parameters and model blocks in future versions of Stan \url{http://andrewgelman.com/2018/02/01/stan-feature-declare-distribute/}.}  

\begin{display}{Transformation Rules for Model Statements:}
	\squad
	\staterule{Trans ParamAssign}
	{ L \neq \kw{target}}
	{ L = E \mtostan \trparamsbl{L = E} }\quad	
	
	\staterule{Trans Model}
	{ }
	{ \kw{target} = E \mtostan \modelbl{\kw{target} = E} }\quad	
	
	\staterule{Trans ParamSeq}
	{ S_1 \mtostan \stan_1 \quad S_2 \tostan \stan_2}
	{ S_1;S_2 \mtostan \stan_1 ; \stan_2} \quad
	
	\\[\GAP]\squad	
	\staterule{Trans ParamIf}
	{ \kw{target} \notin \assset(S_1)\cup\assset(S_2) }
	{ \kw{if}(E)\; S_1\; \kw{else}\; S_2 \mtostan \trparamsbl{\kw{if}(E)\; S_1\; \kw{else}\; S_2}}\quad

	\staterule{Trans ModelIf}
	{ \kw{target} \in \assset(S_1)\cup\assset(S_2) }
	{ \kw{if}(E)\; S_1\; \kw{else}\; S_2 \mtostan \modelbl{\kw{if}(E)\; S_1\; \kw{else}\; S_2}}\quad
	
	\\[\GAP]\squad
	\staterule{Trans ParamFor}
	{ \kw{target} \notin \assset(S) }
	{ \kw{for}(x\;\kw{in}\;E_1:E_2)\;S \mtostan \trparamsbl{\kw{for}(x\;\kw{in}\;E_1:E_2)\;S}}\qquad

	\staterule{Trans ParamSkip}
	{}
	{\kw{skip} \mtostan \emptyprog}\qquad
	
	\\[\GAP]\squad
	\staterule{Trans ModelFor}
	{ \kw{target} \in \assset(S) }
	{ \kw{for}(x\;\kw{in}\;E_1:E_2)\;S \mtostan \modelbl{\kw{for}(x\;\kw{in}\;E_1:E_2)\;S}}\qquad

\end{display}

\begin{theorem}[Semantic Preservation of $\tostan$]\label{th:trans}
Consider a well-formed SlicStan program $F_1, \dots, F_n, S$, such that $S \elab[\emptyset] \elpair{\Gamma'}{S'}$.
Consider also a Core Stan program $P$, such that $\elpair{\Gamma'}{S'} \tostan P$.
Then for any $\params \models \{ (x: (\tau, \lev{data})) \in \Gamma' \mid x \notin \assset(S') \}$ and $\data \models \{ (x: (\tau, \lev{model}))  \in \Gamma' \mid x \notin \assset(S')\}$:
$$\log p^*_{F_1, \dots, F_n, S}(\params \mid \data) = \log p^*_{P}(\params \mid \data)$$
\end{theorem}

\begin{proof}
By rule induction on the derivation of $\elpair{\Gamma'}{S'} \tostan \stan$, and equation \ref{eq:slicposterior} from \autoref{ssec:sem_prog_density}.
\end{proof}

\section{Examples and Discussion} \label{sec:demo}

In this section, we demonstrate and discuss the functionality of SlicStan. We compare several Stan code examples, from Stan's Reference Manual \cite{StanManual} and Stan's GitHub repositories \cite{StanGitHub}, with their equivalent written in SlicStan, and analyse the differences. 
All examples presented in this section have been tested using a preliminary implementation of SlicStan, developed by \citet{SlicStanPPS, SlicStanStanCon}, although in this paper we use \lstinline|for| loops where the work makes use of a vectorised notation.

Firstly, we assume a type inference strategy for level types, which allows us to remove the explicit specification of levels from the language (\autoref{ssec:typeinference}). 
Next, we show that SlicStan allows the user to better follow the principle of locality --- related concepts can be kept closer together (\autoref{ssec:expert}). Secondly, we demonstrate the advantages of the more compositional syntax, when code refactoring is needed (\autoref{ssec:refactor}). The last comparison point shows the usage of more flexible user-defined functions, and points out a few limitations of SlicStan (\autoref{ssec:encaps}). More examples and a further discussion on the usability of the languages is presented in Appendix~\ref{ap:examples}.

\subsection{Type Inference} \label{ssec:typeinference}
Going back to \autoref{ssec:stansyntax}, and \autoref{tab:blocks}, we identify that different Stan blocks are executed a different number of times, which gives us another ordering on the level types: a performance ordering. 

Code associated with variables of level \lev{data} is executed only once, as a \textit{preprocessing} step before inference. Code associated with variables of level \lev{genquant} is executed once per sample, right after inference has completed, as these quantities can be \emph{generated} from the already obtained samples of the model parameters (in other words, this is a \textit{postprocessing} step). Finally, code associated with \lev{model} variables is needed at each step of the inference algorithm itself. In the case of HMC, this means such code is executed once per leapfrog step (many times per sample).   

Thus, there is a \emph{performance ordering} of level types: $\lev{data} \leq \lev{genquant} \leq \lev{model}$: it is cheaper for a variable to be \lev{data} than to be \lev{genquant}, and is cheaper for it to be \lev{genquant} than to be \lev{model}.
We can implement type inference following the rules from \autoref{ssec:slictyping}, to infer the level type of each variable in a SlicStan program, so that:
\begin{itemize}
	\item the hard constraint on the information flow direction $\lev{data} < \lev{model} < \lev{genquant}$ is enforced
	\item the choice of levels is optimised with respect to the ordering $\lev{data} \leq \lev{genquant} \leq \lev{model}$.
\end{itemize} 

We have implemented type inference for a preliminary version of SlicStan. In the rest of this section, we assume that no level type annotations are necessary in SlicStan, except for what the data of the probabilistic model is (specified using the derived form \lstinline|data|$~\tau~x;~S$), and that the optimal level type of each variable is inferred as part of the translation process.

\subsection{Locality} \label{ssec:expert}
With the first example, we demonstrate that SlicStan's blockless syntax makes it easier to follow good software development practices, such as declaring variables close to where they are used, and for writing out models that follow a \textit{generative story}. It abstracts away some of the specifics of the underlying inference algorithm, and thus writing optimised programs requires less mental effort. 

Consider an example adapted from \cite[p.~101]{StanManual}. We are interested in inferring the mean $\mu_y$ and variance $\sigma_y^2$ of the independent and identically distributed variables $\mathbf{y} \sim \normal(\mu_y, \sigma_y)$. The model parameters are $\mu_y$ (the mean of $\mathbf{y}$), and $\tau_y = 1 / \sigma_y^2$ (the precision of $\mathbf{y}$). 

Below, we show this example written in SlicStan (left) and Stan (right). 

\begin{minipage}[t][11.7cm][t]{\linewidth} 
\vspace{2pt} 
\begin{multicols}{2} \label{shred2}
	\centering
	\textbf{SlicStan}
	\vspace{-1pt}
	\begin{lstlisting}[numbers=left,numbersep=\numbdist,numberstyle=\tiny\color{\numbcolor},basicstyle=\small]	
	real alpha = 0.1;
	real beta = 0.1;
	real tau_y ~ gamma(alpha, beta);
	
	data real mu_mu;
	data real sigma_mu;
	real mu_y ~ normal(mu_mu, sigma_mu);
	
	real sigma_y = pow(tau_y, -0.5);
	real variance_y = pow(sigma_y, 2);
	
	data int N;
	data real[N] y;
	for(i in 1:N){ y[i] ~ normal(mu_y, sigma_y); }	
	\end{lstlisting}
	\vspace{4.5cm}
	\textbf{Stan}
	\vspace{-4pt}
	\begin{lstlisting}[numbers=left,numbersep=\numbdist,numberstyle=\tiny\color{\numbcolor},basicstyle=\small]
	data {
		real mu_mu;
		real sigma_mu;
		int N;
		real y[N];
	}
	transformed data {
		real alpha = 0.1;
		real beta = 0.1;
	}
	parameters {
		real mu_y;
		real tau_y;
	}
	transformed parameters {
		real sigma_y = pow(tau_y,-0.5);
	}
	model {
		tau_y ~ gamma(alpha, beta);
		mu_y ~ normal(mu_mu, sigma_mu);
		for(i in 1:N){ y[i] ~ normal(mu_y, sigma_y); }
	}
	generated quantities {
		real variance_y = pow(sigma_y,2);
	}
	
	\end{lstlisting}
\end{multicols}
\end{minipage}

The lack of blocks in SlicStan makes it more flexible in terms of order of statements. The code here is written to follow more closely than Stan the \textit{generative story}: we firstly define the prior distribution over parameters, and then specify how we believe data was generated from them. We also keep declarations of variables close to where they have been used: for example,  \lstinline{sigma_y} is defined right before it is used in the definition of \lstinline|variance_y|. This model can be expressed in SlicStan by using any order of the statements, provided that variables are not used before they are declared. In Stan this is not always possible and may result in closely related statements being located far away from each other.

With SlicStan there is no need to understand when different statements are executed in order to perform inference. The SlicStan code is translated to the hand-optimised Stan code, as specified by the manual, without any annotations from the user, apart from what the input data to the model is. In Stan, however, an inexperienced Stan programmer might have attempted to define the \lstinline|transformed data| variables \lstinline|alpha| and \lstinline|beta| in the \lstinline|data| block, which would result in a syntactic error. Even more subtly, they could have defined \lstinline|alpha|, \lstinline|beta| and \lstinline|variance_y| all in the \lstinline|transformed parameters| block, in which case the program will compile to a less efficient, semantically equivalent model.

\subsection{Code Refactoring} \label{ssec:refactor}
The next example is adapted from \cite[p.~202]{StanManual}, and shows how the absence of program blocks can lead to easier to refactor code. We start from a simple model, standard linear regression, and show what changes need to be made in both SlicStan and Stan, in order to change the model to account for measurement error. 
The initial model is a simple Bayesian linear regression with $N$ predictor points $\mathbf{x}$, and $N$ outcomes $\mathbf{y}$. It has 3 parameters --- the intercept $\alpha$, the slope $\beta$, and the amount of noise $\sigma$. In other words,
$\mathbf{y} \sim \normal(\alpha \mathbf{1} + \beta \mathbf{x}, \sigma I)$.

If we want to account for measurement noise, we need to introduce another vector of variables $\mathbf{x}_{meas}$, which represents the \emph{measured} predictors (as opposed to the true predictors $\mathbf{x}$). We postulate that the values of $\mathbf{x}_{meas}$ are noisy (with standard deviation $\tau$) versions of $\mathbf{x}$: 
$\mathbf{x}_{meas} \sim \normal(\mathbf{x}, \tau I)$.

The next page shows these two models written in SlicStan (left) and Stan (right). Ignoring all the lines/corrections in red gives us the initial regression model, the one \emph{not} accounting for measurement errors. The entire code, including the red corrections, gives us the second regression model, the one that \emph{does} account for measurement errors. Transitioning from model one to model two requires the following corrections: 

\begin{itemize}
	\item \textbf{In SlicStan:} 
	\begin{itemize}
		\item Delete the \lstinline|data| keyword for $\mathbf{x}$ (line 2).
		\item Introduce \emph{anywhere} in the program statements declaring the measurements $\mathbf{x}_{meas}$, their deviation $\tau$, the now parameter $\mathbf{x}$, and its hyperparameters $\mu_x, \sigma_x$ (lines 11--17).
	\end{itemize}
	\item \textbf{In Stan:}
	\begin{itemize}
		\item Move $\mathbf{x}$'s declaration from \lstinline|data| to \lstinline|parameters| (line 5 and line 9).
		\item Declare $\mathbf{x}_{meas}$ and $\tau$ in \lstinline|data| (lines 3--4).
		\item Declare $\mathbf{x}$'s hyperparameters $\mu_x$ and $\sigma_x$ in \lstinline|parameters| (lines 10--11).
		\item Add statements modelling $\mathbf{x}$ and $\mathbf{x}_{meas}$ in \lstinline|model| (lines 18--19).
	\end{itemize}
\end{itemize}

Performing the code refactoring requires the same amount of code in SlicStan and Stan. However, in SlicStan the changes interfere much less with the code already written. We can add statements extending the model anywhere (as long variables are declared before they are used). In Stan, on the other hand, we need to modify each block separately. This example demonstrates a successful step towards our aim of making Stan more compositional --- composing programs is easier in SlicStan.

\begin{minipage}[t][10.7cm][t]{\linewidth}
\vspace{-8pt} 
\begin{multicols}{2} \label{refactor}
	\centering
	\textbf{Regression in SlicStan}
	\vspace{1cm}
\begin{lstlisting}[numbers=left,numbersep=\numbdist,numberstyle=\tiny\color{\numbcolor},basicstyle=\small]	
	data int N;
	$\hbox{\stkw{data}}$ real[N] x; 
	%%real mu_x;%%
	%%real sigma_x;%%
	%%data real[N] x_meas;%%
	%%data real tau;%%
	
	real alpha ~ normal(0, 10);
	real beta ~ normal(0, 10);
	real sigma ~ cauchy(0, 5);
	data real[N] y;
	
	for(i in 1:N){
		%%x[i] ~ normal(mu_x, sigma_x);%%
		%%x_mean[i] ~ normal(x[i], tau);%%
		y[i] ~ normal(alpha + beta*x[i], sigma);
	}
\end{lstlisting}
\vspace{4cm}
\textbf{Regression in Stan}
\begin{lstlisting}[numbers=left,numbersep=\numbdist,numberstyle=\tiny\color{\numbcolor},basicstyle=\small]
	data {
		int N;	
		%%real[N] x_meas;%% 
		%%real tau;%% 	
		$\hbox{\stkw{real}}$$\hbox{\stlst{[N] x;}}$
		real[N] y; 
	}
	parameters {
		%%real[N] x;%% 
		%%real mu_x;%% 
		%%real sigma_x;%% 
		real alpha; 
		real beta;
		real sigma; 
	}
	model {
		alpha ~ normal(0, 10);
		beta ~ normal(0, 10);
		sigma ~ cauchy(0, 5);
		for(i in 1:N){
			%%x[i] ~ normal(mu_x, sigma_x);%%
			%%x_mean[i] ~ normal(x[i], tau);%%
			y[i] ~ normal(alpha + beta*x[i], sigma); }
	}
\end{lstlisting}
\end{multicols}
\end{minipage}
\vspace{-8pt}

\subsection{Code Reuse} \label{ssec:encaps}

Finally, we demonstrate the usage of more flexible functions in SlicStan, which allow for better code reuse, and therefore can lead to shorter, more readable code. 
In the introduction of this paper, we presented a transformation that is commonly used when specifying hierarchical model --- the \textit{non-centred parametrisation} of a normal variable. 
In brief, an MCMC sampler may have difficulties in exploring a posterior density well, if there exist strong non-linear dependencies between variables. In such cases, we can \emph{reparameterise} the model: we can express it in terms of different parameters, so that the original parameters can be recovered. In the case of a normal variable $x \sim \normal(\mu, \sigma)$, we define it as $x = \mu + \sigma x'$, where $x' \sim \normal(0,1)$. We explain in more detail the usage of the non-centered parametrisation in Appendix~\ref{ap:ncp}.

In this section, we show the  ``Eight Schools'' example \cite[p.~119]{Gelman2013}, which also uses non-centred parametrisation in order to improve performance. 
Eight schools study the effects of their SAT-V coaching program. The input data is the estimated effects $\mathbf{y}$ of the program for each of the eight schools, and their shared standard deviation $\boldsymbol{\sigma}$. The task is to specify a model that accounts for errors, by considering the observed effects to be noisy estimates of the \emph{true effects} $\boldsymbol{\theta}$. Assuming a Gaussian model for the effects and the noise, we have $\mathbf{y} \sim \normal(\boldsymbol{\theta}, \boldsymbol{\sigma} I)$ and $\boldsymbol{\theta} \sim \normal(\mu\mathbf{1}, \tau I)$.

Below is this model written in SlicStan (left) and Stan (right, adapted from Stan's GitHub repository \cite{StanGitHub}). In both cases, we use non-centred reparameterisation to improve performance: in Stan, the coaching effect for the \lstinline|i|\textsuperscript{th} school, \lstinline|theta[i]|, is declared as a transformed parameter obtained from the standard normal variable \lstinline|eta[i]|; in SlicStan, we can once again make use of the non-centred reparameterisation function \lstinline|my_normal|.

\begin{minipage}[t][8.7cm][t]{\linewidth}
\vspace{1pt}
\begin{multicols}{2} \label{schools}
	\centering
	\textbf{``Eight Schools'' in SlicStan}
	\begin{lstlisting}[numbers=left,numbersep=\numbdist,numberstyle=\tiny\color{\numbcolor},basicstyle=\small]
	real my_normal(real m, real v){
		real std ~ normal(0, 1); 
		return v * std + m;
	}
	
	data real[8] y;
	data real[8] sigma;
	real[8] theta;
	
	real mu;
	real tau;
	
	for (i in 1:8){
		theta[i] = my_normal(mu, tau);
		y[i] ~ normal(theta[i], sigma[i]);
	}	
	\end{lstlisting}
	\vspace{3cm}
	\textbf{``Eight Schools'' in Stan}
	\begin{lstlisting}[numbers=left,numbersep=\numbdist,numberstyle=\tiny\color{\numbcolor},basicstyle=\small]
	data {
		real y[8]; 
		real sigma[8]; 
	}
	parameters {
		real mu; 
		real tau; 
		real theta_std[8];
	}
	transformed parameters {
		real theta[8]; 
		for (j in 1:8){theta[j] = mu + tau * theta_std[i];}
	}
	model {
		for (j in 1:8){
			y[i] ~ normal(theta[i], sigma[i]);$\footnotemark$
			theta_std[i] ~ normal(0, 1); 	
		}
	}
	\end{lstlisting}
\end{multicols}
\end{minipage}

\footnotetext{In the full version of Stan these statements can be ``vectorised'' for efficiency, e.g. \lstinline|y ~ normal(theta,sigma);|}

\vspace{6pt}
One advantage of the original Stan code compared to SlicStan is the flexibility the user has to name all model parameters. In Stan, the auxiliary standard normal variables \lstinline|theta_std| are named by the user, while in SlicStan, the names of parameters defined inside of a function are automatically generated, and might not correspond to the names of transformed parameters of interest. 
All parameter names are important, as they are part of the output of the sampling algorithm, which is shown to the user. Even though in this case the auxiliary parameters were introduced solely for performance reasons, inspecting their values in Stan's output can be useful for debugging purposes.

\section{Related Work}
\label{sec:related}

There exists a range of probabilistic programming languages and systems. Stan's syntax is inspired by that of BUGS \cite{BUGS}, which uses Gibbs sampling to perform inference. Other languages include Anglican \cite{Anglican}, Church \cite{Church} and Venture \cite{Venture}, which focus on expressiveness of the language and range of supported models. They provide clean syntax and formalised semantics, but use less efficient, more general-purpose inference algorithms. 
The Infer.NET framework \cite{InferNET} uses an efficient inference algorithm called expectation propagation, but supports a limited range of models. Turing \cite{Turing} allows different inference techniques to be used for different sub-parts of the model, but requires the user to explicitly specify which inference algorithms to use as well as their hyperparameters

More recently, there has been the introduction of \textit{deep probabilistic programming}, in the form of Edward \cite{Edward, Edward2} and Pyro \cite{Pyro}, which focus on using deep learning techniques for probabilistic programming. Edward and Pyro are built on top of the deep learning libraries TensorFlow \cite{Tensorflow} and PyTorch \cite{PyTorch} respectively, and support a range of efficient inference algorithms. However, they lack the conciseness and formalism of some of the other systems, and it many cases require sophisticated understanding of inference. 

Other languages and systems include Hakaru \cite{Hakaru}, Figaro \cite{Figaro}, Fun \cite{Fun}, Greta \cite{Greta} and many others.

The rest of this section addresses related work done mostly within the programming languages community, which focuses on the semantics (\autoref{ssec:relatedsem}), static analysis (\autoref{ssec:relatedstatic}), and usability (\autoref{ssec:relatedusab}) of probabilistic programming languages. A more extensive overview of the connection between probabilistic programming and programming language research is given by \citet{GordonPP}.

\subsection{Formalisation of Probabilistic Programming Languages} \label{ssec:relatedsem}
There has been extensive work on the formalisation of probabilistic programming languages syntax and semantics. A widely accepted denotational semantics formalisation is that of \citet{Kozen81}. Other work includes a domain-theoretic semantics \cite{Plotkin},
measure-theoretic semantics \cite{Fun, ScibiorPPMonads, Toronto2015}, operational semantics \cite{Dal2012, MarcinICFP2016, Staton2016}, and more recently, categorical formalisation for higher-order probabilistic programs \cite{HeunenStaton}. Most previous work specifies either a measure-theoretic denotational semantics, or a sampling-based operational semantics. Some work \cite{Staton2016, Huang2016, HurNRS15} gives both denotational and operational semantics, and shows a correspondence between the two.  

The density-based semantics we specify for Stan and SlicStan is inspired by the work of \citet{HurNRS15}, who give an operational sampling-based semantics to the imperative language \textsc{Prob}. Intuitively, the difference between the two styles of operational semantics is:
\begin{itemize}
\item Operational \textit{density-based} semantics specifies how a program $S$ is executed to evaluate the (unnormalised) posterior density $p^*(\params \mid \data)$ at some specific point $\params$ of the parameter space.
\item Operational \textit{sampling-based} semantics specifies how a program $S$ is executed to evaluate the (unnormalised) probability $p^*(\mathbf{t})$ of the program generating some specific trace of samples $\mathbf{t}$.
\end{itemize}

Refer to Appendix~\ref{ap:sampling_semantics} for examples and further discussion of the differences between density-based and sampling-based semantics.

\subsection{Static Analysis for Probabilistic Programming Languages} \label{ssec:relatedstatic}
Work on static analysis for probabilistic programs includes several papers that focus on improving efficiency of inference. R2 \cite{R2} applies a semantics-preserving transformation to the probabilistic program, and then uses a modified version of the Metropolis--Hastings algorithm that exploits the structure of the model. This results in more efficient sampling, which can be further improved by \textit{slicing} the program to only contain parts relevant to estimating a target probability distribution \cite{Slicing}.
\citet{DataFlow2013} present a new inference algorithm that is based on data-flow analysis.
Hakaru \cite{Hakaru} is a relatively new probabilistic programming language embedded in Haskell, which performs automatic and semantic-preserving transformations on the program, in order to calculate conditional distributions and perform exact inference by computer algebra. 
The PSI system \cite{PSI2016} analyses probabilistic programs using a symbolic domain, and outputs a simplified expression representing the posterior distribution. 
The Julia-embedded language Gen \cite{Gen} uses type inference to automatically generate inference tactics for different sub-parts of the model. Similarly to Turing, the user then combines the generated tactics to build a model-specific inference algorithm.

With the exception of the work on slicing \cite{Slicing}, which is shown to work with Church and Infer.NET, each of the above systems either uses its own probabilistic language or the method is applicable only to a restricted type of models (for example boolean probabilistic programs). SlicStan is different in that it uses information flow analysis and type inference in order to self-optimise to Stan --- a scalable probabilistic programming language with a large user-base.  

\subsection{Usability of Probabilistic Programming Languages} \label{ssec:relatedusab} \label{ssec:usab}
This paper also relates to the line of work on usability of probabilistic programming languages.
\citet{Tabular} implement a schema-driven language, Tabular, which allows probabilistic programs to be written as annotated relational schemas. Fabular \cite{Fabular} extends this idea by incorporating syntax for hierarchical linear regression inspired by the lme4 package \cite{lmer}. BayesDB \cite{BayesDB} introduces BQL (Bayesian Query Language), which can be used to answer statistical questions about data, through SQL-like queries. 
Other work includes visualisation of probabilistic programs, in the form of graphical models \cite{BUGS, Vibes, GorinovaIDE}, and  
more data-driven approaches, such as synthesising programs from relational datasets \cite{PPSynth2015, PPSynth2017}.

\section{Conclusion} \label{sec:conc}

Probabilistic inference is a challenging task. As a consequence, existing probabilistic languages are forced to trade off efficiency of inference for range of supported models and usability. For example, Stan, an increasingly popular probabilistic programming language, makes efficient scalable automatic inference possible, but sacrifices compositionality of the language. 

This paper formalises the syntax of a core subset of Stan and gives its operational \textit{density-based semantics}; it introduces a new, compositional probabilistic programming language, SlicStan; and it gives a semantic-preserving procedure for translating SlicStan to Stan.
SlicStan adopts an \emph{information-flow type system}, that captures the taxonomy classes of variables of the probabilistic model. The classes can be inferred to 
automatically optimise the program for probabilistic inference. 
To the best of our knowledge, this work is the first formal treatment of the Stan language. 

We show that the use of static analysis and formal language treatment can facilitate efficient black-box probabilistic inference, and improve usability.  
Looking forward, it would be interesting to formalise the usage of pseudo-random generators inside of Stan. Variables in the \lstinline|generated$$quantities| block can be generated using pseudo-random number generators. In other words, the user can explicitly compose Hamiltonian Monte Carlo with forward (ancestral) sampling to improve inference performance. SlicStan can be extended to automatically determine what the most efficient way to sample a variable is,
which could significantly improve usability. 
Another interesting future direction would be to adapt the sampling-based semantics of \citet{HurNRS15} to SlicStan and establish how the density-based semantics of this paper corresponds to it. 
 
\newpage
\begin{acks}             
We thank Bob Carpenter and the Stan team for insightful discussions, and George Papamakarios for useful comments.
Maria Gorinova was supported by the EPSRC Centre for Doctoral Training in Data Science, funded by the UK Engineering and Physical Sciences Research Council (grant EP/L016427/1) and the University of Edinburgh.
\end{acks}

\bibliography{icfpbib} 

\cleardoublepage

\iftoggle{LONG}{ \appendix 
\section{Definitions and Proofs} \label{ap:proofs}
\subsection{Definitions}


\begin{definition}[Assigns-to set $\assset(S)$] \label{def:assset}
	\assset(S) is the set that contains the names of global variables that have been assigned to within the statement S. It is defined recursively as follows:	\vspace{-10pt}
	\begin{multicols}{2}\noindent
		$\assset(x[E_1]\dots[E_n] = E) = \{x\}$ \\
		$\assset(\{T x; S\}) = \assset(S) \setminus \{x\}$\\
		$\assset(S_1; S_2) = \assset(S_1) \cup \assset(S_2) $ \\
		$\assset(\kw{skip}) = \emptyset $ \\	
		$\assset(\kw{if}(E)\; S_1 \;\kw{else}\; S_2) = \assset(S_1)\cup \assset(S_2) $
		$\assset(\kw{for}(x\;\kw{in}\;E_1:E_2)\;S) = \assset(S) \setminus \{x\}$
	\end{multicols}\vspace{-12pt}
\end{definition} 

\begin{definition}[Reads set $\readset(S)$] \label{def:readset}
	\readset(S) is the set that contains the names of global variables that have been read within the statement S. It is defined recursively as follows:	\vspace{-10pt}
	\begin{multicols}{2}\noindent
		$\readset(x) = \{x\}$ \\
		$\readset(c) = \emptyset$ \\
		$\readset([E_1,\dots,E_n]) = \bigcup_{i=1}^n\readset(E_i)$ \\
		$\readset(E_1[E_2]) = \readset(E_1) \cup \readset(E_2) $ \\
		$\readset(f(E_1,\dots,E_n)) = \bigcup_{i=1}^n\readset(E_i)$\\
		$\readset(F(E_1,\dots,E_n)) = \bigcup_{i=1}^n\readset(E_i)$\\		
		$\readset(x[E_1]\dots[E_n] = E) = \bigcup_{i=1}^n\readset(E_i) \cup \readset(E)$ \\
		$\readset(\{T x; S\}) = \readset(S) \setminus \{x\}$\\
		$\readset(S_1; S_2) = \readset(S_1) \cup \readset(S_2) $ \\
		$\readset(\kw{skip}) = \emptyset $ \\	
		$\readset(\kw{if}(E)\; S_1 \;\kw{else}\; S_2) = \readset(E)\cup\readset(S_1)\cup \readset(S_2) $
		$\readset(\kw{for}(x\;\kw{in}\;E_1:E_2)\;S) = \readset(E_1) \cup \readset(E_2) \cup \readset(S) \setminus \{x\}$
	\end{multicols}\vspace{-12pt}
\end{definition} 

\begin{definition}[Type of expression $E$ in $\Gamma$] \label{def:lvaluetypes}
	$\Gamma(E)$ is the type of the expression $E$ with respect to $\Gamma$:\\
\noindent
$\Gamma(x) = (\tau, \ell)$ for $x: (\tau, \ell) \in \Gamma$ \\
$\Gamma(c) = (\kw{ty}(c), \lev{data})$ \\
$\Gamma([E_1,\dots,E_n]) = (\tau[], \ell_1 \sqcup \dots \sqcup \ell_n )$ if $\Gamma(E_i) = (\tau, \ell_i)$ for $i \in 1..n$, and $\ell' \sqcup \ell''$ denoting the least upper bound of $\ell'$ and $\ell''$, , and it is undefined otherwise.  \\
$\Gamma(E_1[E_2]) = (\tau, \ell \sqcup \ell')$ if $\Gamma(E_1)=(\tau[], \ell)$ and $\Gamma(E_2)=(\kw{int}, \ell')$, and it is undefined otherwise.
$\Gamma(f(E_1,\dots, E_n)) = (\tau, \ell)$ if $\Gamma(E_i) = (\tau_i, \ell_i)$ for $i \in 1..n$, and $f:(\tau_1, \ell_1), \dots, (\tau_n, \ell_n) \rightarrow (\tau, \ell)$.
\end{definition} 

The \textit{elaboration relation} transforms SlicStan statements and expressions to Core Stan statements and expressions. Thus, throughout this document, we use the terms ``elaborated statement'' and ``elaborated expression'' to mean a Core Stan statement and a Core Stan expression respectively. 

\begin{definition}[Vectorising functions $v_{\Gamma}, v_E, v_S$] \label{def:vector} ~
	\begin{enumerate}
		\item $v_{\Gamma}(\Gamma, n) \deq \{x:(\tau[n], \ell)\}_{x:(\tau, \ell) \in \Gamma}$, for any typing environment $\Gamma$.
		\item $v_E(x, \Gamma, E)$ is defined for a variable $x$, typing environment $\Gamma$, and an elaborated expression $E$:
		\vspace{-8pt}
		\begin{multicols}{2}
			$v_E(x, \Gamma, x') = \begin{cases} x'[x] &\text{if } x' \in \dom (\Gamma) \\ x' &\text{if } x' \notin \dom(\Gamma) \end{cases}$ \\
			$v_E(x, \Gamma, c) = c$ \\
			$v_E(x, \Gamma, [E_1,\dots, E_n]) = [v_E(E_1),\dots,v_E(E_n)]$ \\
			$v_E(x, \Gamma, E_1[E_2]) = v_E(E_1)[v_E(E_2)]$ \\
			$v_E(x, \Gamma, f(E_1,\dots,E_n)) = f(v_E(E_1),\dots,v_E(E_n))$
		\end{multicols} \vspace{-6pt}
		\item $v_{S}(x, \Gamma, S)$ is defined for a variable $x$, typing environment $\Gamma$, and an elaborated statement $S$: 
		\vspace{-8pt}
		\begin{multicols}{2}
			$v_{S}(x, \Gamma, L = E) = (v_E(L) = v_E(E))$ \\
			$v_{S}(x, \Gamma, S_1; S_2) = v_{S}(x, \Gamma, S_1);v_{S}(x, \Gamma, S_2) $ \\
			$v_{S}(x, \Gamma, \kw{if}(E)\; S_1\; \kw{else}\; S_2) =$ \\
			$\text{ }\qquad\kw{if}(v_E(E))\; v_{S}(S_1)\; \kw{else}\; v_{S}(S_2)$ \\
			$v_{S}(x, \Gamma, \kw{for}(x'\;\kw{in}\;E_1:E_2)\;S') =$ \\
			$\text{ }\qquad\kw{for}(x'\;\kw{in}\;v_E(E_1):v_E(E_2))\;v_{S}(S'))$ \\
			$v_{S}(\kw{skip}) = \kw{skip}$
		\end{multicols} \vspace{-6pt}
	\end{enumerate}
\end{definition}

\begin{definition}[$\readset_{\Gamma \vdash \ell}(S)$]\label{def:read_level_set}
$\readset_{\Gamma \vdash \ell}(S)$ is the set that contains the names of global variables that have been read at level $\ell$ with respect to $\Gamma$ within the statement S. It is defined recursively as follows:	\vspace{-10pt}
\begin{multicols}{2}\noindent
	$\readset_{\Gamma \vdash \ell}(x) = \emptyset $ \\
	$\readset_{\Gamma \vdash \ell}(x[E_1]\dots[E_n]) = \begin{cases}
	\bigcup_{i=1}^n\readset(E_i) & \text{if}~\Gamma(x) = \ell \\
	\emptyset & \text{otherwise}
	\end{cases}$ \\	
	$\readset_{\Gamma \vdash \ell}(L = E) = 
	\begin{cases}
	\readset_{\Gamma \vdash \ell}(L) \cup \readset(E) & \text{if}~\Gamma(L) = \ell \\
	\emptyset & \text{otherwise}
	\end{cases}$	 \\
	$\readset_{\Gamma \vdash \ell}(S_1; S_2) = \readset_{\Gamma \vdash \ell}(S_1) \cup \readset_{\Gamma \vdash \ell}(S_2) $ \\
	$\readset_{\Gamma \vdash \ell}(\kw{skip}) = \emptyset $ \\	
	For $S = \kw{if}(E)\; S_1 \;\kw{else}\; S_2$, and \\ $A =  \readset_{\Gamma \vdash \ell}(S_1) \cup \readset_{\Gamma \vdash \ell}(S_2)$: \\
	$\readset_{\Gamma \vdash \ell}(S) = 
	\begin{cases}
	\readset(E) \cup A & \text{if}~A \neq \emptyset \\
	\emptyset & \text{otherwise}
	\end{cases}$ \\
	For $S = \kw{for}(x\;\kw{in}\;E_1:E_2)\;S'$, and\\
	$A = \readset_{\Gamma \vdash \ell}(S')$\\
	$\readset_{\Gamma \vdash \ell}(S) = 
		\begin{cases}
		\readset(E_1) \cup \readset(E_2) \cup A & \text{if}~A \neq \emptyset \\
		\emptyset & \text{otherwise}
		\end{cases}$
\end{multicols}\vspace{-12pt}

\end{definition}

\begin{definition}[$\assset_{\Gamma \vdash \ell}(S)$]\label{def:write_level_set}
$\assset_{\Gamma \vdash \ell}(S) \deq \{x \in \assset(S) \mid \Gamma(x) = (\tau, \ell) \text{ for some } \tau\}$
\end{definition}

\begin{definition}{Merging Stan programs $\stan_1; \stan_2$} \label{def:stanmerge} \\
	Let $\stan_1$ and $\stan_2$	be two Stan programs, such that for $i = 1, 2$:
	\begin{align*}
	\stan_i = \squad\, &\databl{\,\Gamma_d^{(i)}\,} \\ &\kw{transformed data}\{\,\Gamma_d^{(i)'}\hquad S_d^{(i)'}\,\} \\ &\paramsbl{\,\Gamma_m^{(i)}\,} \\ &\kw{transformed parameters}\{\,\Gamma_m^{(i)'}\hquad S_m^{(i)'}\,\} \\ &\modelbl{\,\Gamma_m^{(i)''}\hquad S_m^{(i)''}\,} \\ &\kw{generated quantities}\{\,\Gamma_q^{(i)}\hquad S_q^{(i)}\,\}
	\end{align*}
	
	The sequence of $\stan_1$ and $\stan_2$, written $\stan_1; \stan_2$ is then defined as:
	\begin{align*}
	\stan_i = \squad\, &\databl{\,\Gamma_d^{(1)},\Gamma_d^{(2)}\,} \\ &\kw{transformed data}\{\,\Gamma_d^{(1)'},\Gamma_d^{(2)'}\hquad S_d^{(1)'};S_d^{(2)'}\,\} \\ &\paramsbl{\,\Gamma_m^{(1)},\Gamma_m^{(2)}\,} \\ &\kw{transformed parameters}\{\,\Gamma_m^{(1)'},\Gamma_m^{(2)'}\hquad S_m^{(1)'};S_m^{(2)'}\,\} \\ &\modelbl{\,\Gamma_m^{(1)''},\Gamma_m^{(2)''}\hquad S_m^{(1)''};S_m^{(2)''}\,} \\ &\kw{generated quantities}\{\,\Gamma_q^{(1)},\Gamma_q^{(2)}\hquad S_q^{(1)};S_q^{(2)}\,\}
	\end{align*}
	
\end{definition}

\subsection{Proof of Semantic Preservation of Shredding}  \label{ap:proofshred}

\begin{lemma}\label{lemma:dom}
	If $s \models \Gamma$ then $\dom(s)=\dom(\Gamma)$.
\end{lemma}
\begin{proof}
	By inspection of the definition of $s \models \Gamma$.
\end{proof}

\begin{lemma}\label{lemma:assset}
	If $S$ is well-typed in some environment $\Gamma$, $x \in \dom(s)$ and $(s, S) \Downarrow s'$ and $x \notin \assset(S)$ then $s(x)=s'(x)$.
\end{lemma}
\begin{proof}
	By induction on the derivation $(s, S) \Downarrow s'$.
\end{proof}

\begin{lemma} \label{lem:same_effect}
If $(s_1, S) \Downarrow s_1'$ and $(s_2, S) \Downarrow s_2'$ for some $s_1, s_1', s_2, s_2'$, and $s_1(x) = s_2(x)$ for all $x \in A$, where $A \supseteq \readset(S)$, then $s_1'(y) = s_2'(y)$ for all $y \in A\cup\assset(S)$.	
\end{lemma}
\begin{proof}
By induction on the structure of $S$.
\end{proof}

\begin{restate}{Lemma~\ref{lem:shredisleveled}(Shredding produces single-level statements)}
	$$S \shred[\Gamma] \shredded \implies \singlelevelS{\lev{data}}{S_D} \wedge \singlelevelS{\lev{model}}{S_M} \wedge \singlelevelS{\lev{genquant}}{S_Q}$$
\end{restate}
\begin{proof}
	By rule induction on the derivation of $S \shred \shredded$.
\end{proof}

\begin{restate}{Lemma~\ref{lem:reorder} (Statement Reordering)}
	For statements $S_1$ and $S_2$ that are well-typed in $\Gamma$, if $\readset(S_1)\cap\assset(S_2) = \emptyset$, $\assset(S_1)\cap\readset(S_2) = \emptyset$, and $\assset(S_1)\cap\assset(S_2) = \emptyset$ then $S_1;S_2 \eveq S_2; S_1$.
\end{restate}
\begin{proof}
	Let $R_i = \readset(S_i)$ and $W_i = \assset(S_i)$ for $i = 1,2$.
	Take any state $s$ and assume that $s \models \Gamma$. Suppose that
	$(s, S_1) \Downarrow s_1$,
	$(s, S_2) \Downarrow s_2$,
	$(s_1, S_2) \Downarrow s_{12}$, and
	$(s_2, S_1) \Downarrow s_{21}$.
	We want to prove that $s_{12}=s_{21}$.
	
	By Theorem~\ref{th:eval} and Lemma~\ref{lemma:dom}, we have $\dom(\Gamma)=\dom(s)=\dom(s_1)=\dom(s_2)=\dom(s_{12})=\dom(s_{21})$.
	Now, as $S_1$ writes only to $W_1$, by Lemma~\ref{lemma:assset}, we have that for all variables $x \in \dom(\Gamma)$:
	\begin{equation}\label{eq:1}
	x \notin W_1 \implies s(x) = s_1(x) \wedge s_2(x) = s_{21}(x)
	\end{equation}
	
	But $R_2$ and $W_1$ are disjoint, and $W_2$ and $W_1$ are disjoint, therefore $x \notin W_1$ for all $x \in R_2 \cup W_2$, and hence by Lemma~\ref{lemma:assset}:
	\begin{equation} \label{eq:2}
	x \in R_2 \cup W_2 \implies s(x) = s_1(x) \wedge s_2(x) = s_{21}(x)
	\end{equation}
	
	If two states are equal up to all variables in $\readset(S_2)$, then $S_2$ has the same effect on them (Lemma~\ref{lem:same_effect}).
	Combining this with (\ref{eq:2}) gives us:
	\begin{equation} \label{eq:4}
	x \in R_2 \cup W_2 \implies s_2(x) = s_{12}(x)
	\end{equation}
	
	Next, combining (\ref{eq:2}) and (\ref{eq:4}), gives us:
	\begin{equation} \label{eq:5}
	x \in R_2 \cup W_2 \implies s_2(x) = s_{12}(x) = s_{21}(x)
	\end{equation}
	
	Applying the same reasoning, but starting from $S_2$, we also obtain:
	\begin{align} \label{eq:6}
	x \notin W_2 \implies s(x)=s_2(x) \wedge s_1(x) = s_{12}(x) 
	\\ \label{eq:7}
	x \in R_1 \cup W_1 \implies s_1(x) = s_{21}(x) = s_{12}(x)
	\end{align}
	
	Finally, we have:
	\begin{itemize}
		\item $\forall x\in R_1\cap W_2. s_{12}(x)=s_{21}(x)$, as 	
		$R_1\cap W_2 = \emptyset$;
		
		\item $\forall x\in W_1\cap R_2. s_{12}(x)=s_{21}(x)$, as		
		$W_1\cap R_2 = \emptyset$;
		
		\item $\forall x\in W_1\cap W_2. s_{12}(x)=s_{21}(x)$, as
		$W_1\cap W_2 = \emptyset$;
		
		\item $\forall x\notin W_1 \cup W_2. s_{12}(x)=s_{21}(x)$, by combining (\ref{eq:1}) with (\ref{eq:6});
		
		\item $\forall x\in R_2\cup W_2. s_{12}(x)=s_{21}(x)$, by (\ref{eq:5});
		
		\item $\forall x\in R_1\cup W_1. s_{12}(x)=s_{21}(x)$, by (\ref{eq:7});
	\end{itemize}
	
	This covers all possible cases for $x \in \dom(\Gamma)$, therefore $\forall x \in \dom(\Gamma). s_{12}(x) = s_{21}(x)$. But $\dom(\Gamma) = \dom(s_{12}) = \dom(s_{21})$, thus $s_{12} = s_{21}$.
	 
\end{proof}

\begin{lemma}[] \label{lem:vars_in_exps}
	If $~\Gamma(x) = (\tau, \ell)$, $\Gamma \vdash E :(\tau', \ell')$, and $x \in \readset(E)$ then $\ell \leq \ell'$.
\end{lemma}
\begin{proof}
	By induction on the structure of the derivation $\Gamma \vdash E:(\tau', \ell')$. 	
\end{proof}

\begin{lemma}[] \label{lem:exp_in_singleleveled_s}
	If $~\Gamma(E) = (\tau, \ell)$, $\Gamma \vdash \ell'(S)$, and $E$ occurs in $S$, then $\ell \leq \ell'$.
\end{lemma}
\begin{proof}
	By induction on the structure of S. 	
\end{proof}

\begin{restate}{Lemma~\ref{lem:halfway}} ~
	
If $~\Gamma \vdash \ell_1(S_1)$, $\Gamma \vdash \ell_2(S_2)$ and $\ell_1 < \ell_2$ then $\readset(S_1) \cap \assset(S_2) = \emptyset$ and $\assset(S_1) \cap \assset(S_2) = \emptyset$.
\end{restate}
\begin{proof} Suppose $\Gamma \vdash \ell_1(S_1)$, $\Gamma \vdash \ell_2(S_2)$ and $\ell_1 < \ell_2$. Then, directly from Definition~\ref{def:singlelev}, we have that $\assset(S_1) \cap \assset(S_2) = \emptyset$. 
	Next, we prove by contradiction that $\readset(S_1) \cap \assset(S_2) = \emptyset$. 
	Suppose that for some $x$,  $x\in \readset(S_1)$ and $x \in \assset(S_2)$. From $x \in \assset(S_2)$ and $\Gamma \vdash \ell_2(S_2)$, we have $\Gamma(x) = (\tau, \ell_2)$ for some $\tau$. 
	From the definition of $\readset(S_1)$ and $x \in \readset(S_1)$, there must be an expression $E$ in $S_1$, such that $x \in \readset(E)$. 
	Suppose $\Gamma(E) = (\tau_E, \ell_E)$.
       Now, $\Gamma(x) = (\tau, \ell_2)$, $x \in \readset(E)$, and $\Gamma \vdash E:(\tau_E, \ell_E)$, therefore $\ell_2 \leq \ell_E$ (Lemma~\ref{lem:vars_in_exps}). 
	But $\Gamma \vdash \ell_1(S_1)$, and $E$ occurs in $S_1$, therefore $\ell_E \leq \ell_1$ (Lemma~\ref{lem:exp_in_singleleveled_s}).   
	
	We have $\ell_2 \leq \ell_E \leq \ell_1$. But $\ell_1 < \ell_2$, which is a contradiction.
\end{proof}

\begin{restate}{Lemma~\ref{lem:same_sets_when_singlelevel}} ~
	
	If $~\Gamma \vdash \ell(S)$, then $\readset_{\Gamma \vdash \ell}(S) = \readset(S)$ and $\assset_{\Gamma \vdash \ell}(S) = \assset(S)$.
\end{restate}
\begin{proof}
	Suppose $\Gamma \vdash \ell(S)$. The first equality, 
	$\readset_{\Gamma \vdash \ell}(S) = \readset(S)$, follows by structural induction on $\readset_{\Gamma \vdash \ell}(S)$.
	Furthermore, by definition of single-level statements, if $x \in \assset(S)$, then $\Gamma(x) = (\tau, \ell)$ for some $\tau$. Thus, by definition of $\assset_{\Gamma \vdash \ell}(S)$, we have that $\assset_{\Gamma \vdash \ell}(S) = \assset(S)$.
\end{proof}

\begin{lemma}[Associativity of $\eveq$] \label{lem:assoc}
$S_1;(S_2;S_3) \eveq (S_1;S_2);S_3$ for any $S_1$, $S_2$, and $S_3$.
\end{lemma}
\begin{proof}
	By expanding the definition of statement equivalence (Definition~\ref{def:equiv}).
\end{proof}

\begin{lemma}[Congruence of $\eveq$] \label{lem:congr}
	If $S_1 \eveq S_2$ then $S_1; S \eveq S_2; S$ and $S;S_1 \eveq S;S_2$ for any $S$.
\end{lemma}
\begin{proof}
	By expanding the definition of statement equivalence (Definition~\ref{def:equiv}).
\end{proof}

\begin{lemma} \label{lem:forloops}
	For any two states $s, s'$, and a statement $\kw{for}(x\;\kw{in}\;g_1:g_2)\;S$, suppose $n_1 = s(g_1)$, $n_2 = s(g_2)$ and $n_1 \leq n_2$. Then $(s, \kw{for}(x\;\kw{in}\;g_1:g_2)\;S) \Downarrow s'$ if and only if there exists $s_x$, such that $(s, x=n_1;S;x=(n_1+1);S\dots x=n_2;S) \Downarrow s_x$ and $s' = s_x[-x]$.
\end{lemma}
\begin{proof}
	By induction on $n = n_2 - n_1$.
\end{proof}

\begin{lemma} \label{lem:moreassignments}
	If $x \notin \assset{(S_1)}$ then $(x=n; S_1;S_2) \eveq (x=n; S_1; x=n; S_2)$. 
\end{lemma}
\begin{proof}
By expanding the definition of statement equivalence (Definition~\ref{def:equiv}).
\end{proof}

\begin{lemma} \label{lem:loopreorder}
	If $\,\Gamma \vdash \ell_1(S_1)$, $\Gamma \vdash \ell_2(S_2)$, $\ell_1 < \ell_2$, $\Gamma \vdash S_2;S_1 : \lev{data}$, and $x \notin \assset{(S_2;S_1)}$ then $x=i;S_2; x=j;S_1 \eveq x=j;S_1;x=i;S_2;x=j$ for all integers $i$ and $j$.
\end{lemma}
\begin{proof}
By expanding the definition of statement equivalence (Definition~\ref{def:equiv}), and using Lemma~\ref{lem:congr} and Lemma~\ref{lem:moreassignments}.
\end{proof}

\begin{restate}{Lemma~\ref{lem:commutativity}(Commutativity of sequencing single-level statements)} ~ 	

	If $~\singlelevelS{\ell_1}{S_1}$, $\singlelevelS{\ell_2}{S_2}$, $\Gamma \vdash S_2;S_1 : \lev{data}$ and $\ell_1 < \ell_2$ then $S_2; S_1; \eveq S_1; S_2;$
\end{restate}
\begin{proof}
	Since $\Gamma \vdash S_2;S_1 : \lev{data}$, and $\lev{data} \leq \ell_1 < \ell_2$, it must be that $\readset_{\Gamma\vdash\ell_2}(S_2) \cap \assset_{\Gamma\vdash\ell_1}(S_1) = \emptyset$. By Lemma~\ref{lem:same_sets_when_singlelevel}, $\readset_{\Gamma\vdash\ell_2}(S_2) = \readset(S_2)$ and $\assset_{\Gamma\vdash\ell_1}(S_1) = \assset(S_1)$, as $S_1$ and $S_2$ are single-level of level $\ell_1$ and $\ell_2$ respectively. Therefore, $\readset(S_2) \cap \assset(S_1) = \emptyset$. 
	From $\Gamma \vdash \ell_1(S_1)$, $\Gamma \vdash \ell_2(S_2)$, $\ell_1 < \ell_2$ and by Lemma~\ref{lem:halfway}, we have $\readset(S_1) \cap \assset(S_2) = \emptyset$ and $\assset(S_1) \cap \assset(S_2) = \emptyset$. 
	Therefore, by Lemma~\ref{lem:reorder},  $S_2; S_1 \eveq S_1; S_2$.
\end{proof}

\begin{restate}{Theorem~\ref{th:shred} (Semantic Preservation of $\shred$)} ~ 
	
If $~\Gamma \vdash S:\lev{data} $ and $ S \shred[\Gamma] \shredded $ then $ \log p^*_{\Gamma \vdash S}(s) = \log p^*_{\Gamma \vdash (S_D; S_M; S_Q)}(s)$, for all $s \models \Gamma$.
\end{restate}

\begin{proof}
	Note that if $S \eveq S'$ then $\log p^*_{\Gamma \vdash S}(s) = \log p^*_{\Gamma \vdash S'}(s)$ for all states $s \models \Gamma$. 
	
	Semantic preservation then follows from proving the stronger result $$S \shred \shredded \implies \Gamma \vdash S:\lev{data} \implies S \eveq (S_D; S_M; S_Q)$$ by rule induction on $S \shred \shredded$.
	Let	$$\Phi(S, S_D, S_M, S_Q) \deq
	 S \shred \shredded \implies \Gamma \vdash S:\lev{data} \implies S \eveq S_D;S_M;S_Q$$

	Take any $S$, $S_D, S_M, S_Q$, and assume $S \shred \shredded$ and $\Gamma \vdash S: \lev{data}$.
	\begin{itemize}
		\item[\ref{Shred Skip}] 
		\begin{math}
		\copyrule{}{\kw{skip} \shred (\kw{skip}, \kw{skip}, \kw{skip})}
		\end{math}		
		
		For all $s$, $(s, \kw{skip}) \Downarrow s$, and also $(s, \kw{skip}; \kw{skip}; \kw{skip}) \Downarrow s$. Thus $\kw{skip} \eveq \kw{skip};\kw{skip};\kw{skip}$.
		
		\item[\ref{Shred DataAssign}] 
		\begin{math}
		\copyrule
		{\Gamma(L) = (\_,\lev{data})}
		{ L = E \shred (L = E, \kw{skip}, \kw{skip})}
		\end{math} \\
		For any state $s$, if  $(s, L=E) \Downarrow s'$, then $(s, L=E;\kw{skip};\kw{skip}) \Downarrow s'$, and vice versa. Thus, $\Phi(L=E, L=E, \kw{skip}, \kw{skip})$ holds.

		\item[\ref{Shred ModelAssign}] 
		\begin{math}
		\copyrule
		{\Gamma(L) = (\_,\lev{model})}
		{ L = E \shred (\kw{skip}, L = E, \kw{skip})}
		\end{math} \\
		For any state $s$, if  $(s, L=E) \Downarrow s'$, then $(s, \kw{skip};L=E;\kw{skip}) \Downarrow s'$, and vice versa. Thus, $\Phi(L=E, \kw{skip}, L=E, \kw{skip})$ holds.
		
		\item[\ref{Shred GenQuantAssign}] 
		\begin{math}
		\copyrule
		{\Gamma(L) = (\_,\lev{genquant})}
		{ L = E \shred (\kw{skip}, \kw{skip}, L = E)}
		\end{math} \\
		For any state $s$, if  $(s, L=E) \Downarrow s'$, then $(s, \kw{skip};\kw{skip};L=E) \Downarrow s'$, and vice versa. Thus, $\Phi(L=E, \kw{skip}, \kw{skip}, L=E)$ holds.
		
		\item[\ref{Shred Seq}] $S = (S_1; S_2)$. Suppose that $\Phi(\Gamma, S_1)$ and $\Phi(\Gamma, S_2)$, and assume $S_1 \shred \shredded[1]$, $S_2 \shred \shredded[2]$. Thus, $S_1 \eveq \shreddedseqsingle[1]$ and $S_2 \eveq \shreddedseqsingle[2]$. Now:		
		\begin{align*}
		&\eveq (\shreddedseqsingle[1]);(\shreddedseqsingle[2])  && \text{from } \Phi(\Gamma, S_1) \text{ and } \Phi(\Gamma, S_2)\\
		&\eveq (S_{D_1};S_{M_1});S_{D_2};S_{Q_1};(S_{M_2};S_{Q_2})   && \text{by Lemmas~\ref{lem:shredisleveled}, \ref{lem:commutativity}, \ref{lem:assoc} and~\ref{lem:congr}} \\
		&\eveq (S_{D_1});S_{D_2};S_{M_1};(S_{Q_1};S_{M_2};S_{Q_2})   && \text{by Lemmas~\ref{lem:shredisleveled}, \ref{lem:commutativity}, \ref{lem:assoc} and~\ref{lem:congr}} \\
		&\eveq (S_{D_1};S_{D_2};S_{M_1});S_{M_2};S_{Q_1};(S_{Q_2})   && \text{by Lemmas~\ref{lem:shredisleveled}, \ref{lem:commutativity}, \ref{lem:assoc} and~\ref{lem:congr}} \\
		&\eveq (S_{D_1};S_{D_2});(S_{M_1};S_{M_2});(S_{Q_1};S_{Q_2})   && 
		\end{align*}
		
		As by \ref{Shred Seq}, $S \shred (S_{D_1};S_{D_2}), (S_{M_1};S_{M_2}), (S_{Q_1};S_{Q_2})$, it follows that $\Phi(\Gamma, S_1;S_2)$.
		
		\item[\ref{Shred If}] $S = (\kw{if}(g)\; S_1\; \kw{else}\; S_2)$. Suppose that $\Phi(\Gamma, S_1)$ and $\Phi(\Gamma, S_2)$, and assume $S_1 \shred \shredded[1]$, $S_2 \shred \shredded[2]$. Thus, by \ref{Shred If}: 
		$$\kw{if}(g)\; S_1\; \kw{else}\; S_2 \shred  
		(\kw{if}(g)\; S_{D_1}\; \kw{else}\; S_{D_2}),  
		(\kw{if}(g)\; S_{M_1}\; \kw{else}\; S_{M_2}), 
		(\kw{if}(g)\; S_{Q_1}\; \kw{else}\; S_{Q_2})$$
		
		Now take any two states $s$ and $s'$, such that $s \models \Gamma$ and $(s, S) \Downarrow s'$. Given that $\Gamma \vdash S : \lev{data}$, $\Gamma(g) = (\kw{bool}, \_)$ by \ref{If}. Therefore $s(g) = \kw{true}$ or $s(g) = \kw{false}$.  
		\begin{enumerate}
			\item If $s(g) = \kw{true}$, it must be that $(s, S_1) \Downarrow s'$. Then:
			\begin{align*}			
			& (s, \kw{if}(g)\; S_1\; \kw{else}\; S_2) \Downarrow s' && \text{by } \ref{Eval IfTrue} \\
			& (s, (\shreddedseqsingle[1])) \Downarrow s' && \text{from } \Phi(\Gamma, S_1)\\
			& (s, (\kw{if}(g)\; S_{D_1}\,\kw{else}\; S_{D_2};
			\kw{if}(g)\; S_{M_1}\,\kw{else}\; S_{M_2};
			\kw{if}(g)\; S_{Q_1}\,\kw{else}\; S_{Q_2})) \Downarrow s' && 3 \times \ref{Eval IfTrue} \\
			\end{align*}			
			\item If $s(g) = \kw{false}$, it must be that $(s, S_2) \Downarrow s'$. Then:
			\begin{align*}
			& (s, \kw{if}(g)\; S_1\; \kw{else}\; S_2) \Downarrow s' && \text{by } \ref{Eval IfFalse} \\
			& (s, (\shreddedseqsingle[2])) \Downarrow s' && \text{from } \Phi(\Gamma, S_2)\\
			& (s, (\kw{if}(g)\; S_{D_1}\,\kw{else}\; S_{D_2};
			\kw{if}(g)\; S_{M_1}\,\kw{else}\; S_{M_2};
			\kw{if}(g)\; S_{Q_1}\,\kw{else}\; S_{Q_2})) \Downarrow s' && 3 \times \ref{Eval IfFalse} \\
			\end{align*}
		\end{enumerate}	
		Thus, $(s, \kw{if}(g)\; S_1\; \kw{else}\; S_2)) \Downarrow s' \implies (s, (\kw{if}(g)\; S_{D_1}\; \kw{else}\; S_{D_2};  
		\kw{if}(g)\; S_{M_1}\; \kw{else}\; S_{M_2}; 
		\kw{if}(g)\; S_{Q_1}\; \kw{else}\; S_{Q_2})) \Downarrow s'$. For the implication in the opposite direction:
		\begin{enumerate}
			\item If $s(g) = \kw{true}$, take any $s'$ such that $(s, (\kw{if}(g)\; S_{D_1}\,\kw{else}\; S_{D_2};
			\kw{if}(g)\; S_{M_1}\,\kw{else}\; S_{M_2};
			\kw{if}(g)\; S_{Q_1}\,\kw{else}\; S_{Q_2})) \Downarrow s'$. Then:
			\begin{align*}			
			& (s, (\shreddedseqsingle[1])) \Downarrow s' && \text{by } 3 \times \ref{Eval IfTrue} \\
			& (s, S_1) \Downarrow s' && \text{from } \Phi(\Gamma, S_1)\\
			& (s, \kw{if}(g)\; S_1\; \kw{else}\; S_2) \Downarrow s' && \text{by } \ref{Eval IfTrue}
			\end{align*}			
			\item If $s(g) = \kw{false}$, take any $s'$ such that $(s, (\kw{if}(g)\; S_{D_1}\,\kw{else}\; S_{D_2};
			\kw{if}(g)\; S_{M_1}\,\kw{else}\; S_{M_2};
			\kw{if}(g)\; S_{Q_1}\,\kw{else}\; S_{Q_2})) \Downarrow s'$. Then:
			\begin{align*}
			& (s, (\shreddedseqsingle[2])) \Downarrow s' && \text{by } 3 \times \ref{Eval IfFalse} \\
			& (s, S_2) \Downarrow s' && \text{from } \Phi(\Gamma, S_2)\\
			& (s, \kw{if}(g)\; S_1\; \kw{else}\; S_2) \Downarrow s' && \text{by } \ref{Eval IfFalse}
			\end{align*}
		\end{enumerate}	
		Thus, $(s, (\kw{if}(g)\; S_{D_1}\; \kw{else}\; S_{D_2};  
		\kw{if}(g)\; S_{M_1}\; \kw{else}\; S_{M_2}; 
		\kw{if}(g)\; S_{Q_1}\; \kw{else}\; S_{Q_2})) \Downarrow s' \implies (s, \kw{if}(g)\; S_1\; \kw{else}\; S_2)) \Downarrow s'$.
                Therefore, $\kw{if}(g)\; S_1\; \kw{else}\; S_2 \eveq  (\kw{if}(g)\; S_{D_1}\; \kw{else}\; S_{D_2}); (\kw{if}(g)\; S_{M_1}\; \kw{else}\; S_{M_2}); (\kw{if}(g)\; S_{Q_1}\; \kw{else}\; S_{Q_2})$, and $\Phi(\Gamma, \kw{if}(g)\; S_1\; \kw{else}\; S_2)$.
		
		\item[\ref{Shred For}] Suppose $S = (\kw{for}(x\;\kw{in}\;g_1:g_2)\;S') = LHS$. Then: \\
		\begin{math}
		\copyrule
		{   S' \shred (S_D', S_M', S_Q')  }
		{LHS \shred  (\kw{for}(x\;\kw{in}\;g_1:g_2)\;S_D'),  
			(\kw{for}(x\;\kw{in}\;g_1:g_2)\;S_M'), 
			(\kw{for}(x\;\kw{in}\;g_1:g_2)\;S_Q')}
		\end{math}\\
		Take $RHS = (\kw{for}(x\;\kw{in}\;g_1:g_2)\;S_D');  
		(\kw{for}(x\;\kw{in}\;g_1:g_2)\;S_M'); 
		(\kw{for}(x\;\kw{in}\;g_1:g_2)\;S_Q')$
		
		We must show $\Phi(S', S_D', S_M', S_Q') \implies LHS \eveq RHS$. 
		
		Assume $\Phi(S', S_D', S_M', S_Q')$, and consider $s, s'$, such that $(s, LHS) \Downarrow s'$ to show $(s, RHS) \Downarrow s'$.
	
		Suppose $n_1 = s(g_1)$ and $n_2 = s(g_2)$. Then either $n_1 \leq n_2$ or $n_1 > n_2$:
		\begin{enumerate}
		\item Case $n_1 \leq n_2$.
		
		Using Lemma~\ref{lem:forloops}, we have that there exists $s_x$, such that $(s, x=n_1;S';x=(n_1+1);S'\dots x=n_2;S') \Downarrow s_x$ and $s' = s_x[-x]$.
		
		As $\Phi(S', S_D', S_M', S_Q')$, $S' \shred (S_D', S_M', S_Q')$ and $\Gamma \vdash S':\lev{data}$ (by \ref{For}), we have that $S' \eveq S_D';S_M';S_Q'$. Combined with the result from above and using Lemma~\ref{lem:congr}, this gives us $(s, x=n_1;S_D';S_M';S_Q';x=(n_1+1);S_D';S_M';S_Q'\dots x=n_2;S_D';S_M';S_Q') \Downarrow s_x$ and $s' = s_x[-x]$.
		
		By Lemma~\ref{lem:moreassignments}, we then have $(s, x=n_1;S_D';x=n_1;S_M';x=n_1;S_Q';\dots x=n_2;S_D';x=n_2;S_M';x=n_2;S_Q') \Downarrow s_x$ and $s' = s_x[-x]$.
		
		By Lemma~\ref{lem:shredisleveled} $\Gamma \vdash  \lev{data}(S_D')$, $\Gamma \vdash  \lev{model}(S_M')$, and $\Gamma \vdash  \lev{genquant}(S_Q')$. Thus, we apply Lemma~\ref{lem:loopreorder} to get $(s, x=n_1;S_D'; \dots x=n_2; S_D'; x=n_1;S_M'; \dots x=n_2; S_M'; x=n_1;S_Q';\dots x=n_2;S_Q') \Downarrow s_x$ and $s' = s_x[-x]$.
		
		By applying \ref{Eval Seq}, we split this into:
		\begin{itemize}
			\item $(s, x=n_1;S_D'; \dots x=n_2; S_D') \Downarrow s_{xd}$
			\item $(s_{xd}, x=n_1;S_M'; \dots x=n_2; S_M') \Downarrow s_{xm}$
			\item $(s_{xm},x=n_1;S_Q';\dots x=n_2;S_Q') \Downarrow s_x$
		\end{itemize}
		
		For some $s_{xd}$ and $s_{xm}$. By taking $s_d = s_{xd}[-x]$ and $s_m = s_{xm}[-x]$, and applying \ref{lem:forloops}, we get:
		\begin{itemize}
			\item $(s, \kw{for}(x\;\kw{in}\;g_1:g_2)\;S_D') \Downarrow s_d$
			\item $(s_{xd}, \kw{for}(x\;\kw{in}\;g_1:g_2)\;S_M') \Downarrow s_m$
			\item $(s_{xm}, \kw{for}(x\;\kw{in}\;g_1:g_2)\;S_Q') \Downarrow s'$
		\end{itemize}
		
		As $x \notin \readset(\kw{for}(x\;\kw{in}\;g_1:g_2)\;S_M)$, $x \notin \readset(\kw{for}(x\;\kw{in}\;g_1:g_2)\;S_Q')$, $s_d = s_{xd}[-x]$ and $s_m = s_{xm}[-x]$, we also have:
		\begin{itemize}
			\item $(s, \kw{for}(x\;\kw{in}\;g_1:g_2)\;S_D') \Downarrow s_d$
			\item $(s_d, \kw{for}(x\;\kw{in}\;g_1:g_2)\;S_M') \Downarrow s_m$
			\item $(s_m, \kw{for}(x\;\kw{in}\;g_1:g_2)\;S_Q') \Downarrow s'$
		\end{itemize}
		
		Therefore, by \ref{Eval Seq}: $(s, (\kw{for}(x\;\kw{in}\;g_1:g_2)\;S_D');  
		(\kw{for}(x\;\kw{in}\;g_1:g_2)\;S_M'); 
		(\kw{for}(x\;\kw{in}\;g_1:g_2)\;S_Q')) \Downarrow s'$, so $(s, LHS) \Downarrow s' \implies (s, RHS) \Downarrow s'$.	
		
		\item Case $n_1 > n_2$. By \ref{Eval ForFalse} $s' = s$. Also, $(s, \kw{for}(x\;\kw{in}\;g_1:g_2)\;S_D') \Downarrow s$, $(s, \kw{for}(x\;\kw{in}\;g_1:g_2)\;S_M') \Downarrow s$ and $(s, \kw{for}(x\;\kw{in}\;g_1:g_2)\;S_Q') \Downarrow s$. So by \ref{Eval Seq}, $(s, RHS) \Downarrow s$, and thus $(s, LHS) \Downarrow s' \implies (s, RHS) \Downarrow s'$.
		\end{enumerate}
		
		By assuming instead $s$ and $s'$, such that $(s, RHS) \Downarrow s'$, and reversing this reasoning, we also obtain $(s, RHS) \Downarrow s' \implies (s, LHS) \Downarrow s'$.
		
		Therefore $(s, LHS) \Downarrow s' \iff (s, RHS) \Downarrow s'$, so $LHS \eveq RHS$.		
	\end{itemize}	
\end{proof}

\section{Further discussion on semantics} \label{ap:sem}
\subsection{Semantics of Generated Quantities} \label{ap:gq_semantics}

In addition to defining random variables to be sampled using HMC, Stan also supports sampling using pseudo-random number generators. For example, a standard normal parameter $x$ can be sampled in two ways:
\begin{enumerate}
	\item By declaring $x$ to be a parameter of the model, and giving it a prior:
	\begin{lstlisting}
	parameters { real x; }
	model { x ~ normal(0, 1); }
	\end{lstlisting}
	
	\item Treating $x$ as a generated quantity and using a pseudo-random number generator:
	\begin{lstlisting}
	generated quantities { x = normal_rng(0, 1); }
	\end{lstlisting}	
\end{enumerate}

Option (1) will sample $x$ using HMC, which is not needed in this case. Option (2) is a much more efficient solution. Thus, a Stan user can explicitly optimise their program by specifying how HMC should be composed with forward (ancestral) sampling.

In the density-based semantics presented in this paper, we do not formalise this usage of pseudo-random number generators. We treat the function \lstinline|normal_rng(mu, sigma)| as any other function, ignoring its random nature. We define the semantics of a Stan program to be the unnormalised log posterior over parameters only --- $\log p^*(\params \mid \data)$. However, this semantics can be extended to cover the generated quantities $\mathbf{g}$ as well: $\log p^*(\params, \mathbf{g} \mid \data)$.

The easiest way to do that is to simply treat \lstinline| normal_rng| as another derived form:
\begin{display}[.5]{}
	\clause{L =  \kw{d_rng}(E_1, \dots E_n) \deq L \sim \kw{d}(E_1, \dots E_n)}{random number generation} 	
\end{display}

However, this causes a discrepancy with the current information-flow type system. Perhaps a more suitable treatment is as an assignment to another reserved variable, which holds a different density to that of \lstinline|target|:  
\begin{display}[.5]{}
	\clause{L =  \kw{d_rng}(E_1, \dots E_n) \deq \kw{gen} = \kw{gen} + \kw{d_lpdf}(L, E_1, \dots E_n)}{random number generation} 	
\end{display}

The density-based semantics of a Stan program can then be defined as:
$$\log p^*(\params, \mathbf{g} \mid \data) = \log p^*(\params \mid \data) + \log p(\mathbf{g} \mid \params, \data)$$
where $\log p^*(\params \mid \data)$, as before, is given by the \lstinline|target| variable, and $\log p(\mathbf{g} \mid \params, \data)$ is given by \lstinline|gen|.
 
An interesting direction for future work is to extend the semantics and type system of SlicStan, so that modelling statements, such as \lstinline|x ~ normal(0, 1)| can be treated either as modifying the \lstinline|target| density, or random number generation, depending on the level of the variable \lstinline|x|. This can allow SlicStan programs to be optimised by automatically determining the most efficient way to compose HMC and forward sampling, based on the concrete model.

\subsection{Relation of Density-based Semantics to Sampling-based Semantics} \label{ap:sampling_semantics}

The density-based semantics of Stan and SlicStan given in this paper is inspired by the sampling-based semantics that \citet{HurNRS15} give to the imperative language \textsc{Prob}.
This section outlines the differences between the two semantics.

\subsubsection{Operational semantics relations }

The intuition behind the density-based semantics of Stan is that the relation $(s, S) \Downarrow s'$ specifies what the value of the (unnormalised) posterior is at a specific point in the parameter space. For $\params \subset s$ and $\data \subset s$, $p^*(\params \mid \data) = s'(\kw{target})$.

The intuition behind the operational semantics relation by \citet{HurNRS15}, $(s, S) \Downarrow^t (s', p)$, is that there is probability $p$ for the program $S$, started in initial state $s$, to sample a trace $t$ and terminate in state $s'$.

For programs with no observed values, and single probabilistic assignment (\lstinline|x ~ d1(...); x ~ d2(...)| is not allowed), we guess that if $(s, S) \Downarrow^t (s', p)$, then $((s\cup t), S) \Downarrow s'[\kw{target} \mapsto p]$, and $\params = t$.

\subsubsection{Difference in the meaning of $\sim$ }

In Stan, a model statement such as \lstinline|x ~ normal(0,1)|, does not denote sampling, but a change to the target density. The value of $x$ remains the same; we only compute the standard normal density at the current value of $x$. This is also similar to the score operator of \citet{Staton2016}.
\begin{display}{Operational Density-based Semantics of Model Statements (Derived Rule)}
	\quad
	\staterule{Eval Model}
	{ (s, E) \Downarrow V \quad (s, E_i) \Downarrow V_i  \quad \forall i \in 1..n \quad V' = s(\kw{target}) + \kw{d_lpdf}(V, V_1, \dots, V_n)}
	{ (s, E \sim \kw{d}(E_1,\dots,E_n)) \Downarrow s[\kw{target} \mapsto V']}
\end{display}

In the sampling-based semantics of \citet{HurNRS15}, on the other hand, \lstinline|x ~ normal(0,1)| is understood as ``we sample a standard normal variable and assign its value to $x$.'' 
\begin{display}{Operational Sampling-based Semantics of Model Statements \cite{HurNRS15}}
	\quad
	\staterule{Sampling Model}
	{ v \in \text{Val} \quad p = \kw{Dist}(s(\overline{\theta}))(x)}
	{ (s, x \sim \kw{Dist}(\overline{\theta})) \Downarrow^{x\mapsto [v]} (s[x \mapsto v], p)}
\end{display}

In this sampling-based semantics, variables can be sampled more than once, and we keep track of the entire trace of samples. In Stan's density-based semantics, modelling a variable more than once would mean modifying the target density more than once. For example, consider the program:
\begin{lstlisting}
	x ~ normal(-5, 1);
	x ~ normal(5, 1);
\end{lstlisting}

The difference between the density-based and sampling-based semantics is then as follows:
\begin{itemize}
	\item \textbf{Density-based:} the program denotes the unnormalised density ${p^*(x) = \normal(x \mid -5, 1)\normal(x \mid 5, 1)}$.
	\item \textbf{Sampling-based:} the program denotes the unnormalised density $p^*(x^{(1)}, x^{(2)}) = \normal(x^{(1)} \mid -5, 1)\normal(x^{(2)} \mid 5, 1)$, with $x^{(1)}$ and $x^{(2)}$ being variables denoting the value of $x$ we sampled the first and second time in the program respectively. 
\end{itemize}

\subsubsection{Difference in the meaning of \kw{observe}}

As mentioned previously, we presume that for \textit{single probabilistic assignment} programs that contain only \textit{unobserved} parameters, out density-based semantics is equivalent to the sampling-based semantics of \citet{HurNRS15}. However the two semantics treat observations differently.

Consider the following SlicStan program, where $y$ is an observed variable:
\begin{lstlisting}
	real mu ~ normal(0, 1);
	data real y ~ normal(mu, 1);
\end{lstlisting}

The density-based semantics of this program is a function of $\mu$: $$p^*(\mu \mid y = v) = \normal(\mu \mid 0, 1)\normal(v \mid \mu, 1)$$ where $v$ is some concrete value for the data $y$, which is supplied externally. 

The corresponding \textsc{Prob} program is:
\begin{lstlisting}
	double mu ~ normal(0, 1);
	double y ~ normal(mu, 1);
	observe(y = v);
\end{lstlisting}

The data $v$ is encoded in the program, and the sampling-based semantics is a function of $\mu$ and $y$: $$p^*(\mu, y) = \begin{cases}
\normal(\mu \mid 0, 1)\normal(y \mid \mu, 1), & \text{if}\ y=v \\
0, & \text{otherwise}
\end{cases}$$

Intuitively, the operational sampling-based semantics defines how to sample the variables $\mu$ and $y$, and then reject the run if $y \neq v$. This introduces a zero-probability conditioning problem when working with continuous variables, and fails in practise.  

The operational density-based semantics of this paper puts SlicStan's observations closer to the idea of soft constraints. Using the \textit{score} operator of \citet{Staton2016}, we can write:
\begin{lstlisting}
	let x ~ normal(0,1) in
	score(density_normal(y,(x, 1));
\end{lstlisting}

Once again, $y$ has some concrete value $v$, and the score operator calculates the density of $y$ at $v$. \citet{Staton2016} make the score operator part of their metalanguage, while we build it into the density-based semantics itself.

\section{Elaborating and shredding if or for statements} \label{ap:shred}
This section demonstrates with an example the elaboration and  shredding of \lstinline|if| and \lstinline|for| statements. 

Consider the following SlicStan program:
\begin{lstlisting}
	data real x;
	real +++data+++ d;
	real +++model+++ m;
		
	if(x > 0){
		d = 2 * x;
		m ~ normal(d, 1);
	}	   
\end{lstlisting}

The body of the \lstinline|if| statement contains an assignment to a \lev{data} variable (\lstinline|d = 2 * x|), and a model statement (\lstinline|m ~ normal(d, 1)|). The former belongs to the \lstinline|transformed data| block of a Stan program, while the latter belongs to the \lstinline|model| block. We need to shred the entire body of the \lstinline|if| statement, into several \lstinline|if| statements, each of which contains statements of a single level only. 

Firstly, the elaboration step ensures that the guard of each \lstinline|if| statement (and the bounds of each \lstinline|for| loop) is a fresh boolean variable, \lstinline|g|, which is not modified anywhere in the program:
\begin{lstlisting}
	g = (x > 0)
	if(g){
		d = 2 * x;
		m ~ normal(d, 1);
	}	   
\end{lstlisting}

Then the shredding step can copy the \lstinline|if| statement at each level, without changing the meaning of the original program:
\begin{lstlisting}
	$S_D =~\,$g = (x > 0)
		    if(g){ d = 2 * x; }	 
	
	$S_M =~$if(g){ m ~ normal(d, 1); }	   
	
	$S_Q =~$skip
\end{lstlisting}

Finally, this translates to the Stan program:
\begin{lstlisting}
	data { 
		real x; 
	}
	transformed data { 
		real d;
		bool g = (x > 0); 
		if(g){ d = 2 * x; }
	}
	parameters { 
		real m; 
	}
	model {
		if(g){ m ~ normal(d, 1); }
	}
\end{lstlisting}


\section{Non-centred Reparameterisation} \label{ap:ncp}

\emph{Reparameterising} a model means expressing it in terms of different parameters, so that the original parameters can be recovered from the new ones. Reparametrisation plays a key role in optimising some models for MCMC inference, as it could transform a posterior distribution that is difficult to sample from in practice, into a flatter, easier to sample from distribution. 

To show the importance of one such reparameterisation technique, the non-centred reparameterisation, consider the pathological \emph{Neal's Funnel} example, which was chosen by \cite{Funnel} to demonstrate the difficulties Metropolis--Hastings runs into when sampling from a distribution with strong non-linear dependencies. The model defines a density over variables $x$ and $y$,\footnotemark such that: 
\footnotetext{For simplicity, we consider a 2-dimensional version of the funnel, as opposed to the original 10-dimensional version.}
$$y \sim \normal(0,3) \qquad\qquad x \sim \normal(0, e^{\frac{y}{2}})$$
The density has the form of a funnel (thus the name \emph{``Neal's Funnel''}), with a very sharp neck, as shown in \autoref{fig:neals}. We can easily implement the model in a straightforward way (centred parameterisation), as on the left below. 
\vspace{-8pt}\begin{multicols}{2} 
\textbf{Centred parameterisation in Stan}
\vspace{1cm}
\begin{lstlisting}[basicstyle=\small]
	parameters {
		real y;
		real x;
	}
	model {	
		y ~ normal(0, 3);		
		x ~ normal(0, exp(y/2));
	}
\end{lstlisting}
\vspace{4cm}
\textbf{Non-centred parameterisation in Stan}
\begin{lstlisting}[basicstyle=\small]
	parameters {
		real y_std;
		real x_std;
	}
	transformed parameters {
		real y = 3.0 * y_std;
		real x = exp(y/2) * x_std;
	}
	model {
		y_std ~ normal(0, 1); // \!implies y $\sim \normal(0, 3)$
		x_std ~ normal(0, 1); // \!implies x $\sim \normal(0, e^{(y/2)})$
	}
\end{lstlisting}
\end{multicols}\vspace{-8pt}
However, in that case, Stan's sampler has trouble obtaining samples from the neck of the funnel, because there exists a strong non-linear dependency between $x$ and $y$, and the posterior geometry is difficult for the sampler to explore well (see \autoref{fig:nealsineff}).
	
\begin{figure}[!b]
\centering
\begin{subfigure}[b]{0.49\textwidth}
\includegraphics[width=\textwidth]{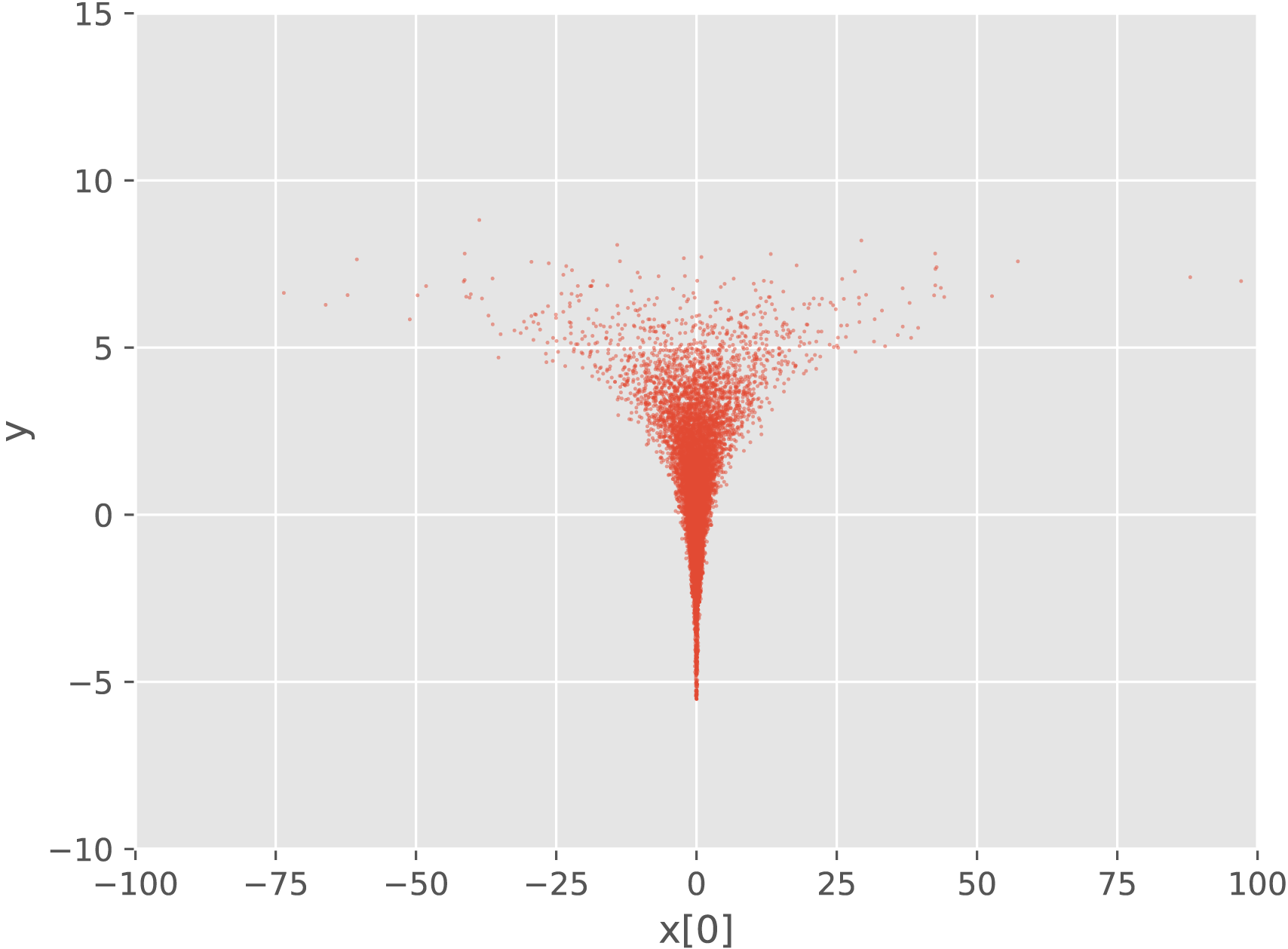}
\caption{$24,000$ samples obtained using Stan (default settings) for the \emph{non-efficient} form of Neal's Funnel.}
\label{fig:nealsineff}
\end{subfigure}
\begin{subfigure}[b]{0.49\textwidth}
\includegraphics[width=\textwidth]{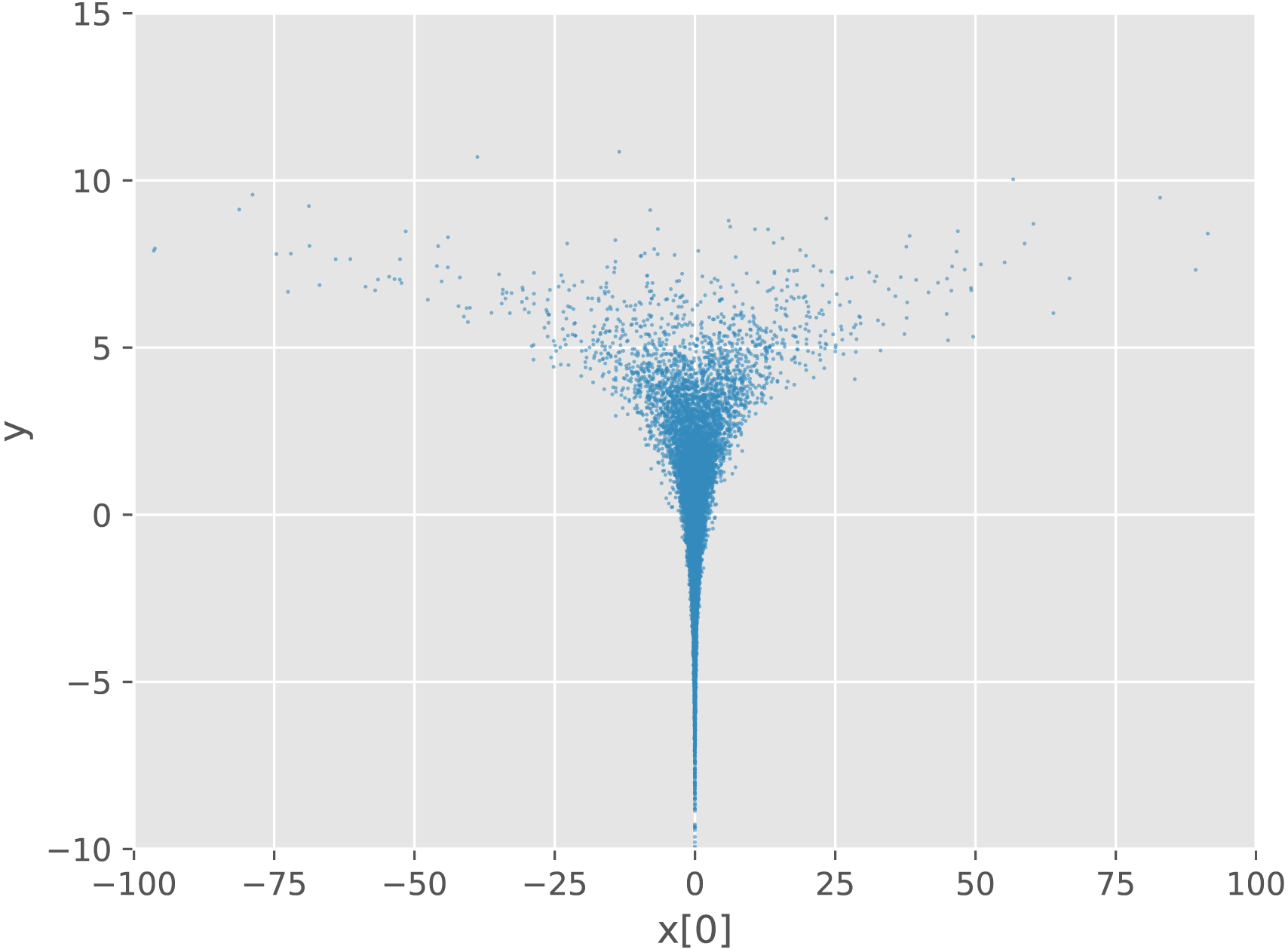}
\caption{$24,000$ samples obtained using Stan (default settings) for the \emph{efficient} form of Neal's Funnel.}
\label{fig:nealseff}
\end{subfigure}	
\vspace{8pt} \\
\begin{subfigure}[b]{0.7\textwidth}
\includegraphics[width=\textwidth]{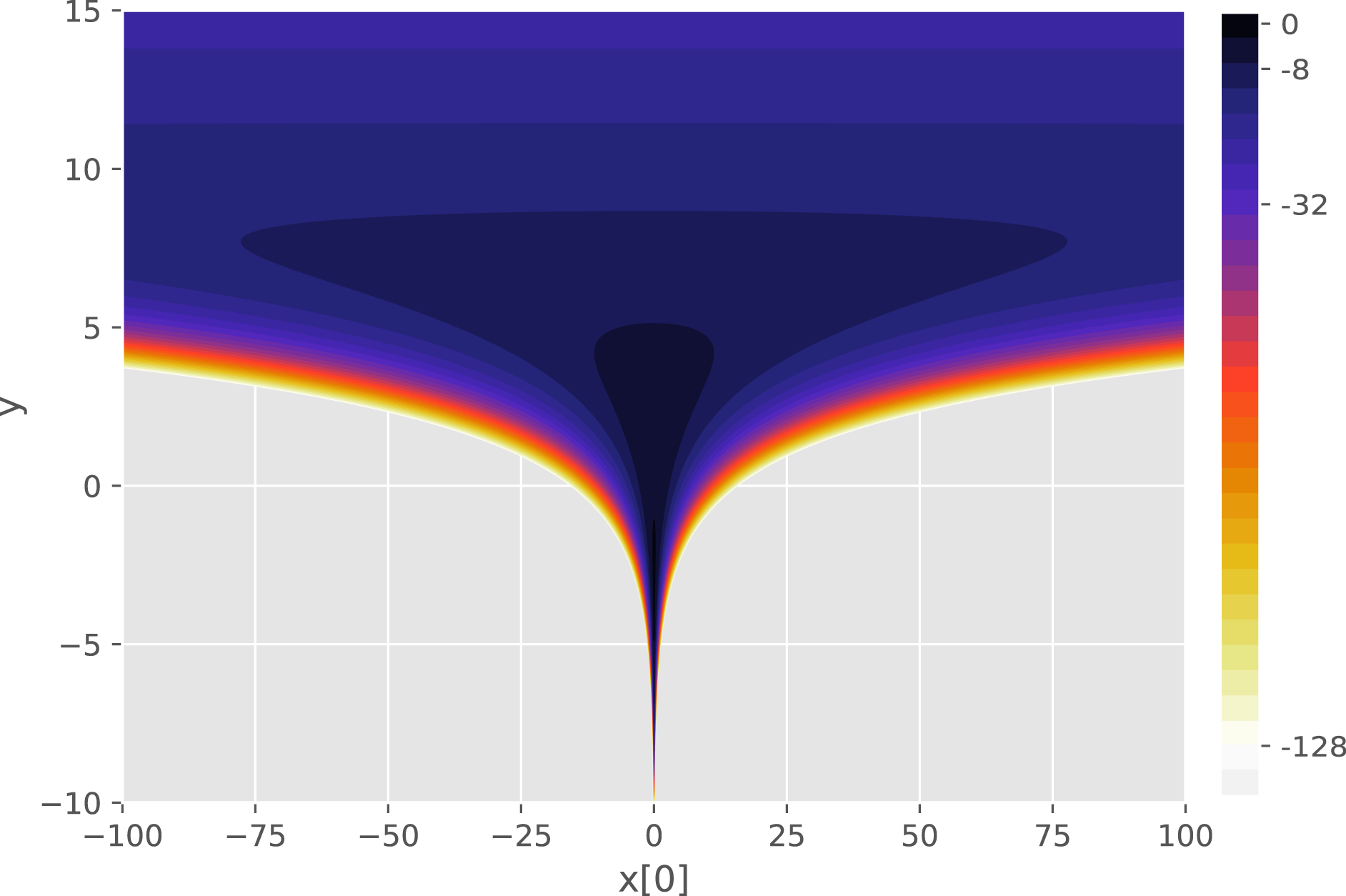}
\caption{The actual log density of Neal's Funnel. Dark regions are of high density (log density greater than $-8$).}
\label{fig:neals} 
\end{subfigure}
\caption{}	
\end{figure}

Alternatively, we can reparameterise the model, so that the model parameters are changed from $x$ and $y$ to the standard normal parameters $x^{(std)}$ and $y^{(std)}$, and the original parameters are recovered using shifting and scaling:
$$y^{(std)} \sim \normal(0, 1) \qquad\qquad x^{(std)} \sim \normal(0, 1) \qquad\qquad y = y^{(std)} \times 3 \qquad\qquad x =  x^{(std)} \times e^{\frac{y}{2}}$$
This reparameterisation, called non-centred parametrisation, is a special case of a more general transform introduced by \citet{Reparam}.

As shown on the right above, the non-centred model is longer and less readable. However, it performs much better than the centred one. \autoref{fig:nealseff} shows that by reparameterising the model, we are able to explore the tails of density better than if we use the straightforward implementation.  

Neal's Funnel is a typical example of the dependencies that priors in hierarchical models could have. The example demonstrates that in some cases, especially when there is little data available, using non-centred parameterisation could be vital to the performance of the inference algorithm. The centred to non-centred parameterisation transformation is therefore common to Stan models, and is extensively described by the \citet{StanManual} as a useful technique.

\section{Examples} \label{ap:examples}

\subsection{Neal's Funnel}
We continue with the Neal's funnel example from Appendix~\ref{ap:ncp}, to demonstrate the usage of user-defined functions in SlicStan. 

Reparameterising a model, which involves a (centred) Gaussian variable $x \sim \normal(\mu, \sigma)$ involves introducing a new parameter $x^{(std)}$. Therefore, a non-centred reparameterisation function cannot be defined in Stan, as Stan user-defined functions cannot declare new parameters. In SlicStan, on the other hand, reparameterising the model to a non-centred parameterisation can be implemented by simply calling the function \lstinline|my_normal|.

Below is the Neal's funnel in SlicStan (left) and Stan (right).

\begin{minipage}[t][6cm][t]{\linewidth} 
	\vspace{1pt}
	\begin{multicols}{2}
		\centering
		\textbf{``Neal's Funnel'' in SlicStan}
		\vspace{1cm}
		\begin{lstlisting}[numbers=left,numbersep=\numbdist,numberstyle=\tiny\color{\numbcolor},basicstyle=\small]
	real my_normal(real m, real s) {
		real std ~ normal(0, 1); 
		return s * std + m;
	}
	real y = my_normal(0, 3); 
	real x = my_normal(0, exp(y/2));
	\end{lstlisting}
		\vspace{4cm}
		\textbf{``Neal's Funnel'' in Stan}
		\begin{lstlisting}[numbers=left,numbersep=\numbdist,numberstyle=\tiny\color{\numbcolor},basicstyle=\small]
	parameters {
		real y_std;
		real x_std;
	}
	transformed parameters {
		real y = 3.0 * y_std;
		real x = exp(y/2) * x_std;
	}
	model {
		y_std ~ normal(0, 1); 
		x_std ~ normal(0, 1); 
	}
		\end{lstlisting}
	\end{multicols} 
\end{minipage} 

The non-centred SlicStan program (left) is only longer than its centred version, due to the presence of the definition of the function. In comparison, Stan requires defining the new parameters \lstinline|x_std| and \lstinline|y_std| (lines 2,3), moving the declarations of \lstinline|x| and \lstinline|y| to the \lstinline|transformed parameters| block (lines 6,7), defining them in terms of the parameters (lines 8,9), and changing the definition of the joint density accordingly (lines 12,13). 

We also present the ``translated'' Neal's Funnel model, as it would be outputted by an implemented compiler. We notice a major difference between the two Stan programs --- in one case the variables of interest $x$ and $y$ are defined in the \lstinline|transformed parameters| block, while in the other they are defined in the \lstinline|generated quantities| block. In an intuitive, centred parameterisation of this model, $x$ and $y$ are in fact the parameters. Therefore, it is much more natural to think of those variables as transformed parameters when using a non-centred parameterisation. However, as shown in \autoref{tab:blocks}, variables declared in the \lstinline|transformed parameters| block are re-evaluated at every leapfrog, while those declared in the \lstinline|generated quantities| block are re-evaluated at every sample. This means that even though it is more intuitive to think of $x$ and $y$ as transformed parameters (original Stan program), declaring them as generated quantities where possible results in a better optimised inference algorithm in the general case. 

There are some advantages of the original Stan code that the translated from SlicStan Stan code does not have. 
The original version is considerably shorter than the translated one. This is due to the lack of the additional variables \lstinline|m|, \lstinline|mp|, \lstinline|s|, \lstinline|sp|, \lstinline|ret|, and \lstinline|retp|, which are a consequence of statically unrolling the function calls in the elaboration step. When using SlicStan, the produced Stan program acts as an intermediate representation of the probabilistic program, meaning that the reduced readability of the translation is not necessarily problematic. However, the presence of the additional variables can also, in some cases, lead to slower inference. This problem can be tackled by introducing standard optimising compilers techniques, such as variable and common subexpression elimination.

\begin{minipage}[t][14cm][t]{\linewidth}
\vspace{12pt}
\textbf{``Neal's Funnel'' translated to Stan} 
\vspace{6pt}
\begin{lstlisting}[numbers=left,numbersep=\numbdist,numberstyle=\tiny\color{\numbcolor}]
	transformed data {
		real m;
		real mp;
		m = 0;
		mp = 0;
	}	
	parameters {
		real x_std;
		real x_stdp;
	} 
	model{
		x_std ~ normal(0, 1);
		xr_stdp ~ normal(0, 1);
	}	
	generated quantities {
		real s;
		real sp;
		real ret;
		real retp;
		real x;
		real y;
		s = 3;
		ret = s * x_std + m;
		y = ret;
		sp = exp(y * 0.5);
		retp = (sp * x_stdp + mp);
		x = retp;
	}
\end{lstlisting}
\end{minipage}

Moreover, we notice the names of the new parameters in the translated code: \lstinline|x_std| and \lstinline|x_stdp|. The names are important, as they are part of the output of the sampling algorithm. Unlike Stan, with the user-defined-function version of Neal's funnel, in SlicStan the programmer does not have control on the names of the newly introduced parameters. One can argue that the user was not interested in those parameters in the first place (as they are solely used to reparameterise the model for more efficient inference), so it does not matter that their names are not descriptive. However, if the user wants to debug their model, the output from the original Stan model would be more useful than that of the translated one.

\newpage
\subsection{Cockroaches} \label{ssec:roaches}

The ``Cockroaches'' example is described by \citet[p.~161]{GelmanHill}, and it concerns measuring the effects of integrated pest management on reducing cockroach numbers in apartment blocks. They use \emph{Poisson regression} to model the number of caught cockroaches $y_i$ in a single apartment $i$, with exposure $u_i$ (the number of days that the apartment had cockroach traps in it), and regression predictors:
\begin{itemize}
	\item the pre-treatment cockroach level $r_i$;
	\item whether the apartment is in a senior building (restricted to the elderly), $s_i$; and
	\item the treatment indicator $t_i$.
\end{itemize}

In other words, with $\beta_0, \beta_1, \beta_2, \beta_3$ being the regression parameters, we have:
$$y_i \sim Poisson(u_i \exp(\beta_0 + \beta_1r_i + \beta_2s_i + \beta_3t_i))$$

After specifying their model this way, Gelman and Hill simulate a replicated dataset $\mathbf{y}_{rep}$, and compare it to the actual data $\mathbf{y}$ to find that the variance of the simulated dataset is much lower than that of the real dataset. In statistics, this is called \emph{overdispersion}, and is often encountered when fitting models based on a single parameter distributions,\footnotemark~such as the Poisson distribution. A better model for this data would be one that includes an \emph{overdispersion parameter} $\boldsymbol{\lambda}$ that can account for the greater variance in the data. 

\footnotetext{In a distribution specified by a single parameter $\alpha$, the mean and variance both depend on $\alpha$, and are thus not independent.}

The next page shows the ``Cockroach'' example before (ignoring the red lines) and after (assuming the red lines) adding the overdispersion parameter, in both SlicStan (left) and Stan (right). Similarly to before, SlicStan gives us more flexibility as to where the statements accounting for overdispersion can be added.
Stan, on the other hand, introduces an entirely new to this program block --- \lstinline|transformed parameters|.

\newpage
\begin{minipage}[t][14cm][t]{\linewidth}
	\vspace{1pt}
\begin{multicols}{2} \label{roaches}
	\centering
	\textbf{``Cockroaches'' in SlicStan}
	\vspace{1cm}
	\begin{lstlisting}[numbers=left,numbersep=\numbdist,numberstyle=\tiny\color{\numbcolor}]
	data int N; 
	data real[N] exposure2; 
	data real[N] roach1; 
	data real[N] senior; 
	data real[N] treatment; 
	
	real[N] log_expo$\,$=$\,$log(exposure2);
	
	real[4] beta;
	
	%%real tau ~ gamma(0.001, 0.001);%%
	%%real sigma = 1.0 / sqrt(tau);%%
	%%real[N] lambda$\,$~$\!\!$normal(0, sigma);%%
	
	data int[N] y
		~ poisson_log$\footnotemark$(log_expo + beta[1] 
									$\,$+ beta[2] * roach1 
									$\,$+ beta[3] * treatment 
									$\,$+ beta[4] * senior
									%%$\,$+ lambda%%);
	\end{lstlisting}
	\vspace{4cm}
	\textbf{``Cockroaches'' in Stan}
	\begin{lstlisting}[numbers=left,numbersep=\numbdist,numberstyle=\tiny\color{\numbcolor}]
	data {
		int N; 
		real[N] exposure2;
		real[N] roach1; 
		real[N] senior; 
		real[N] treatment; 
		int y[N];
	}
	transformed data {
		real[N] log_expo$\,$=$\,$log(exposure2);
	}
	parameters {
		real[4] beta;
		%%real[N] lambda;%%
		%%real tau;%%
	} 
	%%transformed parameters {%%
		%%real sigma = 1.0 / sqrt(tau);%%
	%%}%%
	model {
		%%tau ~ gamma(0.001, 0.001);%%
		%%lambda ~ normal(0, sigma);%%
		y ~ poisson_log$^{\ref{note2}}$(log_expo + beta[1] 
										$\,$+ beta[2] * roach1 
										$\,$+ beta[3] * treatment 
										$\,$+ beta[4] * senior
										%%$\,$+ lambda%%);
	}
	\end{lstlisting}
\end{multicols} \vspace{-18pt}
\begin{center}Example adapted from \url{https://github.com/stan-dev/example-models/}.\end{center}	
\end{minipage}

\footnotetext{\label{note2}Stan's \lstinline|poisson_log| is a numerically stable way to model a Poisson variable where the event rate is $e^{\alpha}$ for some $\alpha$.}

\newpage
\subsection{Seeds}
Next, we take the ``Seeds'' example introduced by \citet[p.~300]{BUGSBook} in ``\emph{The BUGS Book}''. In this example, we have $I$ plates, with plate $i$ having a total of $N_i$ seeds on it, $n_i$ of which have germinated. Moreover, each plate $i$ has one of 2 types of seeds $x_1^{(i)}$, and one of 2 types of root extract $x_2^{(i)}$. We are interested in modelling the number of germinated seeds based on the type of seed and root extract, which we do in two steps. Firstly, we model the number of germinated seeds with a Binomial distribution, whose success probability is the probability of a single seed germinating:
$$n_i \sim Binomial(N, p_i)$$

We model the probability of a single seed on plate $i$ germinating as the output of a logistic regression with input variables the type of seed and root extract:
$$p_i = \sigma(\alpha_0 + \alpha_1x_1^{(i)} + \alpha_2x_2^{(i)} + \alpha_{12}x_1^{(i)}x_2^{(i)} + \beta^{(i)})$$

In the above, $\alpha_0, \alpha_1, \alpha_2, \alpha_{12}$ and $\beta^{(i)}$ are parameters of the model, with $\beta^{(i)}$ allowing for \emph{over-dispersion} (see \autoref{ssec:roaches}).

The next page shows the ``Seeds'' model written in SlicStan (left) and in Stan (right). The Stan code was adapted from the example models listed on Stan's GitHub page. 

As before, we see that SlicStan's code is shorter than that of Stan. It also allows for more flexibility in the order of declarations and definitions, making it possible to keep related statements together (e.g. lines 14 and 15 of the example written in SlicStan). Once again, SlicStan provides more abstraction, as the programmer does not have to specify how each variable of the model should be treated by the underlying inference algorithm. Instead it automatically determines this when it translates the program to Stan. 

\newpage
\begin{minipage}[t][14cm][t]{\linewidth}
\begin{multicols}{2} \label{seeds}
	\centering
	\textbf{``Seeds'' in SlicStan}
	\vspace{3cm}
	\begin{lstlisting}[numbers=left,numbersep=\numbdist,numberstyle=\tiny\color{\numbcolor}]
	data int I;
	data int[I] n; 
	data int[I] N; 
	data real[I] x1; 
	data real[I] x2;
	
	real[I] x1x2 = x1 .* x2; 
	
	real alpha0 ~ normal(0.0,1000);
	real alpha1 ~ normal(0.0,1000);
	real alpha2 ~ normal(0.0,1000);
	real alpha12 ~ normal(0.0,1000);
	
	real tau ~ gamma(0.001,0.001);
	real sigma = 1.0 / sqrt(tau);
	
	real[I] b ~ normal(0.0, sigma);
	n ~ binomial_logit$\footnotemark$(N, alpha0 
			$\qquad\qquad$ + alpha1 * x1 
			$\qquad\qquad$ + alpha2 * x2 
			$\qquad\qquad$ + alpha12 * x1x2 
			$\qquad\qquad$ + b);
	\end{lstlisting}
	
	\vspace{5cm}
	\textbf{``Seeds'' in Stan}
	\begin{lstlisting}[numbers=left,numbersep=\numbdist,numberstyle=\tiny\color{\numbcolor}]
	data {
		int I;
		int n[I];
		int N[I];
		real[I] x1; 
		real[I] x2; 
	}
	
	transformed data {
		real[I] x1x2 = x1 .* x2;
	} 
	parameters {
		real alpha0;
		real alpha1;
		real alpha12;
		real alpha2;
		real tau;
		real[I] b;
	}
	transformed parameters {
		real sigma = 1.0 / sqrt(tau);
	}
	model {
		alpha0 ~ normal(0.0,1000);
		alpha1 ~ normal(0.0,1000);
		alpha2 ~ normal(0.0,1000);
		alpha12 ~ normal(0.0,1000);
		tau ~ gamma(0.001,0.001);
	
		b ~ normal(0.0, sigma);
		n ~ binomial_logit$^{\ref{note1}}$(N, alpha0 
				$\qquad\qquad$ + alpha1 * x1 
				$\qquad\qquad$ + alpha2 * x2 
				$\qquad\qquad$ + alpha12 * x1x2 
				$\qquad\qquad$ + b);
	}
	\end{lstlisting}
\end{multicols}
\begin{center}Example adapted from \url{https://github.com/stan-dev/example-models/}.\end{center}	
\end{minipage}

\footnotetext{\label{note1}Stan's \lstinline|binomial_logit| distribution is a numerically stable way to use a logistic sigmoid in combination with a Binomial distribution.} } { }

\end{document}